\pdfoutput=1
\documentclass[10pt]{article}
\usepackage[left=0.75in,right=0.75in,top=0.75in,bottom=0.75in]{geometry}
\usepackage[utf8]{inputenc}
\usepackage{times}
\usepackage{amssymb,amsmath,latexsym,amsthm,amsfonts,mathtools}

\usepackage{mathrsfs}
\usepackage{multirow}
\usepackage{graphicx}
\usepackage{color}
\usepackage{textcomp,siunitx}
\usepackage{float}
\usepackage{subcaption}
\usepackage{booktabs}
\usepackage{enumitem}

\usepackage[hyperref]{xcolor}
\definecolor{ao(english)}{rgb}{0.0, 0.5, 0.0}
\usepackage{hyperref}
\hypersetup{colorlinks, breaklinks, citecolor=blue, linkcolor=ao(english), urlcolor=blue}

\usepackage[nameinlink,noabbrev,sort&compress]{cleveref}
\crefname{assumption}{Assumption}{Assumptions}
\usepackage{setspace}
\usepackage{cite}

\theoremstyle{plain}
\newtheorem{theorem}{Theorem}
\newtheorem{lemma}{Lemma}

\theoremstyle{remark}
\newtheorem{remark}{Remark}

\theoremstyle{definition}
\newtheorem{definition}{Definition}

\newtheorem{problem}{Problem}
\newcommand{\sign}{\mathrm{sign}}

\colorlet{velcol}{blue!80!black}
\colorlet{acccol}{red!80!black}
\colorlet{rotcol}{purple!80!black}
\colorlet{uavbody}{black!85}
\colorlet{uavrotor}{gray!15}
\colorlet{axiscol}{gray!80}
\colorlet{loscol}{black!80}
\colorlet{landcol}{orange!90!black}
\allowdisplaybreaks

\usepackage{scalerel}
\usepackage{tikz}
\usetikzlibrary{automata, shapes, arrows, calc, arrows.meta, fit, positioning, shadows,angles}
\usetikzlibrary{svg.path}
\definecolor{orcidlogocol}{HTML}{A6CE39}
\tikzset{
	orcidlogo/.pic={
		\fill[orcidlogocol] svg{M256,128c0,70.7-57.3,128-128,128C57.3,256,0,198.7,0,128C0,57.3,57.3,0,128,0C198.7,0,256,57.3,256,128z};
		\fill[white] svg{M86.3,186.2H70.9V79.1h15.4v48.4V186.2z}
		svg{M108.9,79.1h41.6c39.6,0,57,28.3,57,53.6c0,27.5-21.5,53.6-56.8,53.6h-41.8V79.1z M124.3,172.4h24.5c34.9,0,42.9-26.5,42.9-39.7c0-21.5-13.7-39.7-43.7-39.7h-23.7V172.4z}
		svg{M88.7,56.8c0,5.5-4.5,10.1-10.1,10.1c-5.6,0-10.1-4.6-10.1-10.1c0-5.6,4.5-10.1,10.1-10.1C84.2,46.7,88.7,51.3,88.7,56.8z};
	}
}
\newcommand\orcidicon[1]{\href{https://orcid.org/#1}{\mbox{\scalerel*{
				\begin{tikzpicture}[yscale=-1,transform shape]
					\pic{orcidlogo};
				\end{tikzpicture}
			}{|}}}}
		
\title{Equivalent-Agent Guidance for Cooperative UAV Payload Transportation}
\author{Saurabh~Kumar\textsuperscript{\orcidicon{0000-0002-8344-5966}}
\thanks{Saurabh Kumar and Shashi Ranjan Kumar are with the Intelligent Systems \& Control Lab, Department of Aerospace Engineering, Indian Institute of Technology Bombay, Powai -- 400076, Mumbai, India. \newline Abhinav Sinha is with the Guidance, Autonomy, Learning, and Control for Intelligent Systems (GALACxIS) Lab, Department of Aerospace Engineering, University of Cincinnati, Cincinnati, OH, 45221,  USA. \newline
E-mails: saurabh.k@aero.iitb.ac.in, srk@aero.iitb.ac.in, abhinav.sinha@uc.edu.}	
\and Shashi~Ranjan~Kumar\textsuperscript{\orcidicon{0000-0001-6446-7281}}
\and Abhinav~Sinha\textsuperscript{\orcidicon{0000-0001-6419-2353}} 
}	
\date{}
\begin{document}
\maketitle
\onehalfspacing
\begin{abstract}
	This paper develops a guidance framework for cooperative transportation of a rigid payload by two uncrewed aerial vehicles (UAVs) to stationary and maneuvering landing platforms. A virtual equivalent-agent representation is first introduced to describe the translational motion of the rigidly coupled UAV–payload system, allowing the transportation problem to be formulated in terms of relative range and line-of-sight dynamics with respect to the landing platform. A geometric analysis establishes the terminal feasibility conditions for payload delivery. In particular, an arbitrary prescribed approach angle can be achieved for a stationary platform, whereas successful delivery to a maneuvering platform with zero relative velocity requires terminal velocity and heading angle synchronization and consequently a zero landing angle. Leveraging this framework, a robust fixed-time sliding mode guidance strategy is developed to regulate both relative range and line-of-sight dynamics. A separate link-orientation controller and control allocation scheme is presented to map virtual equivalent agent commands to the individual UAV's control inputs. The proposed strategy guarantees convergence to the desired landing configuration within a uniformly bounded time, independent of initial engagement geometries, while explicitly accommodating uncertainties arising from target maneuvers.  Numerical simulations demonstrate accurate delivery under different terminal approach angles and platform maneuvers, while processor-in-the-loop implementation on a Raspberry Pi demonstrates that the guidance algorithm satisfies the real-time computational requirements. Nonetheless, a comparative analysis shows that the proposed framework achieves better tracking accuracy and faster sliding surface convergence while requiring significantly less control energy from each UAV.
	
\end{abstract}
	
\section{Introduction}\label{sec:introduction}
Load transportation refers to the task of delivering a payload from a given initial location to a desired destination and constitutes a fundamental problem in autonomous systems mobility \cite{doi:10.2514/1.G006608,doi:10.1016/j.ast.2024.109225}. The rapid advancement of autonomous vehicles has significantly transformed payload transportation by enabling precise, reliable, and repeatable operations even under off-nominal conditions. A wide range of autonomous platforms, including uncrewed ground vehicles (UGVs), uncrewed surface vessels (USVs), and uncrewed aerial vehicles (UAVs), have been utilized for payload transportation, each offering unique advantages that cater to specific applications. While early developments in autonomous load transportation primarily focused on ground-based platforms, recent research has emphasized UAV-based solutions due to their superior maneuverability, accessibility, and operational flexibility. 

Despite these advantages, the payload transportation capability of a single UAV is inherently limited by its lifting capacity, which is determined by actuator power and vehicle size. Increasing the payload beyond this limit degrades flight endurance, efficiency, and safety margins. An immediate way to increase the UAV's payload capacity is to increase its size, which will increase lift and, consequently, its payload capacity. While this concept is effective in the case of ground vehicles, it is less favorable for UAVs due to increased cost, complexity, and reduced agility. As an alternative, cooperative load transportation, where multiple smaller UAVs collaboratively transport a payload, has emerged as a promising solution. Compared to a single large UAV, teams of smaller UAVs offer advantages such as modularity, redundancy, lower cost, and ease of deployment. However, cooperative load transportation introduces significant control challenges arising from strongly coupled nonlinear dynamics \cite{doi:10.1016/j.ast.2025.110042,doi:10.1016/j.ast.2024.108873}, underactuation \cite{doi:10.1016/j.ast.2025.110713,doi:10.1016/j.ast.2021.107139}, and interaction forces between the UAVs and the payload \cite{doi:10.1016/j.ast.2023.108201}. When the payload is suspended \cite{doi:10.1016/j.ast.2017.07.028,doi:10.1016/j.ast.2020.105770} or rigidly attached \cite{doi:10.2514/1.G005098}, its motion introduces additional degrees of freedom, often leading to oscillations, internal force redistribution, and degradation of tracking performance. Consequently, the design of robust and efficient control strategies for cooperative load transportation remains an active area of research.

A wide variety of control architectures have been proposed to address cooperative load transportation, including centralized and decentralized strategies, formation-based approaches, and nonlinear feedback control methods. These architectures aim to ensure coordinated motion of the UAV team while simultaneously regulating payload position and suppressing oscillations. Centralized approaches typically rely on a coordinating unit that computes control inputs for all vehicles and communicates them to the individual agents \cite{doi:10.1007/s10514-010-9205-0}. In contrast, decentralized frameworks aim to improve scalability and robustness by allowing each agent to compute its control input locally while pursuing a shared objective. One such approach was discussed in \cite{doi:10.1007/978-3-642-32723-0_39}, where individual control actions of the agents were calculated based on the grasping point on the payload. From a formation control perspective, decentralized controllers have been developed in \cite{5504175} to regulate internal forces and ensure velocity convergence of UAV teams carrying payloads. Another decentralized paradigm, utilizing concepts from continuum mechanics, known as \emph{cooperative payload lift and manipulation (CALM)}, is discussed in \cite{doi:10.1016/j.ast.2018.09.005}. Adaptive and learning-based decentralized strategies, see, for example \cite{doi:10.1016/j.ast.2024.108960}, have also been proposed to handle uncertainties such as unknown or time-varying payload mass distributions.

Leader–follower and master–slave paradigms constitute another important class of cooperative transportation strategies. The work in \cite{7989678} discussed a master-slave paradigm in which a master UAV, controlled remotely, lifts the payload and moves it in the desired direction, and the slave UAV, which is also attached to the payload, provides compliance to the master's movement by sensing the force applied by the master UAV using an admittance controller. A leader-follower strategy in which the payload serves as the leader whose desired trajectory dictates force commands for the follower vehicles is discussed in \cite{doi:10.2514/1.G004680}. UAV-based manipulation has also been extensively studied in the context of construction and assembly tasks. Early efforts include the use of quadrotors to assemble cubic structures using magnetic connections \cite{doi:10.1007/978-3-642-36279-8_13} and architectures for building foam towers using teams of UAVs \cite{6853477}. Other approaches employ ropes, cables, and rigid components to construct tensile structures \cite{6696853}. Closely related application domains include mid-air refueling and rendezvous problems, in which UAVs are required to track, synchronize with, and maintain relative motion with moving platforms \cite{8089403,doi:10.1016/j.ast.2020.105756}. In \cite{8089403}, an optimal control framework is employed to generate a feasible trajectory for a fixed-wing UAV that must rendezvous with a moving ground vehicle within a prescribed time interval. Similarly, the work in \cite{doi:10.1016/j.ast.2020.105756} formulates an optimal control problem by linearizing the UAV’s nonlinear dynamics and accommodating control constraints through inequality-based formulations. From a different perspective, sliding-mode-based guidance strategies have also been investigated for mid-air refueling scenarios. For instance, the authors in \cite{5530997} propose a pure-pursuit sliding mode guidance law that enables a UAV to robustly follow a maneuvering tanker aircraft. Additionally, capture problems involving cooperative UAV teams have been studied. In particular, the work in \cite{doi:10.2514/1.G004626} presents a collision-cone-based guidance strategy for coordinating multiple UAVs equipped with a net to capture intruding targets.

The authors in \cite{6696850} studied the dynamics of a deformable ring supported by multiple UAVs that were rigidly attached around its perimeter, focusing on controlling the ring’s deformation through cooperative vehicle behavior. The work in \cite{doi:10.1002/rob.20401} presented experimental results on payload lifting using ropes and a team of helicopters, with an emphasis on controlling the helicopters’ motion and attitude while accounting for the additional dynamics introduced by flexible tethers. A related problem was addressed in \cite{doi:10.1007/s10514-010-9205-0}, where three aerial drones manipulated the position and orientation of a disc suspended from cables. The works \cite{7989609,8460529} developed control strategies for manipulating a rigid bar suspended by cables from two UAVs. In particular, \cite{7989609} proposed a vision-based control scheme within a leader–follower framework, while \cite{8460529} focused on stabilizing the rigid bar when the UAVs were heterogeneous and the cable lengths were unequal. The authors in \cite{doi:10.2514/6.2019-3270} introduced a leader–follower formation control approach for two UAVs transporting a suspended payload via flexible beams and designed optimal feedback controllers that compensated for payload-induced effects. The work in \cite{doi:10.1007/s10846-019-01048-4} examined cooperative load transport by multiple quadrotor-type UAVs connected to a rigid payload and proposed an adaptive control method incorporating estimation and consensus to ensure equal load sharing. Similarly, the work in \cite{doi:10.15607/RSS.2013.IX.011} investigated cooperative payload transport by multiple quadrotor UAVs using cables, modeled the payload as both a point mass and a three-dimensional rigid body, and focused on deriving a differentially flat hybrid system to facilitate controller design. A distributed control strategy for a team of UAVs connected to the payload via ropes is discussed in \cite{doi:10.1016/j.ast.2018.10.027}. In \cite{doi:10.1016/j.ast.2020.106284}, a load transportation strategy is proposed for a 2-UAV system, under the assumption that the length of the payload is greater than the horizontal separation between the UAVs. Another strategy, wherein two hovercraft jointly carry a payload, was discussed in \cite{doi:10.1016/j.ast.2025.110713}. The work in \cite{doi:10.1016/j.ast.2024.109078} addressed the load transportation problem for tethered UAVs using the concept of differential graphical games. Note that in all of these studies, the UAVs remained continuously connected to the payload.

In contrast to the previously described scenario, studies had also considered UAVs that were initially free-flying and later grasped a payload \cite{doi:10.1007/s10846-012-9743-0,doi:10.1002/rob.21816,8098676}. The work \cite{doi:10.1007/s10846-012-9743-0} presented the design of sliding-mode-based motion controllers for four multirotor UAVs that simultaneously reached predefined contact points on an object, grasped it, and cooperatively manipulated its position along simple trajectories. They emphasized the importance of modeling the multi-UAV system as a single entity interacting with the payload. The authors in \cite{doi:10.1002/rob.21816} employed Dubins-based algorithms to generate smooth trajectories for a UAV to search for and locate objects intended for pickup. After detecting the object, waypoint-based maneuvers were used to grasp, transport, and release it at a specified location. A similar methodology was adopted in \cite{8098676} for object detection, grasping, and delivery using UAVs.

Despite the substantial progress in cooperative load transportation, most existing approaches formulate the problem primarily in terms of trajectory tracking, formation regulation, payload stabilization, or direct motion control of the individual UAVs. Such formulations are effective for maintaining the desired cooperative configuration, but they do not explicitly address the terminal geometry of payload delivery, particularly when the platform itself is maneuvering. From a guidance perspective, payload delivery can instead be viewed as a terminal engagement problem between the transported payload and the landing platform. This viewpoint raises two fundamental questions: what terminal approach geometries are physically feasible for stationary and maneuvering platforms, and how can the cooperative system be guided to the corresponding terminal state within a convergence time that is independent of the initial engagement geometry? Addressing these questions requires a guidance formulation that captures the coupled motion of the UAVs through the payload geometry while retaining the relative range and line-of-sight variables that directly characterize terminal delivery. In this work, we develop such a formulation by representing the translational motion of the rigidly coupled UAV--payload system through a virtual equivalent agent and casting the payload delivery problem as a nonlinear engagement problem. This perspective enables the terminal feasibility of different landing geometries to be characterized explicitly and provides a systematic basis for the development of fixed-time cooperative guidance laws for both stationary and maneuvering landing platforms.

A virtual equivalent-agent representation is developed for a rigid two-UAV payload system, through which the coupled translational motion of the UAV--payload system is formulated as a reduced-order engagement problem with respect to the landing platform. This formulation establishes a direct guidance-level representation of cooperative payload delivery in terms of relative range, line-of-sight angle, and their associated dynamics.

The terminal delivery conditions are systematically analyzed for stationary and maneuvering landing platforms. It is established that a stationary platform enables an arbitrary prescribed landing angle, whereas exact delivery to a maneuvering platform with zero relative velocity necessarily requires synchronization of velocity and heading, resulting in a zero terminal landing angle. This analysis reveals a fundamental distinction between interception and successful payload landing on a moving platform.

A unified fixed-time guidance framework is developed for regulating the relative range and line-of-sight dynamics of the equivalent-agent--platform engagement. The proposed guidance laws provide convergence-time bounds that are independent of the initial engagement geometry and explicitly account for bounded platform maneuvers. The resulting framework enables the cooperative system to achieve the terminal conditions required for payload delivery to both stationary and maneuvering platforms.

A systematic acceleration-allocation and link-orientation control framework is developed to transform the virtual equivalent-agent guidance commands into commands for the individual UAVs while maintaining the prescribed rigid-link configuration. 

The performance of the proposed framework is evaluated through rigorous numerical simulations, including Monte Carlo studies, comparative performance analysis, and processor-in-the-loop implementation, which demonstrate its tracking performance and real-time computational feasibility.

\section{Problem Formulation}\label{ch8_sec:load_probelm}
Consider a system comprising two UAVs, denoted by $A$ and $B$, collaboratively transporting a rigid payload of length $L$ to a moving landing platform $P$. Each UAV is modeled as a point-mass nonholonomic vehicle operating at a constant altitude. Consequently, the engagement is confined to the horizontal $X$-$Y$ plane, and all motion along the vertical direction is neglected. A schematic representation of the UAV-payload geometry is shown in \Cref{fig:2UAV_system}.
\begin{figure}[!ht]
\centering
\begin{tikzpicture}[>=Latex, scale=1, transform shape]
\tikzset{
uav/.pic={
	\draw[line width=2pt, black!70] (-0.25,-0.25) -- (0.25,0.25);
	\draw[line width=2pt, black!70] (-0.25,0.25)  -- (0.25,-0.25);
	\filldraw[fill=uavrotor, draw=gray, thin] (-0.25,-0.25) circle (0.14);
	\filldraw[fill=uavrotor, draw=gray, thin] (0.25,0.25) circle (0.14);
	\filldraw[fill=uavrotor, draw=gray, thin] (-0.25,0.25) circle (0.14);
	\filldraw[fill=uavrotor, draw=gray, thin] (0.25,-0.25) circle (0.14);
	\filldraw[fill=uavbody, draw=black, thick, rounded corners=1pt] (-0.12,-0.12) rectangle (0.12,0.12);
	\filldraw[fill=red!80, draw=none] (0.08,0) circle (0.04); 
},
platform/.pic={
	\fill[black!30, rounded corners=4pt] (-0.55,-0.55) rectangle (0.75,0.45);
	
	\filldraw[fill=black!85, draw=black, rounded corners=1pt] (-0.45, 0.45) rectangle (-0.1, 0.65);
	\filldraw[fill=black!85, draw=black, rounded corners=1pt] (0.15, 0.45) rectangle (0.5, 0.65);
	\filldraw[fill=black!85, draw=black, rounded corners=1pt] (-0.45, -0.65) rectangle (-0.1, -0.45);
	\filldraw[fill=black!85, draw=black, rounded corners=1pt] (0.15, -0.65) rectangle (0.5, -0.45);
	
	\filldraw[fill=gray!25, draw=black!80, line width=1pt, rounded corners=4pt] (-0.6,-0.5) rectangle (0.7,0.5);
	\filldraw[fill=gray!15, draw=gray!40, rounded corners=2pt] (-0.5,-0.4) rectangle (0.6,0.4);
	
	\filldraw[fill=black!70, draw=black, rounded corners=2pt] (0.4,-0.3) rectangle (0.6,0.3);
	\filldraw[fill=cyan!20, draw=none, rounded corners=1pt] (0.45,-0.2) rectangle (0.55,0.2); 
	
	\filldraw[fill=yellow!90!white, draw=none] (0.65, 0.25) rectangle (0.7, 0.4);
	\filldraw[fill=yellow!90!white, draw=none] (0.65, -0.4) rectangle (0.7, -0.25);
	\filldraw[fill=red!80, draw=none] (-0.6, 0.25) rectangle (-0.55, 0.4);
	\filldraw[fill=red!80, draw=none] (-0.6, -0.4) rectangle (-0.55, -0.25);
	
	\filldraw[fill=black!80, draw=yellow, line width=1.5pt] (-0.05,0) circle (0.35);
	\draw[dashed, yellow, line width=0.8pt] (-0.05,0) circle (0.25);
	\node[text=yellow, rotate=-90] at (-0.05,0) {\Large \textbf{\textsf{L}}};
}
}
	\coordinate (Origin) at (0,0);
	\coordinate (U) at (3, 3);
	\def\xiAngle{30}
	\def\lenUA{2.5} 
	\def\lenUB{3.5} 
	\coordinate (A) at ($(U) + (\xiAngle:\lenUA)$);
	\coordinate (B) at ($(U) + (\xiAngle+180:\lenUB)$);
	\coordinate (P) at (7, 0);
	\draw[->, thick, gray!80] (-1.5,-1) -- (8.5,-1) node[right, text=black] {$X_I$};
	\draw[->, thick, gray!80] (-1.5,-1) -- (-1.5,6) node[above, text=black] {$Y_I$};
	\draw[line width=2.5pt, gray!60] (B) -- (A);
	\draw[dashed, thick, black!60] (B) -- (A);
	\pic[rotate=45] at (A) {uav}; 
	\node[below right=2pt and 12pt] at (A) {$A$}; 
	\pic[rotate=80] at (B) {uav}; 
	\node[below left=12pt and 2pt] at (B) {$B$}; 
	\filldraw[black] (U) circle (2.5pt) node[below=8pt] {$U$};
	\pic[rotate=45] at (P) {platform}; 
	\node[below right=20pt] at (P) {$P$};
	\draw[dashed, thin, gray!80] (U) -- ++(1.8,0) coordinate (Ux);
	\draw[dashed, thin, gray!80] (A) -- ++(1.5,0) coordinate (Ax);
	\draw[dashed, thin, gray!80] (B) -- ++(1.5,0) coordinate (Bx);
	\draw[dashed, thin, gray!80] (P) -- ++(1.5,0) coordinate (Px);
	\coordinate (A_dim) at ($(A) + (\xiAngle-90:0.8)$);
	\coordinate (U_dim) at ($(U) + (\xiAngle-90:0.8)$);
	\coordinate (B_dim) at ($(B) + (\xiAngle-90:0.8)$);
	\draw[<->, thick, black!70] (U_dim) -- (A_dim) node[midway, sloped, below=4pt] {$\sigma L$};
	\draw[<->, thick, black!70] (B_dim) -- (U_dim) node[midway, sloped, below=4pt] {$(1-\sigma) L$};
	\draw[thin, gray] (A) -- (A_dim);
	\draw[thin, gray] (U) -- (U_dim);
	\draw[thin, gray] (B) -- (B_dim);
	\draw[->, thick, velcol] (U) -- ++(65:1.5) coordinate (Vu) node[right=2pt] {$v_u$};
	\draw[->, thick, acccol] (U) -- ++(155:1.5) coordinate (Au) node[above left=-2pt] {$a_u$};
	\pic [draw, thin, gray, angle radius=0.25cm] {right angle = Vu--U--Au};
	\draw[->, thick, velcol] (A) -- ++(45:1.5) coordinate (Va) node[right=2pt] {$v_a$};
	\draw[->, thick, acccol] (A) -- ++(135:1.5) coordinate (Aa) node[above left=-3pt] {$a_a$};
	\pic [draw, thin, gray, angle radius=0.25cm] {right angle = Va--A--Aa};
	\draw[->, thick, velcol] (B) -- ++(80:1.5) coordinate (Vb) node[right=2pt] {$v_b$};
	\draw[->, thick, acccol] (B) -- ++(170:1.5) coordinate (Ab) node[below right = 3pt] {$a_b$};
	\pic [draw, thin, gray, angle radius=0.25cm] {right angle = Vb--B--Ab};
	\draw[->, thick, velcol] (P) -- ++(45:1.8) coordinate (Vp) node[above right=2pt] {$v_p$};
	\draw[->, thick, acccol] (P) -- ++(135:1.4) coordinate (Ap) node[above left=-2pt] {$a_p$};
	\pic [draw, thin, gray, angle radius=0.25cm] {right angle = Vp--P--Ap};
	\pic [draw, ->, angle radius=0.8cm] {angle = Ux--U--A};
	\node at ($(U) + (15:1.1)$) {$\xi$};
	\pic [draw, ->, velcol, angle radius=0.6cm] {angle = Ux--U--Vu};
	\node[velcol] at ($(U) + (48:0.9)$) {$\psi_u$}; 
	\pic [draw, ->, acccol, angle radius=1cm] {angle = Ux--U--Au};
	\node[acccol] at ($(U) + (110:1.2)$) {$\gamma_u$}; 
	\pic [draw, ->, velcol, angle radius=0.6cm] {angle = Ax--A--Va};
	\node[velcol] at ($(A) + (22.5:0.85)$) {$\psi_a$};
	\pic [draw, ->, acccol, angle radius=1.05cm] {angle = Ax--A--Aa};
	\node[acccol] at ($(A) + (100:1.2)$) {$\gamma_a$};
	\pic [draw, ->, velcol, angle radius=0.6cm] {angle = Bx--B--Vb};
	\node[velcol] at ($(B) + (50:0.85)$) {$\psi_b$}; 
	\pic [draw, ->, acccol, angle radius=1cm] {angle = Bx--B--Ab};
	\node[acccol] at ($(B) + (135:1.2)$) {$\gamma_b$};
	\pic [draw, ->, velcol, angle radius=0.9cm] {angle = Px--P--Vp};
	\node[velcol] at ($(P) + (22.5:1.2)$) {$\psi_p$};
	\pic [draw, ->, acccol, angle radius=0.7cm] {angle = Px--P--Ap};
	\node[acccol] at ($(P) + (100:0.95)$) {$\gamma_p$};
	\draw[->, rotcol] ($(U)+(-10:1.2)$) arc (-15:15:1.2) node[right=0pt] {$(\omega, \alpha)$};
	\draw[dash dot, thick, gray!80] (U) -- (P) node[pos=0.5, sloped, above=1pt, text=black] {LOS ($r$)};
\end{tikzpicture}
\caption{Cooperative two-UAV payload transportation and virtual equivalent agent representation.}
\label{fig:2UAV_system}
\end{figure}

To facilitate controller synthesis, the translational motion of the payload is represented by an equivalent agent ($U$) located on the rigid payload. The position of the equivalent agent is defined as a convex combination of the UAV positions. Under the rigid-link constraint, the configuration of the two-UAV--payload system can be parametrized by the equivalent-agent position $r_u$ together with the link orientation $\xi$. The mapping from the admissible physical configuration $(r_a,r_b)$ to $(r_u,\xi)$ is one-to-one and admits a smooth inverse. Thus, the translational motion can be represented through $r_u$, while the rotational degree of freedom is retained separately through $\xi$. This decomposition enables the translational guidance problem to be formulated as an engagement between the equivalent agent $U$ and the landing platform $P$.

To this end, we introduce three non-rotating and non-accelerating reference frames in the two-dimensional Euclidean space: two attached to UAVs $A$ and $B$, and one attached to the equivalent agent $U$. The rigid payload rotates about the axis perpendicular to the plane with angular velocity $\omega$ and angular acceleration $\alpha$. Let $\xi$ denote the angle between the segment $UA$ and the $X$--axis of the equivalent agent frame, measured counterclockwise. The angle subtended by the segment $UB$ is then $(\pi+\xi)$.

Let $j \in \{a,b,u,p \}$ denote the index for UAV $A$, UAV $B$, equivalent agent $U$, platform $P$, respectively.
Let $v_{j}$ and $a_{j}$ denote the magnitude of velocity and acceleration with corresponding direction $\psi_{j}$ and $\gamma_{j}$, respectively. The velocity and acceleration components of each vehicle can be expressed in the component form as
\begin{subequations}
	\begin{align}
		v_{j X} &= v_{j} \cos\psi_{j}, \ v_{j Y} = v_{j} \sin\psi_{j}, \label{eqn:vix_viy}\\
		a_{j X} &= a_{j} \cos\gamma_j, \ a_{j Y} = a_{j} \sin\gamma_{j}. \label{eqn:amix_amiy}
	\end{align}
\end{subequations}

Based on rigid-body kinematics, the velocities and accelerations of UAVs $A$ and $B$ are related to the equivalent agent $U$ as
\begin{subequations}
\begin{align}
	\begin{bmatrix}
		v_{j X}\\
		v_{j Y}\\
		0
	\end{bmatrix}
	&=
	\begin{bmatrix}
		v_{uX}\\
		v_{uY}\\
		0
	\end{bmatrix}
	+
	\begin{bmatrix}
		0\\
		0\\
		\omega
	\end{bmatrix}
	\times r_{j u}, \label{ch8_eq:vl_component_1}\\
	\begin{bmatrix}
		a_{j X}\\
		a_{j Y}\\
		0
	\end{bmatrix}
	&=
	\begin{bmatrix}
		a_{u X}\\
		a_{u Y}\\
		0
	\end{bmatrix}
	+
	\begin{bmatrix}
		0\\
		0\\
		\alpha
	\end{bmatrix}
	\times r_{j u}
	+
	\begin{bmatrix}
		0\\
		0\\
		\omega
	\end{bmatrix}
	\times
	\left(
	\begin{bmatrix}
		0\\
		0\\
		\omega
	\end{bmatrix}
	\times r_{j u}
	\right), \label{ch8_eq:am_component_1}
\end{align}
\end{subequations}
for $j \in \{a,b\}$, where $r_{j u}$ denotes the position of UAV $j$ relative to $U$. The relative position vectors are given by
\begin{align}
	r_{au} &=
	\begin{bmatrix}
		\sigma L \cos\xi\\
		\sigma L \sin\xi\\
		0
	\end{bmatrix},\,
	r_{bu} =
	\begin{bmatrix}
		-(1-\sigma)L \cos{\xi}\\
		-(1-\sigma)L \sin{\xi}\\
		0
	\end{bmatrix}. \label{ch8_eq:rau_rbu}
\end{align}
By substituting the values of $r_{au}$ and $r_{bu}$ from \eqref{ch8_eq:rau_rbu} into \eqref{ch8_eq:vl_component_1}, the components of the velocities can be obtained as 
\begin{subequations}\label{eqn:compvab}
	\begin{align}
		v_{aX}&=v_{uX} - \omega \sigma L \sin{\xi},\
		v_{aY}=v_{uY} + \omega \sigma L \cos{\xi},\\
		v_{bX}&=v_{uX} + \omega \left(1-\sigma\right) L \sin{\xi},\
		v_{bY}=v_{uY} - \omega \left(1-\sigma\right) L \cos{\xi}.
	\end{align} 
\end{subequations}
Similarly, by substituting  \eqref{ch8_eq:rau_rbu} into \eqref{ch8_eq:am_component_1}, one may obtain the components of the acceleration as
\begin{subequations}\label{eqn:compamab}
	\begin{align}
		a_{aX}&=a_{uX} - \alpha \sigma L \sin{\xi} - \omega^2 \sigma L \cos{\xi},\\
		a_{aY}&=a_{uY} + \alpha \sigma L \cos{\xi} - \omega^2 \sigma L \sin{\xi},\\
		a_{bX}&=a_{uX} + \alpha \left(1-\sigma\right) L \sin{\xi} + \omega^2 \left(1-\sigma\right) L \cos{\xi},\\
		a_{bY}&=a_{uY} - \alpha \left(1-\sigma\right) L \cos{\xi} + \omega^2 \left(1-\sigma\right) L \sin{\xi}.
	\end{align}
\end{subequations}
Note that \eqref{eqn:compvab} and \eqref{eqn:compamab} represent the components of the velocities and accelerations of UAVs $A$ and $B$. These components depend on the velocity and acceleration of the equivalent agent, as well as the parameters $\sigma$, $\alpha$, $L$, and $\xi$. The parameters $\sigma$, $\alpha$, and $\xi$ are related to the orientation of the rigid link, the design of which will be discussed later.

Equation \eqref{eqn:compamab} defines a one-to-one correspondence between the equivalent agent acceleration $(a_{uX},a_{uY})$ and the individual UAV accelerations $(a_{aX},a_{aY},a_{bX},a_{bY})$, for a given set of link parameters $(\sigma,L,\xi,\omega,\alpha)$. This enables the control design to be carried out at the equivalent agent level and subsequently mapped to the individual UAVs while preserving the rigid-link constraint.

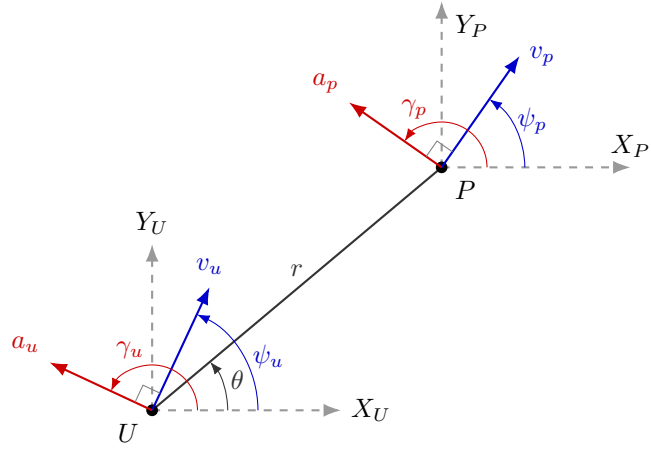
\begin{figure}[!ht]
\centering
\begin{tikzpicture}[>=Latex, scale=1]
	\coordinate (U) at (0,0);
	\def\thetaAng{40}
	\def\rDist{5}
	\coordinate (P) at (\thetaAng:\rDist);
	\draw[thick, loscol] (U) -- (P) node[midway, above=1pt] {$r$};
	\draw[dashed, ->, axiscol, thick] (U) -- ++(2.5,0) coordinate (Ux) node[right, text=black] {$X_U$};
	\draw[dashed, ->, axiscol, thick] (U) -- ++(0,2.2) coordinate (Uy) node[above, text=black] {$Y_U$};
	\filldraw[black] (U) circle (2pt) node[below left=2pt] {$U$};
	\draw[dashed, ->, axiscol, thick] (P) -- ++(2.5,0) coordinate (Px) node[above, text=black] {$X_P$};
	\draw[dashed, ->, axiscol, thick] (P) -- ++(0,2.2) coordinate (Py) node[below right =2pt, text=black] {$Y_P$};
	\filldraw[black] (P) circle (2pt) node[below right=2pt] {$P$};
	\def\psiU{65}
	\def\gammaU{155} 
	\def\psiP{55}
	\def\gammaP{145} 
	\draw[->, thick, velcol] (U) -- ++(\psiU:1.8) coordinate (Vu) node[above=1pt] {$v_u$};
	\draw[->, thick, acccol] (U) -- ++(\gammaU:1.5) coordinate (Au) node[above left=0pt] {$a_u$};
	\pic [draw, thin, gray, angle radius=0.25cm] {right angle = Vu--U--Au}; 
	\draw[->, thick, velcol] (P) -- ++(\psiP:1.8) coordinate (Vp) node[right=0pt] {$v_p$};
	\draw[->, thick, acccol] (P) -- ++(\gammaP:1.5) coordinate (Ap) node[above left=0pt] {$a_p$};
	\pic [draw, thin, gray, angle radius=0.25cm] {right angle = Vp--P--Ap}; 
	\pic [draw, ->, loscol, angle radius=1cm] {angle = Ux--U--P};
	\node[loscol] at ($(U) + (20:1.2)$) {$\theta$};
	\pic [draw, ->, velcol, angle radius=1.4cm] {angle = Ux--U--Vu};
	\node[velcol] at ($(U) + (25:1.7)$) {$\psi_u$};
	\pic [draw, ->, acccol, angle radius=0.6cm] {angle = Ux--U--Au};
	\node[acccol] at ($(U) + (110:0.85)$) {$\gamma_u$}; 
	\pic [draw, ->, velcol, angle radius=1.1cm] {angle = Px--P--Vp};
	\node[velcol] at ($(P) + (27.5:1.35)$) {$\psi_p$};
	\pic [draw, ->, acccol, angle radius=0.6cm] {angle = Px--P--Ap};
	\node[acccol] at ($(P) + (115:0.85)$) {$\gamma_p$}; 
\end{tikzpicture}
\caption{Engagement geometry between the equivalent agent and the landing platform.}
\label{ch8_fig:Virtual_UAV_Platform}
\end{figure}

The equivalent agent $U$, can maneuver in the two-dimensional space by changing its acceleration magnitude and direction ($a_{u}$ and $\gamma_{u}$), which will results into change in its speed $v_{u}$ and heading angle $\psi_{u}$. Similarly, the landing platform also changes its speed $v_p > 0$ and heading angle $\psi_{p}$ by adjusting the magnitude and direction of its acceleration $a_{p}$ and $\gamma_{p}$, respectively. 

To characterize the equation of motion between $U$ and $P$, we define four reference frames-- the earth-fixed inertial frame $\{X_{I}, Y_{I}\}$, two body-fixed frames $\{X_{U}, Y_{U}\}$ and $\{X_{P}, Y_{P}\}$, and a line-of-sight (LOS) frame $\{X_{L}, Y_{L}\}$. The LOS frame is oriented relative to the Earth-fixed inertial frame by an angle $\theta$, measured counterclockwise from the $X_{I}$ axis. Without loss of generality, we assume that the center of gravity of the equivalent agent $U$ is initially located at the origin of the Earth-fixed frame. Let $r$ denote the relative distance between $U$ and $P$, and $\theta$ denote the LOS angle measured counterclockwise from the inertial $X$-axis. This leads to the following engagement dynamics:
\begin{subequations} \label{eq:dynamics}
	\begin{align}
		\dot{r}=v_{r}=&~v_p\cos\left(\psi_{p}-\theta\right)-v_u\cos\left(\psi_{u}-\theta\right),\label{eq:rdot}\\
		r\dot{\theta}=v_{\theta}=&~v_p\sin\left(\psi_{p}-\theta\right)-v_u\sin\left(\psi_{u}-\theta\right), \label{eq:rthetadot}\\
		\dot{v}_u=&~a_{u}\cos\left(\gamma_{u}-\psi_{u}\right),\label{eq:vudot}\\
		v_{u}\dot{\psi}_u=&~a_{u}\sin\left(\gamma_{u}-\psi_{u}\right),\label{eq:vupsiudot}\\
		\dot{v}_p=&~a_{p}\cos\left(\gamma_{p}-\psi_{p}\right),\label{eq:vpdot}\\
		v_{p}\dot{\psi}_p=&~a_{p}\sin\left(\gamma_{p}-\psi_{p}\right).\label{eq:vppsipdot}
	\end{align}
\end{subequations}
It can be observed from \eqref{eq:vudot}-\eqref{eq:vupsiudot} that the scalar acceleration $a_u$ influences the translational dynamics through the relative orientation between the velocity vector and the acceleration direction, characterized by $(\gamma_u-\psi_u)$. Thus, variations in $a_u$ simultaneously change the speed and heading of the equivalent agent, with their effects governed by the instantaneous relative orientation. Consequently, directly deriving the control inputs becomes challenging. To circumvent this issue, the dynamics are reformulated in terms of acceleration components along suitably chosen directions, thereby yielding an equivalent representation that is affine in the newly defined control inputs. Using the geometrical relation from \Cref{ch8_fig:Virtual_UAV_Platform}, the relations in \eqref{eq:vudot}-\eqref{eq:vppsipdot} can be expressed in their equivalent, control-affine component form as
\begin{subequations}\label{eq:v_psi_dyn_comp}
\begin{align}	\dot{v}_{u}=&~a_{uX}\cos{\psi_{u}}+a_{uY}\sin{\psi_{u}},\label{eq:vudot_comp}\\
v_{u}\dot{\psi}_{u}=&~a_{uY}\cos{\psi_{u}}-a_{uX}\sin{\psi_{u}},\label{eq:vupsiudot_comp} \\
\dot{v}_{p}=&~a_{pX}\cos{\psi_{p}}+a_{pY}\sin{\psi_{p}},\label{eq:vpdot_comp}\\
v_{p}\dot{\psi}_{p}=&~a_{pY}\cos{\psi_{p}}-a_{pX}\sin{\psi_{p}}.\label{eq:vppsipdot_comp}
\end{align}
\end{subequations}
\begin{problem}
Consider the engagement between the equivalent agent and the landing platform, with relative kinematics governed by \eqref{eq:dynamics} and \eqref{eq:v_psi_dyn_comp}. Let the state vector be $\mathcal{C}$, comprising the engagement variables $(r,\ v_r,\ v_\theta)$. Design measurable control inputs $a_{uX}$ and $a_{uY}$ for the equivalent agent such that, for any admissible initial condition $\mathcal{C}(0)$, there exists a fixed time $T_{f}> 0$ for which the closed-loop trajectories remain in the target set $\mathscr{C}=\{\mathcal{C}\ \rvert\ r(t) =0,\ v_{\theta}(t)=0,\ v_{r}(t)=0\}$ $\forall$ $t \geq T_f$. Moreover, the control design must guarantee this fixed-time convergence independently of the initial engagement geometry.
\end{problem}

\begin{remark}
The target set $\mathscr C$ corresponds to a smooth landing configuration. The condition $r=0$ guarantees payload-platform coincidence, while $v_r=0$ and $v_\theta=0$ eliminate normal and tangential relative motion at contact. Consequently, achieving $\mathscr C$ ensures a soft and precise landing of the transported payload.
\end{remark}
\section{Design of the Guidance Strategy for the Two-UAV System} \label{ch8_sec:load_main}
In what follows, we develop guidance strategies for a two-UAV system to deliver a rigid object or payload to a landing platform, which may be either stationary or maneuvering. Before presenting the guidance strategy, we define the landing angle and derive the conditions on it for stationary and moving platforms.
\begin{definition}[Landing angle] \label{def:landing_angle}
	The \emph{landing angle}, shown in \Cref{fig:landing_angle}, is defined as the difference between the heading angle of the equivalent agent $U$, and the heading angle of the landing platform $P$, at the instant of landing, that is, $\theta_{L} = (\psi_{pf}-\psi_{uf})$ or $(\psi_{uf}-\psi_{pf})$.
\end{definition}
\begin{figure}[!ht]
\centering
\begin{tikzpicture}[>=Latex, scale=1]
	\coordinate (U) at (0, 5);
	\coordinate (P) at (3, 0);
	\coordinate (L) at (8, 3.5); 
	\def\psiU{45}
	\def\gammaU{135}   
	\def\psiP{50}
	\def\gammaP{140}   
	\def\psiUf{20}     
	\def\psiPf{60}     
	\draw[thick, black!80] (P) to[out=\psiP, in=210] (L);
	\draw[thick, black!80] (U) to[out=\psiU, in=190] (L);
	\draw[dashed, ->, gray!80, thick] (U) -- ++(3.0,0) coordinate (RefU) node[below, text=black] {$X_U$};
	\draw[dashed, ->, gray!80, thick] (P) -- ++(3.0,0) coordinate (RefP) node[right, text=black] {$X_P$};
	\draw[dashed, ->, gray!80, thick] (L) -- ++(3.0,0) coordinate (RefL) node[right, text=black] {};
	\draw[thin, loscol] (U) -- (P) coordinate[midway] (MidLOS);
	\node[below left=2pt] at (MidLOS) {$r$};
	\coordinate (Vu) at ($(U) + (\psiU:2.2)$);
	\coordinate (Au) at ($(U) + (\gammaU:1.5)$);
	\draw[->, thick, velcol] (U) -- (Vu) node[above right=-2pt] {$v_u$};
	\draw[->, thick, acccol] (U) -- (Au) node[above left=0pt] {$a_u$};
	\pic [draw, thin, gray, angle radius=0.25cm] {right angle = Vu--U--Au};
	\draw[->, loscol, thick] (U) ++(1.4,0) arc (0:-59:1.4) node[pos=0.7, right=2pt] {$-\theta$};
	\draw[->, velcol, thick] (U) ++(1.8,0) arc (0:\psiU:1.8) node[pos=0.6, above right=-2pt] {$\psi_u$};
	\draw[->, acccol, thick] (U) ++(0.9,0) arc (0:\gammaU:0.9) node[pos=0.7, above=2pt] {$\gamma_u$};
	\coordinate (Vp) at ($(P) + (\psiP:2.2)$);
	\coordinate (Ap) at ($(P) + (\gammaP:1.5)$);
	\draw[->, thick, velcol] (P) -- (Vp) node[above right=-2pt] {$v_p$};
	\draw[->, thick, acccol] (P) -- (Ap) node[above left] {$a_p$};
	\pic [draw, thin, gray, angle radius=0.25cm] {right angle = Vp--P--Ap};
	\draw[->, velcol, thick] (P) ++(1.4,0) arc (0:\psiP:1.4) node[pos=0.6,  right=2pt] {$\psi_p$};
	\draw[->, acccol, thick] (P) ++(0.6,0) arc (0:\gammaP:0.6) node[pos=0.7, above=2pt] {$\gamma_p$};
	\coordinate (Vuf) at ($(L) + (\psiUf:2.4)$);
	\coordinate (Vpf) at ($(L) + (\psiPf:2.4)$);
	\draw[->, thick, velcol] (L) -- (Vpf) node[above right=-2pt] {$v_{pf}$};
	\draw[->, thick, velcol] (L) -- (Vuf) node[below right=-2pt] {$v_{uf}$};
	\draw[->, velcol, thick] (L) ++(1.0,0) arc (0:\psiPf:1.0) node[pos=0.6, above right=-2pt] {$\psi_{pf}$};
	\draw[->, velcol, thick] (L) ++(1.4,0) arc (0:\psiUf:1.4) node[pos=0.9, below right=-2pt] {$\psi_{uf}$};
	\draw[<->, landcol, thick] (L) ++(\psiUf:1.9) arc (\psiUf:\psiPf:1.9) node[pos=0.3, above=3pt] {$\theta_L$};
	\filldraw[black] (U) circle (2pt) node[below left=2pt] {$U$};
	\filldraw[black] (P) circle (2pt) node[below=4pt, fill=white, inner sep=1pt] {$P$};
	\filldraw[black] (L) circle (2pt);
	\node[above=30pt of L, align=left] (PoLText) {Point of\\Landing};
	\draw[->, thin, gray] (PoLText) -- ($(L) + (0, 0.10)$);
\end{tikzpicture}
\caption{Illustration of landing angle.}
\label{fig:landing_angle}
\end{figure}
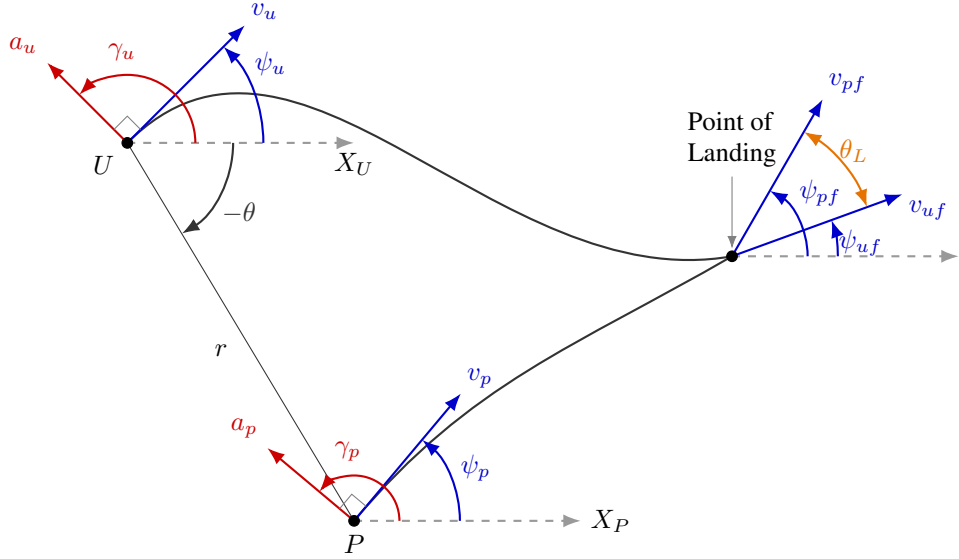

Intuitively, the desired landing angle can be achieved by controlling the LOS angle $\theta$, using the available control inputs. We first consider the scenario in which the UAV delivers the payload to a stationary platform. In this case, the platform velocity $v_p$ and heading angle $\psi_{p}$ are zero. Consequently, the landing angle is given by $\theta_{L}=-\psi_{uf}$, where $\psi_{uf}$ denotes the UAV heading angle at the instant of landing. Since the platform is stationary $a_{p}=0$, one can calculate the desired LOS angle $\theta_{d}$ accordingly. Ensuring convergence of the LOS angle to this desired value during the terminal phase is one of the control objectives. 

To land the UAVs on the platform at a prescribed landing angle $\theta_{L}$, the cooperative system must approach the platform along a \emph{constant-bearing approach trajectory}. This notion is directly analogous to the \emph{collision course} concept in classical intercept guidance, wherein LOS angle remains invariant over time. In intercept problems, a non-rotating LOS guarantees that the interceptor is on a trajectory that leads to eventual impact with the target. In the present load-transport setting, the same geometric condition is leveraged to ensure a well-defined and stable approach direction toward the landing platform. Formally, this requirement is expressed as $\dot{\theta}=0$, which enforces that the relative bearing between the equivalent agent and the platform remains constant throughout the terminal phase. From LOS dynamics \eqref{eq:rthetadot}, the desired landing angle $\theta_{d}$ can be calculated to enforce $\dot{\theta}=0$, and is given as
\begin{equation*}
	\dot{\theta}=0 \implies 0-\dfrac{v_u}{r}\sin\left(-\theta_{L}-\theta_{d}\right)=0,
\end{equation*}
which on solving for $\theta_d$, results into
\begin{equation}\label{eq:thetad}
	\theta_{d}=-\theta_{L}=\psi_{uf}.
\end{equation}
The expression in \eqref{eq:thetad} implies that any value of the landing angle $\theta_{L}$ can be selected, and the UAV can be maneuvered in real time to reach the landing platform at the specified angle. Furthermore, from \eqref{eq:rdot}, we have $\dot{r} = v_r = -v_u$, which indicates that the UAV must decelerate and eventually come to rest at the stationary platform to ensure successful payload delivery.  

Next, we consider the case of landing on a maneuvering platform shown in \Cref{fig:landing_angle}. In this scenario, both the platform velocity $v_p$ and acceleration $a_{p}$ are non-zero. Enforcing the collision course condition $\dot{\theta} = 0$ leads to a relationship between the UAV and platform velocities and headings at the landing instant as
\begin{align}  
	\dot{\theta} = &~ \dfrac{v_{\theta}}{r} = \dfrac{v_p\sin\left(\psi_{p} - \theta\right) - v_u\sin\left(\psi_{u} - \theta\right)}{r} = 0,
\end{align}  
which implies
\begin{align}  
	v_p\sin\left(\psi_{pf} - \theta_{d}\right) = v_u\sin\left(\psi_{uf} - \theta_{d}\right), \label{eq:thetadot_maneuvering}  
\end{align}
where $\psi_{pf}$ denotes the platform heading angle at the instant of landing. From \Cref{def:landing_angle}, the desired landing angle is given by
\begin{equation}\label{eq:theta_L_maneuvering}
	\theta_{L}=\psi_{pf}-\psi_{uf}.
\end{equation}
On using \eqref{eq:thetadot_maneuvering} and \eqref{eq:theta_L_maneuvering}, the desired LOS angle, $\theta_{d}$, can be determined. 

For a stationary platform, achieving $r=0$ is sufficient to complete landing, since persistence on the platform does not require matching any nonzero translational motion of the landing surface. However, for a maneuvering platform, terminal position coincidence alone is insufficient. If $\dot{r}(t_f)\neq 0$, then the UAV system and platform separate immediately after touchdown. Therefore, the landing task necessarily requires both terminal position agreement and terminal radial-velocity matching. Thus, whereas the constant-bearing condition $\dot{\theta}=0$ guarantees a collision-course-type approach geometry, the additional condition $\dot{r}=0$ is what distinguishes a feasible landing on a maneuvering platform from a mere intercept. Enforcing $\dot{r} = 0$ from \eqref{eq:rdot} yields an additional condition relating the velocities and headings of the UAV and the platform as
\begin{align}  
	\dot{r} = v_{r} = &~ v_p\cos\left(\psi_{pf} - \theta_{d}\right) - v_u\cos\left(\psi_{uf} - \theta_{d}\right) = 0 \implies v_p\cos\left(\psi_{pf} - \theta_{d}\right) = v_u\cos\left(\psi_{uf} - \theta_{d}\right). \label{eq:rdot_maneuvering}  
\end{align}  
The only feasible solution satisfying both the LOS and range constraints from \eqref{eq:thetadot_maneuvering} and \eqref{eq:rdot_maneuvering} is $v_u = v_p$ and $\psi_{uf} = \psi_{pf}$, implying that the UAV must exactly match both the speed and heading of the platform. As a result, the only admissible landing angle in this case is $\theta_{L} = 0$. It is worth noting that when $v_u = v_p$ and $\psi_{uf} = \psi_{pf}$, the collision course and range-rate conditions given in \eqref{eq:thetadot_maneuvering} and \eqref{eq:rdot_maneuvering} are satisfied regardless of the final value of the LOS angle, $\theta_{d}$. Therefore, for maneuvering platforms, these conditions constitute sufficient conditions for successful landing and post-landing tracking.

To determine the control inputs for the UAV that achieve the specified objectives, we first derive the dynamics of the relative range and LOS angle with respect to the UAV's steering control inputs. The following lemmas establish the relative degree of the range and LOS angle dynamics with respect to the UAV and platform control inputs.
\begin{lemma}
	The dynamics of the relative range $r$ has a relative degree of two with respect to the control inputs of the equivalent agent $U$ and the platform $P$.
\end{lemma}
\begin{proof}
On differentiating the range rate given in \eqref{eq:rdot} with respect to time, one may obtain
\begin{align*}
	\ddot{r}=&~-v_p\sin\left(\psi_{p}-\theta\right)\left(\dot{\psi}_{p}-\dot{\theta}\right)+\cos\left(\psi_{p}-\theta\right) \dot{v}_{p}+v_u\sin\left(\psi_{u}-\theta\right)\left(\dot{\psi_{u}}-\dot{\theta}\right)-\cos\left(\psi_{u}-\theta\right) \dot{v}_{u},
\end{align*}
which can be written using \eqref{eq:vpdot} and \eqref{eq:vudot} as
\begin{align}
	\nonumber \ddot{r}&=-v_p\sin\left(\psi_{p}-\theta\right)\left(\dot{\psi_{p}}-\dot{\theta}\right)+\cos\left(\psi_{p}-\theta\right) \left(a_{pX}\cos{\psi_{p}}+a_{pY}\sin{\psi_{p}}\right)+v_u\sin\left(\psi_{u}-\theta\right)\left(\dot{\psi_{u}}-\dot{\theta}\right)\\
	&~-\cos\left(\psi_{u}-\theta\right) \left(a_{uX}\cos{\psi_{u}}+a_{uY}\sin{\psi_{u}}\right). \label{eq:rddot1}
\end{align}
By expanding and doing some algebraic simplifications, the expression in \eqref{eq:rddot1} becomes
\begin{align*}
	\ddot{r}=&~-v_p \dot{\psi}_{p} \sin\left(\psi_{p}-\theta\right) + v_p  \dot{\theta}\sin\left(\psi_{p}-\theta\right)+a_{pX}\cos{\psi_{p}}\cos\left(\psi_{p}-\theta\right)+a_{pY}\sin{\psi_{p}}\cos\left(\psi_{p}-\theta\right)\\
	&~+v_u\dot{\psi}_{u}\sin\left(\psi_{u}-\theta\right)-v_u \dot{\theta}\sin\left(\psi_{u}-\theta\right) -a_{uX}\cos{\psi_{u}}\cos\left(\psi_{u}-\theta\right)-a_{uY}\sin{\psi_{u}}\cos\left(\psi_{u}-\theta\right),
\end{align*}
which can be simplified using \eqref{eq:vupsiudot} and \eqref{eq:vppsipdot} to
\begin{align}
	\ddot{r}=&~-v_p \sin\left(\psi_{p} -\theta\right) \left(\dfrac{a_{pY}\cos{\psi_{p}}-a_{pX}\sin{\psi_{p}}}{v_p}\right)+ v_p  \dot{\theta}\sin\left(\psi_{p}-\theta\right)-v_u \dot{\theta}\sin\left(\psi_{u}-\theta\right)\nonumber\\
	&~+a_{pX}\cos{\psi_{p}}\cos\left(\psi_{p}-\theta\right)+v_u\sin\left(\psi_{u}-\theta\right)\left(\dfrac{a_{uY}\cos{\psi_{u}}-a_{uX}\sin{\psi_{u}}}{v_u}\right) \nonumber\\
	&~+a_{pY}\sin{\psi_{p}}\cos\left(\psi_{p}-\theta\right)-a_{uX}\cos{\psi_{u}}\cos\left(\psi_{u}-\theta\right)-a_{uY}\sin{\psi_{u}}\cos\left(\psi_{u}-\theta\right).\label{eq:rddot2}
\end{align}
Rearranging terms of \eqref{eq:rddot2} and performing some algebraic simplifications leads us to arrive at
\begin{align}
	\ddot{r}=&~a_{pY}\left(\sin{\psi_{p}}\cos\left(\psi_{p}-\theta\right)-\cos\psi_{p}\sin\left(\psi_{p} -\theta\right)\sin\left(\psi_{p}\right)\right)+a_{pX}\left(\sin{\psi_{p}}\sin\left(\psi_{p} -\theta\right)\right. \nonumber\\
	&\left. +\cos{\psi_{p}}\cos\left(\psi_{p}-\theta\right)\right)-a_{uX}\left(\sin{\psi_{u}} \sin\left(\psi_{u}-\theta\right)+ \cos{\psi_{u}}\cos\left(\psi_{u}-\theta\right) \right)  \nonumber\\
	&~+a_{uY}\left( \cos{\psi_{u}} \sin\left(\psi_{u}-\theta\right)-\sin{\psi_{u}}\cos\left(\psi_{u}-\theta\right)\right) +v_p  \dot{\theta}\sin\left(\psi_{p}-\theta\right)-v_u \dot{\theta}\sin\left(\psi_{u}-\theta\right),\nonumber
\end{align}
which, on some trigonometric simplification, results in
\begin{align*}
	\ddot{r}=a_{pX}\cos{\theta}-a_{uX}\cos{\theta}-a_{uY}\sin{\theta}+ \dot{\theta}\left(v_p  \sin\left(\psi_{p}-\theta\right)-v_u \sin\left(\psi_{u}-\theta\right)\right)+a_{pY}\sin{\theta},
\end{align*}
which can be further simplified using \eqref{eq:rthetadot} to
\begin{align}
	\ddot{r}=&~\left(a_{pY}-a_{uY}\right) \sin{\theta}+ \left(a_{pX}-a_{uX}\right)\cos{\theta}+ r \dot{\theta}^2 =\left(a_{pY}-a_{uY}\right) \sin{\theta}+ \left(a_{pX}-a_{uX}\right)\cos{\theta}+ r \dot{\theta}^2. \label{ch8_eq:r_dynamics_1}
\end{align}
It readily follows from \eqref{ch8_eq:r_dynamics_1} that the relative range has a relative degree of two with respect to the lateral acceleration components of the UAV and the platform. This completes the proof.
\end{proof}
\begin{lemma}\label{lem:theta}
	The dynamics of the LOS angle has a relative degree of two with respect to the UAV's and platform control inputs.
\end{lemma}
\begin{proof}
We differentiate both sides of \eqref{eq:rthetadot} with respect to time to get
\begin{align*}
	\dot{r}\dot{\theta} + r\ddot{\theta}&= \dot{v}_p\sin(\psi_p - \theta) + v_p\cos(\psi_p - \theta)(\dot{\psi}_p - \dot{\theta}) -\left[\dot{v}_u\sin(\psi_u - \theta) + v_u\cos(\psi_u - \theta)(\dot{\psi}_u - \dot{\theta})\right],
\end{align*}
implying
\begin{align}
	r\ddot{\theta} = &\left[\dot{v}_p\sin(\psi_p - \theta) + v_p\dot{\psi}_p\cos(\psi_p - \theta)\right] - \left[\dot{v}_u\sin(\psi_u - \theta) + v_u\dot{\psi}_u\cos(\psi_u - \theta)\right] \nonumber \\
	&- \dot{\theta}\left[v_p\cos(\psi_p - \theta) - v_u\cos(\psi_u - \theta)\right] - \dot{r}\dot{\theta}\label{eq:theta_ddot_1},
\end{align}
By substituting the value of $\dot{v}_p$, $\dot{\psi}_p$, $\dot{v}_u$, and $\dot{\psi}_u$ from \eqref{eq:v_psi_dyn_comp} into \eqref{eq:theta_ddot_1}, one may obtain
\begin{align}
	r\ddot{\theta} &= (a_{pX}\cos{\psi_{p}}+a_{pY}\sin{\psi_{p}})\sin(\psi_p - \theta) + (a_{pY}\cos{\psi_{p}}-a_{pX}\sin{\psi_{p}})\cos(\psi_p - \theta) \nonumber\\
	&-\left[(a_{uX}\cos{\psi_{u}}+a_{uY}\sin{\psi_{u}})\sin(\psi_u - \theta) + (a_{uY}\cos{\psi_{u}}-a_{uX}\sin{\psi_{u}})\cos(\psi_u - \theta)\right] \nonumber\\
	&- \dot{\theta}\left[v_p\cos(\psi_p - \theta) - v_u\cos(\psi_u - \theta)\right] - \dot{r}\dot{\theta}. \label{eq:theta_ddot_2}
\end{align}
Collecting the acceleration components terms, $a_{pX}$ and $a_{pY}$, the expression in \eqref{eq:theta_ddot_2} becomes
\begin{align}
	r\ddot{\theta} &= a_{pY}\left(\cos\psi_p\cos(\psi_p - \theta) + \sin\psi_p\sin(\psi_p - \theta)\right)  - a_{pX}\left(\sin\psi_p\cos(\psi_p - \theta) - \cos\psi_p\sin(\psi_p - \theta)\right) \nonumber\\
	&-\left[a_{uY}(\cos{\psi_{u}}\cos(\psi_u - \theta) + \sin{\psi_{u}}\sin(\psi_u - \theta))- a_{uX}(\sin{\psi_{u}}\cos(\psi_u - \theta) - \cos{\psi_{u}}\sin(\psi_u - \theta))\right] \nonumber\\
	&- \dot{\theta}\left[v_p\cos(\psi_p - \theta) - v_u\cos(\psi_u - \theta)\right] -  \dot{r}\dot{\theta}.\label{eq:theta_ddot_3}
\end{align}
Using the trigonometric relations $\cos(P-Q) = \cos P \cos Q + \sin{P}\sin Q$, $\sin(P-Q) = \sin P \cos Q - \cos P \sin Q$, and \eqref{eq:rdot}, the expression in \eqref{eq:theta_ddot_3} simplifies to
\begin{align*}
	r\ddot{\theta} =&~ a_{pY}\cos(\psi_p - (\psi_p - \theta)) - a_{pX}\sin(\psi_p - (\psi_p - \theta)) - \left[a_{uY}\cos(\psi_u - (\psi_u - \theta)) \right.\\
	&\left.- a_{uX}\sin(\psi_u - (\psi_u - \theta))\right] - 2\dot{r}\dot{\theta}= a_{pY}\cos\theta - a_{pX}\sin\theta - \left(a_{uY}\cos(\theta) - a_{uX}\sin(\theta)\right) - 2\dot{r}\dot{\theta},
\end{align*}
which can be written as
\begin{align}
	\ddot{\theta} =&~ \frac{ (a_{pY} - a_{uY})\cos\theta - (a_{pX} - a_{uX})\sin\theta - 2\dot{r}\dot{\theta} }{r}, \nonumber  \\
	=&~ -\frac{2\dot{r}\dot{\theta}}{r} + \frac{ (a_{pY}\cos\theta - a_{pX}\sin\theta) - (a_{uY}\cos\theta - a_{uX}\sin\theta) }{r}.\label{eq:los_dynamics_1}
\end{align}
Thus, the dynamics of $\theta$ have a relative degree of two with respect to UAV's and platform control inputs.
\end{proof}

The dynamics of the relative range and the LOS angle can also be expressed as
\begin{subequations}\label{eq:r_theta_dynamics}
	\begin{align}
		\ddot{r}=&~  r \dot{\theta}^2 + \left( a_{p r} - a_{u r} \right), \label{eq:r_dynamics}   \\
		\ddot{\theta}=&~-\dfrac{2\dot{r}\dot{\theta}}{r} + \dfrac{1}{r} \left( a_{p \theta}-a_{u  \theta} \right), \label{eq:los_dynamics}
	\end{align}
\end{subequations}
where a change of coordinates is introduced as
\begin{align}
	\begin{bmatrix}
		a_{j r} \\
		a_{j \theta}
	\end{bmatrix}
	=
	\begin{bmatrix}
		\cos{\theta} & \sin{\theta}\\
		-\sin{\theta} & \cos{\theta}
	\end{bmatrix}
	\begin{bmatrix}
		a_{jX} \\
		a_{jY}
	\end{bmatrix}, \label{ch8_eq:am_transformation}
\end{align}
for $j=p$, $u$. One can note that the matrix \eqref{ch8_eq:am_transformation} is invertible for all values of LOS angle $\theta$. 
\begin{remark}\label{rem:los_singularity}
	The LOS angle dynamics \eqref{eq:los_dynamics} are well-defined only for $r>0$. This is not restrictive, since regulation of the LOS angle is required only during the approach phase. Once $r=0$, the payload is in contact with the platform, and the LOS angle becomes physically irrelevant.
\end{remark}
\subsection{Payload Delivery on a Stationary Platform}
We now derive the guidance strategies for the UAV system to deliver the payload on a stationary platform, that is, $a_{pX}=a_{pY}=v_{p}=0$. For such a scenario, the relative range and LOS angle dynamics can be written from \eqref{eq:r_theta_dynamics} as
\begin{subequations}\label{eq:r_theta_dynamics_stat}
	\begin{align}
		\ddot{r}=&~ r \dot{\theta}^2  - a_{u r} , \label{eq:r_dynamics_stat}   \\
		\ddot{\theta}=&~-\dfrac{2\dot{r}\dot{\theta}}{r}  - \dfrac{a_{u  \theta}}{r}.\label{eq:los_dynamics_stat}
	\end{align}
\end{subequations}
Having said that, the UAV can be steered to land the payload on the stationary platform at any arbitrary predefined landing angle. Thus, for a given $\theta_{L}$ our objective is to design $a_{ur}$ and $a_{u\theta}$ such that the conditions $r=\dot{r}=0$ and $\theta=\theta_{d}$ are satisfied. 

It is a fundamental requirement to steer the UAV towards the landing platform in order to deliver the payload onto it. Thus, we next endeavor to design the guidance law $a_{ur}$ such that the target set $\mathscr{R}=\{\mathcal{R} \; \rvert \; r=\dot{r}=0 \; \forall\; t\; \geq \; t_1 \}$, where $\mathcal{C}$ is a set of relevant states and $t_1 \in \mathbb{R}_{+}$ is some fixed time instant, is achieved. Toward this objective, we choose a sliding manifold motivated by the work in \cite{10.1016/j.automatica.2021.110009} as
\begin{equation} \label{eq:sr}
	\mathcal{S}_{1} = \dot{r} + \beta_1  r^{\dfrac{\lambda_1 {r}^2}{1+\mu_1 {r}^2}},
\end{equation}
where $\beta_{1}>0$, $\mu_1>0$, $\lambda_1>0$ such that $\Upsilon_1 \coloneqq \dfrac{\lambda_1}{1+\mu_1} > 1$ are design parameters. 
\begin{remark}
	The function $\chi: r \to r^{\dfrac{\lambda_1 {r}^2}{1+\mu_1 {r}^2}} = \exp\left({\dfrac{\lambda_1 {r}^2 }{1+\mu_1 {r}^2}}\ln(r)\right)$ is continuous at $r=0$ with $\chi(0)=0$ and thus the second term of sliding manifold \eqref{eq:sr} is locally bounded.
\end{remark}
The following lemma aids in designing the UAV's guidance command $a_{ur}$.
\begin{lemma} \label{lem:sr}
	The sliding manifold given in \eqref{eq:sr} has a relative degree of one with respect to the UAV's control input $a_{ur}$.
\end{lemma}
\begin{proof}
One may obtain the derivative of $\mathcal{S}_{1}$ by differentiating \eqref{eq:sr} with respect to time as
\begin{align}
	\dot{\mathcal{S}}_1 &= \frac{d}{dt} \left( \dot{r} + \beta_1 r^{\dfrac{\lambda_1 {r}^2}{1+\mu_1 {r}^2}} \right) = \ddot{r} + \beta_1 \frac{d}{dt} \left( r^{\dfrac{\lambda_1 {r}^2}{1+\mu_1 {r}^2}} \right). \label{eq:s1_dot_1}
\end{align}
To calculate the derivative of the second term, we define a function of the form $f(r) = r^{g(r)}$, where the exponent is $g(r) = \dfrac{\lambda_1 r^2}{1+\mu_1 r^2}$. We use logarithmic differentiation to find its derivative with respect to $r$. Let $y = r^{g(r)}$. Taking the natural logarithm of both sides results in $\ln(y) = \ln\left(r^{g(r)}\right) = g(r) \ln(r)$.
Differentiating implicitly with respect to $r$ using the product rule on the right-hand side, one may obtain
\begin{align}
	\frac{1}{y} \frac{dy}{dr} &= g'(r) \ln(r) + g(r) \frac{1}{r} \implies
	\frac{dy}{dr} = y \left( g'(r) \ln(r) + \frac{g(r)}{r} \right) = r^{g(r)} \left( g'(r) \ln(r) + \frac{g(r)}{r} \right). \label{eq:dydr_1}
\end{align}
Next, we find the derivative of the exponent, $g'(r)$, using the quotient rule as
\begin{align*}
	g'(r) &= \frac{d}{dr} \left( \frac{\lambda_1 r^2}{1+\mu_1 r^2} \right)= \frac{(2 \lambda_1 r)(1+\mu_1 r^2) - (\lambda_1 r^2)(2 \mu_1 r)}{(1+\mu_1 r^2)^2}= \frac{2 \lambda_1 r + 2 \lambda_1 \mu_1 r^3 - 2 \lambda_1 \mu_1 r^3}{(1+\mu_1 r^2)^2} \\
	&~= \frac{2 \lambda_1 r}{(1+\mu_1 r^2)^2}.
\end{align*}
Now, by substituting the values of $g(r)$ and $g'(r)$ into \eqref{eq:dydr_1}, the expression of $\dfrac{dy}{dr}$ becomes
\begin{align}
	\frac{dy}{dr} &= r^{\frac{\lambda_1 r^2}{1+\mu_1 r^2}} \left( \frac{2 \lambda_1 r \ln(r)}{(1+\mu_1 r^2)^2} + \frac{1}{r}  \frac{\lambda_1 r^2}{1+\mu_1 r^2} \right) = r^{\frac{\lambda_1 r^2}{1+\mu_1 r^2}} \left( \frac{2 \lambda_1 r \ln(r)}{(1+\mu_1 r^2)^2} + \frac{\lambda_1 r}{1+\mu_1 r^2} \right). \label{eq:dydr}
\end{align}
Using the chain rule, the time derivative of the second term of \eqref{eq:s1_dot_1} is $\dfrac{d}{dt} \left( r^{g(r)} \right) = \dfrac{dy}{dr} \dfrac{dr}{dt} = \dfrac{dy}{dr} \dot{r}$. By substituting the value of $\dfrac{dy}{dr}$ from \eqref{eq:dydr} into \eqref{eq:s1_dot_1}, one may obtain the value of $\dot{\mathcal{S}}_{1}$ as
\begin{equation}
	\dot{\mathcal{S}}_1 = \ddot{r} + \beta_1 \dot{r} r^{\frac{\lambda_1 r^2}{1+\mu_1 r^2}} \left( \frac{2 \lambda_1 r \ln(r)}{(1+\mu_1 r^2)^2} + \frac{\lambda_1 r}{1+\mu_1 r^2} \right),
\end{equation}
which, after some algebraic simplifications, results in
\begin{align} 
	\dot{\mathcal{S}}_{1} &= \ddot{r} + \beta_1 \dot{r} r^{\frac{\lambda_1 r^2}{1+\mu_1 r^2}}\left(\frac{2 \lambda_1 r \ln(r) + \lambda_1 r + \lambda_1 \mu_1 r^3}{(1+\mu_1 r^2)^2}\right) =\ddot{r}+ \beta_1 \dot{r} r^{\frac{\lambda_1 r^2}{1+\mu_1 r^2}} \left(\frac{\lambda_1 r (2 \ln(r) + 1 + \mu_1 r^2)}{(1+\mu_1 r^2)^2}\right) \nonumber\\
	&=\ddot{r} + \frac{\beta_{1} \lambda_1  r^{{\frac{\lambda_1 r^{2}}{\mu_{1} r^{2}+1}}+1}\left(2 \ln ( r)+\mu_{1} r^{2}+1\right)}{\left(\mu_{1} r^{2}+1\right)^{2}}\dot{r} \label{ch8_eq:dot_sr_1}.
\end{align}
On substituting the value of $\ddot{r}$ from \eqref{eq:r_dynamics} into \eqref{ch8_eq:dot_sr}, we get
\begin{equation}\label{ch8_eq:dot_sr}
	\dot{\mathcal{S}}_1 =    r \dot{\theta}^2  - a_{u r} + \frac{\beta_{1} \lambda_1  r^{{\dfrac{\lambda_1 r^{2}}{\mu_{1} r^{2}+1}}+1}\left(2 \ln ( r)+\mu_{1} r^{2}+1\right)}{\left(\mu_{1} r^{2}+1\right)^{2}}\dot{r} =r \dot{\theta}^2  + \dot{r} \mathcal{F}(r) - a_{u r} , 
\end{equation}
where
\begin{equation}\label{eq:fr}
	\mathcal{F}(r) =  \dfrac{\beta_1 \lambda_1  r}{1+\mu_1 {r}^2}\left(\dfrac{2 \ln( r)}{1+\mu_1 {r}^2} + 1\right)  r^{\dfrac{\lambda_1 {r}^2}{1+\mu_1 {r}^2}}. 
\end{equation}
It follows from \eqref{ch8_eq:dot_sr} that $\mathcal{S}_{1}$ has a relative degree of one with respect to the control input $a_{u r}$. 
\end{proof}
\begin{remark}\label{rem:log_bound}
	Although the sliding manifold \eqref{eq:sr} contains logarithmic terms, the function $r^{\frac{\lambda_1 r^2}{1+\mu_1 r^2}}\ln(r)$
	is bounded in a neighborhood of the origin. Consequently, all terms appearing in $\dot{\mathcal{S}}_1$ remain locally bounded, ensuring well-posed closed-loop dynamics.
\end{remark}
We now present the design of the UAV's guidance command $a_{ur}$ in the following theorem. 
\begin{theorem}\label{thm:stat_amr}
	Consider the relative range dynamics as in \eqref{eq:r_dynamics_stat} with the sliding manifold chosen as per \eqref{eq:sr}. If the UAV's guidance command $a_{ur}$ is designed as
	\begin{align}
		a_{ur} = r \dot{\theta}^2 +  \dot{r} \mathcal{F}(r) +\mathcal{K}_{1} \lvert \mathcal{S}_{1}\rvert^{\dfrac{\lambda_2 \mathcal{S}_{1}^2}{1+\mu_2 \mathcal{S}_{1}^2}}\sign( \mathcal{S}_{1}), \label{ch8_eq:amr}
	\end{align}
	with $\mathcal{K}_1>0$, $\lambda_2>0$, $\mu_2>0$, such that $\Upsilon_2 \coloneqq \dfrac{\lambda_2}{1+\mu_2} > 1$, then, the transported payload is brought to the stationary landing platform with zero residual closing speed in the radial direction within a fixed-time $\mathcal{T}_{r}$, upper bounded by
	\begin{align}
		\mathcal{T}_{r} \leq  \dfrac{1}{\mathcal{K}_{1}\left(\Upsilon_2-1\right)}+\dfrac{1}{\mathcal{K}_{1} \mathcal{M}_{2}}+\dfrac{1}{\beta_{1}\left(\Upsilon_1-1\right)}+ \dfrac{1}{\beta_{1} \mathcal{M}_{1}} \label{eq:Tr},
	\end{align}
	where $\mathcal{M}_1=e^{\frac{-\lambda_{1}}{2 e}}$ and $\mathcal{M}_2=e^{\frac{-\lambda_{2}}{2 e}}$, regardless of the initial UAV-platform configuration. 
\end{theorem}
\begin{proof}
Consider a radially unbounded quadratic Lyapunov function candidate as $\mathcal{V}_{1}=\mathcal{S}_{1}^2$, which on differentiating with respect to time results in $\dot{\mathcal{V}}_{1} = 2\mathcal{S}_{1} \dot{\mathcal{S}}_{1}$. By substituting the value of $\dot{\mathcal{S}}_{1}$ from \eqref{ch8_eq:dot_sr}, one may obtain $\dot{\mathcal{V}}_{1}$ as
\begin{equation}\label{ch8_eq:dot_v_2}
	\dot{\mathcal{V}}_1 =    2\mathcal{S}_{1} \left(r \dot{\theta}^2  - a_{u r} + \frac{\beta_{1} \lambda_1  r^{{\dfrac{\lambda_1 r^{2}}{\mu_{1} r^{2}+1}}+1}\left(2 \ln ( r)+\mu_{1} r^{2}+1\right)}{\left(\mu_{1} r^{2}+1\right)^{2}}\dot{r}\right). 
\end{equation}
With the guidance command given in \eqref{ch8_eq:amr}, the expression in \eqref{ch8_eq:dot_v_2} becomes
\begin{equation*}
	\dot{\mathcal{V}}_{1} =   2 \mathcal{S}_{1} \left[r \dot{\theta}^2  - \left(r \dot{\theta}^2 +\mathcal{K}_{1} | \mathcal{S}_{1}|^{\dfrac{\lambda_2 \mathcal{S}_{1}^2}{1+\mu_2 \mathcal{S}_{1}^2}} \sign( \mathcal{S}_{1})+  \dot{r} \mathcal{F}(r)\right)+ \dfrac{\beta_{1} \lambda_1  r^{{\dfrac{\lambda_1 r^{2}}{\mu_{1} r^{2}+1}}+1}\left(2 \ln ( r)+\mu_{1} r^{2}+1\right)}{\left(\mu_{1} r^{2}+1\right)^{2}}\dot{r}\right] , 
\end{equation*}
which, after some algebraic simplifications, results in
\begin{equation}\label{ch8_eq:dot_v_3}
	\dot{\mathcal{V}}_{1} =   2 \mathcal{S}_{1} \left( -\mathcal{K}_{1}| \mathcal{S}_{1}|^{\dfrac{\lambda_2 \mathcal{S}_{1}^2}{1+\mu_2 \mathcal{S}_{1}^2}} \sign( \mathcal{S}_{1}) \right)= -2\mathcal{K}_{1}| \mathcal{S}_{1}|^{\left(\dfrac{\lambda_2 \mathcal{S}_{1}^2}{1+\mu_2 \mathcal{S}_{1}^2}+1\right)}. 
\end{equation}
Now, consider the case when $\mathcal{V}_{1} \geq 1$. Then, one has $\frac{\lambda_2 \mathcal{S}_{1}^2}{1+\mu_2 \mathcal{S}_{1}^2}+1 \geq \frac{\lambda_2}{1+\mu_2}+1>2$. Also, for $\mathcal{V}_{1} \geq 1$,  $\mathcal{S}_{1} \geq 1$ together with the fact $\Upsilon_{2}= \frac{\lambda_2}{1+\mu_2} >1$, implies
\begin{equation}\label{ch8_eq:dot_v_4}
	\dot{\mathcal{V}}_{1} =  -2\mathcal{K}_{1}  \lvert \mathcal{S}_{1}\rvert^{\left(\dfrac{\lambda_2 \mathcal{S}_{1}^2}{1+\mu_2 \mathcal{S}_{1}^2}+1\right)}\leq - 2\mathcal{K}_{1}  \lvert \mathcal{S}_{1}\rvert^{\Upsilon_{2}+1} \leq -2\mathcal{K}_{1}  \mathcal{V}_{1}^{\dfrac{\Upsilon_{2}+1}{2}}. 
\end{equation}
It immediately follows from \cite[Lemma 1]{6104367} that all the solutions starting from $\{ \mathcal{V}_{1} \geq 1\}$ reach to the set $\{ \mathcal{V}_{1} \leq 1\}$ within a fixed time $\mathcal{T}_{1} \leq \dfrac{1}{\mathcal{K}_{1}\left(\Upsilon_{2}-1\right)}$ as $\mathcal{K}_{1} > 0$ and $\dfrac{\Upsilon_{2}+1}{2}>1$.

Next consider the case when $\mathcal{V}_{1} \leq  1$. For this case, we have $\mathcal{S}_{1} \leq 1$, so the expression of $\dot{\mathcal{V}}_{1}$ given in \eqref{ch8_eq:dot_v_3} can be written as
\begin{equation}\label{ch8_eq:dot_v_5}
	\dot{\mathcal{V}}_{1}  = - 2 \mathcal{K}_{1} \lvert \mathcal{S}_{1}\rvert \lvert \mathcal{S}_{1}\rvert^{\dfrac{\lambda_2 \mathcal{S}_{1}^2}{1+\mu_2 \mathcal{S}_{1}^2}}.
\end{equation}
As $\mathcal{S}_{1} \leq 1$, we  have  $1+\mu_2 \mathcal{S}_{1}^2 \geq 1$, which implies $\min\left( \lvert \mathcal{S}_{1}\rvert^{\frac{\lambda_2 \mathcal{S}_{1}^2}{1+\mu_2 \mathcal{S}_{1}^2}} \right) \geq \min \left(  \lvert \mathcal{S}_{1} \rvert^{\lambda_2  \mathcal{S}_{1} ^2} \right) = e^{\frac{-\lambda_2}{2e}}$, and \eqref{ch8_eq:dot_v_5} can be written as 
\begin{equation}\label{ch8_eq:dot_v_6}
	\dot{\mathcal{V}}_{1} \leq  -     2\mathcal{K}_{1} \lvert \mathcal{S}_{1}\rvert e^{\frac{-\lambda_2}{2e}} \leq -2\mathcal{K}_{1} e^{\frac{-\lambda_2}{2e}} \sqrt{\mathcal{V}_{1}}.
\end{equation}
Using the results from \cite[Theorem 4]{doi:10.1137/S0363012997321358} implies that the all solutions staring from $\{ \mathcal{V}_{1} \leq 1\}$ converge to zero within a fixed time $\mathcal{T}_2=\dfrac{1}{\mathcal{K}_{1} e^{\frac{-\lambda_2}{2e}}}$.

It follows from \eqref{ch8_eq:dot_v_4} and \eqref{ch8_eq:dot_v_6} that all the solutions starting from $\{ \mathcal{V}_{1} \geq 0 \}$ reach to the origin with in a fixed time $\mathcal{T}(s_{r0}) \leq \mathcal{T}_1 + \mathcal{T}_2 \leq \dfrac{1}{\mathcal{K}_{1}\left(\Upsilon_{2}-1\right)}+\dfrac{1}{\mathcal{K}_{1} e^{\frac{-\lambda_{2}}{2 e}}}$. In other words, the state trajectories starting from an arbitrary initial condition $\mathcal{S}_{1}(0)=s_{r0}$ reach the sliding manifold $\{\mathcal{S}_{1}=0\}$ with the fixed time $\mathcal{T}(s_{r0})$.

Once the trajectories reach the sliding manifold \eqref{eq:sr}, the reduced order dynamics is given by
\begin{equation}
	\dot{r} =- \beta_{1}  r^{\dfrac{\lambda_1 {r}^2}{1+\mu_1 {r}^2}}.
\end{equation}
By following a similar procedure, one can show that $r(t)$ starting from an arbitrary initial condition $r(0)=r_{0}$ converges to zero with in fixed time $\mathcal{T}(r_0) \leq \dfrac{1}{\beta_{1} \left(\Upsilon_{1}-1\right)}+\dfrac{1}{\beta_{1} e^{\frac{-\lambda_{1}}{2 e}}}$. Finally, the closed-loop system \eqref{eq:r_dynamics_stat}, \eqref{eq:sr}, and \eqref{ch8_eq:amr} reaches to the equilibrium point ($0,0$) within a fixed-time satisfying $\mathcal{T}_r =\mathcal{T}(s_{r0}) +\mathcal{T}(r_0)$. In other words, the relative range and its rate go to zero within a fixed time $\mathcal{T}_r$. This completes the proof.
\end{proof}
\begin{remark}
	\Cref{thm:stat_amr} guarantees fixed-time convergence of the relative range dynamics independent of the initial engagement geometry. This property is particularly suitable for cooperative payload delivery, where initial conditions may vary significantly across missions.
\end{remark}

The guidance command presented in \Cref{thm:stat_amr} only ensures that the UAV system will reach the landing platform within a user-defined time. However, it cannot guarantee that it will deliver the payload at the predefined landing angle $\theta_L$ (or reach the landing platform at the specified angle). As discussed in the sequel, the UAV system can deliver the payload to a stationary platform at any arbitrary landing angle, $\theta_L \in [0,2\pi]$. 

It is also important to note that, for a given landing angle $\theta_L$, one can always obtain the desired LOS angle $\theta_d$ using the relation \eqref{eq:thetad}. Since the relation in \eqref{eq:thetad} is bijective in nature, one can regulate the LOS angle $\theta$ to a desired value based on the mapping between the LOS and the landing angle. By doing so, we can ensure that the UAV system always delivers the object to the landing platform at the prespecified angle. Thus, we now design the guidance law $a_{u\theta}$ such that a target set $\mathscr{J}=\{\mathcal{J} \rvert\; \theta=\theta_d,\; \dot{\theta}=0 \; \forall \; t \geq t_2 \}$, where $\mathcal{J}$ is a set of relevant states and $t_2$ is some fixed time instant, is achieved. Toward this objective, we define the LOS angle error as $\varrho \coloneqq \theta - \theta_d$ and choose another sliding manifold as
\begin{equation} \label{ch8_eq:s_theta}
	\mathcal{S}_2 = \dot{\varrho}  + \beta_2 |\varrho|^{\dfrac{\lambda_3 {\varrho}^2}{1+\mu_3 {\varrho}^2}} \sign(\varrho),
\end{equation}
where $\beta_2>0$, $\mu_3>0$, $\lambda_3>0$, and $\dfrac{\lambda_3}{1+\mu_3} > 1$ are design parameters. 
\begin{remark}
	Note that the function $\Gamma: \varrho \to \lvert \varrho \rvert ^{\dfrac{\lambda_3 {\varrho}^2}{1+\mu_3 {\varrho}^2}} = \exp\left({\dfrac{\lambda_3 {\varrho}^2}{1+\mu_3 {\varrho}^2}\ln(\lvert \varrho \rvert) }\right)$ is continuous at $\varrho =0$ with $\Gamma(0)=0$ and thus the second term of sliding surface \eqref{ch8_eq:s_theta} is locally bounded.
\end{remark}
By ensuring sliding mode on this manifold $\{\mathcal{S}_{2}=0\}$, the UAV system can maintain the desired LOS angle from the platform. We now present the dynamics of the sliding manifold through the following lemma.
\begin{lemma}
	The dynamics of the sliding manifold $ \mathcal{S}_{2}$ has a relative degree of one with respect to the UAV's guidance command $a_{u\theta}$.
\end{lemma}
\begin{proof}
We begin by differentiating \eqref{ch8_eq:s_theta} with respect to time, which yields
\begin{align}
	\dot{\mathcal{S}}_2 &= \ddot{\varrho} + \beta_2 \frac{d}{dt} \left( |\varrho|^{\dfrac{\lambda_3 {\varrho}^2}{1+\mu_3 {\varrho}^2}} \sign(\varrho) \right). \label{eq:s_theta_1}
\end{align}
For $\varrho \neq 0$, we can write $\sign(\varrho) = \dfrac{\varrho}{|\varrho|}$. Let the exponent be $g(\varrho) \coloneqq \dfrac{\lambda_3 \varrho^2}{1+\mu_3 \varrho^2}$. The term in brackets of \eqref{eq:s_theta_1} can be written as
\begin{align}
	|\varrho|^{g(\varrho)} \sign(\varrho) = |\varrho|^{g(\varrho)} \frac{\varrho}{|\varrho|} = \varrho |\varrho|^{g(\varrho) - 1} \coloneqq F(\varrho).
\end{align}
Using the chain rule, the time derivative of $F(\varrho)$ is given by $\dfrac{dF(\varrho)}{dt} = \dfrac{dF(\varrho)}{d\varrho} \dot{\varrho}$. We now find $\dfrac{dF}{d\varrho}$ using the product rule on $F(\varrho) = \varrho|\varrho|^{g(\varrho) - 1}$ as
\begin{align}
	\frac{dF(\varrho)}{d\varrho} = |\varrho|^{g(\varrho) - 1} + \varrho \frac{d}{d\varrho}\left(|\varrho|^{g(\varrho) - 1}\right). \label{eq:f_dot}
\end{align}
To find the derivative of the second term in the right-hand side of \eqref{eq:f_dot}, we define $H(\varrho) \coloneqq |\varrho|^{g(\varrho) - 1}$, and use logarithmic differentiation as
\begin{align*}
	\ln|H(\varrho)| &= (g(\varrho) - 1) \ln|\varrho| \implies \frac{1}{H(\varrho)}\frac{dH}{d\varrho} = g'(\varrho)\ln|\varrho| + (g(\varrho) - 1)\frac{1}{\varrho},
\end{align*}
which implies
\begin{align}
	\frac{dH}{d\varrho} &= H(\varrho) \left( g'(\varrho)\ln|\varrho| + \frac{g(\varrho) - 1}{\varrho} \right) = |\varrho|^{g(\varrho) - 1} \left( g'(\varrho)\ln|\varrho| + \frac{g(\varrho) - 1}{\varrho} \right). \label{eq:h_dot}
\end{align}
By substituting the value of  $\dfrac{dH}{d\varrho}$ from \eqref{eq:h_dot}  into \eqref{eq:f_dot}, one may get
\begin{align*}
	\frac{dF}{d\varrho} &= |\varrho|^{g(\varrho) - 1} + \varrho \left[ |\varrho|^{g(\varrho) - 1} \left( g'(\varrho)\ln|\varrho| + \frac{g(\varrho) - 1}{\varrho} \right) \right]\\
	&= |\varrho|^{g(\varrho) - 1} \left[ 1 + \varrho g'(\varrho)\ln|\varrho| + (g(\varrho) - 1) \right]= |\varrho|^{g(\varrho) - 1} \left[ g(\varrho) + \varrho g'(\varrho)\ln|\varrho| \right],
\end{align*}
which on substituting the value of $g(\varrho)$ and using the fact that $g'(\varrho) = \dfrac{2\lambda_3 \varrho}{(1+\mu_3 \varrho^2)^2}$, one gets
\begin{align}
	\frac{dF}{d\varrho} &= |\varrho|^{g(\varrho) - 1} \left[ \frac{\lambda_3 \varrho^2}{1+\mu_3 \varrho^2} + \varrho \left( \frac{2\lambda_3 \varrho}{(1+\mu_3 \varrho^2)^2} \right) \ln|\varrho| \right] = |\varrho|^{g(\varrho) - 1} \left[ \frac{\lambda_3 \varrho^2 (1+\mu_3 \varrho^2)}{(1+\mu_3 \varrho^2)^2} + \frac{2\lambda_3 \varrho^2 \ln|\varrho|}{(1+\mu_3 \varrho^2)^2} \right], \nonumber \\
	&= |\varrho|^{g(\varrho) - 1} \frac{\lambda_3 \varrho^2 (1 + \mu_3 \varrho^2 + 2\ln|\varrho|)}{(1+\mu_3 \varrho^2)^2} =|\varrho|^{g(\varrho) + 1} \frac{\lambda_3 (1 + \mu_3 \varrho^2 + 2\ln|\varrho|)}{(1+\mu_3 \varrho^2)^2}. \label{eq:f_fot_final}
\end{align}
Now, using \eqref{eq:f_fot_final}, the expression in \eqref{eq:s_theta_1} becomes
\begin{equation*}
	\dot{\mathcal{S}}_{2} = \ddot{\varrho} + \beta_2 \dot{\varrho} |\varrho|^{\frac{\lambda_3 \varrho^2}{1+\mu_3 \varrho^2}+1} \frac{\lambda_3 (1 + 2\ln|\varrho| + \mu_3 \varrho^2)}{(1+\mu_3 \varrho^2)^2},
\end{equation*}
which after substituting the value of $\ddot{\varrho}$ from \eqref{eq:los_dynamics_stat} and using the fact that $\ddot{\theta}_d=0$, lead us to arrive at
\begin{equation}
	\dot{\mathcal{S}}_{2}=	-\dfrac{2\dot{r}\dot{\theta}}{r} - \dfrac{a_{u  \theta}}{r} +  \frac{\lambda_3\beta_2 |\varrho|^{\frac{\lambda_3 \varrho^2}{1+\mu_3 \varrho^2}+1} (1 + 2\ln|\varrho| + \mu_3 \varrho^2)}{(1+\mu_3 \varrho^2)^2} \dot{\varrho} =-\dfrac{2\dot{r}\dot{\theta}}{r} - \dfrac{a_{u  \theta}}{r} + \mathcal{F}(\varrho) \dot{\varrho}, \label{eq:s_theta_dynamics_final}
\end{equation}
where
\begin{equation*}
	\mathcal{F}(\varrho) =  \dfrac{ \beta_2\lambda_3 |\varrho|}{1+\mu_3 {\varrho}^2}|\varrho|^{\dfrac{\lambda_3 {\varrho}^2}{1+\mu_3 {\varrho}^2}}\left(\dfrac{2 \ln(|\varrho|)}{1+\mu_3 {\varrho}^2} + 1\right).
\end{equation*}
It immediately follows from \eqref{eq:s_theta_dynamics_final} that the dynamics of the sliding manifold has a relative degree of one with respect to the UAV's control input $a_{u  \theta}$.
\end{proof}
We now present the guidance command to achieve the desired LOS angle in the following theorem.
\begin{theorem}\label{thm:stat_amtheta}
Consider the LOS angle dynamics given in \eqref{eq:los_dynamics} and the sliding manifold chosen as \eqref{ch8_eq:s_theta}. If the UAV's guidance command, $a_{u\theta}$, is designed as
\begin{equation} \label{eq:am_u_theta}
	a_{u\theta} = -{2\dot{r}\dot{\theta}} + r\mathcal{F}(\varrho)\dot{\varrho} +\mathcal{K}_{2} r \sign(\mathcal{S}_{2})|\mathcal{S}_{2}|^{\dfrac{\lambda_4 {\mathcal{S}_2}^2}{1+\mu_4 {\mathcal{S}_2}^2}} ,
\end{equation}
with $\mathcal{K}_{2}>0$, $\lambda_4>0$, $\mu_4>0$ such that $\Upsilon_4 \coloneqq \dfrac{\lambda_4}{1+\mu_4} > 1$, then, the UAV  converges to the desired LOS angle ($\theta_d$) within a fixed time, upper bounded by
\begin{equation} \label{ch8_eq:T_theta}
	\mathcal{T}_{\theta} \leq  \dfrac{1}{\mathcal{K}_{2}\left(\Upsilon_4-1\right)}+\dfrac{1}{\mathcal{K}_{2} \mathcal{M}_{4} }+\dfrac{1}{\beta_2\left(\Upsilon_3-1\right)}+ \dfrac{1}{\beta_2 \mathcal{M}_{3}}; \; \Upsilon_3 \coloneqq \dfrac{\lambda_3}{1+\mu_3} > 1. 
\end{equation}
where $\mathcal{M}_3 = e^{\frac{-\lambda_{3}}{2 e}}$  and  $\mathcal{M}_4 = e^{\frac{-\lambda_{4}}{2 e}}$ are constants, regardless of initial UAV-platform engagement.
\end{theorem}
\begin{proof}
	Consider another positive definite, radially unbounded Lyapunov function candidate as $\mathcal{V}_{2}=\mathcal{S}_{2}^2$, which, on differentiating with respect to time along the state trajectories and using \eqref{eq:s_theta_dynamics_final}, yields
	\begin{equation}\label{ch8_eq:dot_w_1}
		\dot{\mathcal{V}}_{2} =  2\mathcal{S}_{2} \dot{\mathcal{S}}_{2}=  2\mathcal{S}_{2}  \left( \ddot{\varrho}+ \dfrac{\beta_2 \lambda_3 |\varrho|^{\dfrac{\lambda_3 \varrho^{2}}{\mu_{3} \varrho^{2}+1}+1}\left(2 \ln (|\varrho|)+\mu_{3} \varrho^{2}+1\right)}{\left(\mu_{3} \varrho^{2}+1\right)^{2}}\dot{\varrho}\right). 
	\end{equation}
	By substituting the value of $\ddot{\varrho}$ from \eqref{eq:los_dynamics} into \eqref{ch8_eq:dot_w_1} and using the fact that $a_{m\theta}^{p}=0$, one may obtain
	\begin{equation*}
		\dot{\mathcal{V}}_{2} =  2\mathcal{S}_{2}  \left(-\dfrac{2\dot{r}\dot{\theta}}{r} - \dfrac{a_{u  \theta}}{r}+ \dfrac{\beta_2 \lambda_3 |\varrho|^{\dfrac{\lambda_3 \varrho^{2}}{\mu_{3} \varrho^{2}+1}+1}\left(2 \ln (|\varrho|)+\mu_{3} \varrho^{2}+1\right)}{\left(\mu_{3} \varrho^{2}+1\right)^{2}}\dot{\varrho}\right),
	\end{equation*}
	which, upon substituting the guidance command given in \eqref{eq:am_u_theta} leads us to arrive at
	\begin{equation}
		\dot{\mathcal{V}}_{2} \leq -2\mathcal{K}_{2}|\mathcal{S}_{2}|^{\dfrac{\lambda_{4} \mathcal{S}_{2}^{2}}{1+\mu_{4} \mathcal{S}_{2}^{2}}+1}.
	\end{equation}
	By following a similar procedure to \Cref{thm:stat_amr}, one can show that the UAV will attain the desired LOS angle, $\theta_d$, within a fixed time $\mathcal{T}_{\theta}$. This completes the proof.
\end{proof}

\subsection{Payload Delivery on a Maneuvering Platform}
In this subsection, we design the guidance strategies for the UAV system to deliver a payload to a maneuvering platform. As discussed earlier, the only condition for landing or delivering an object to a maneuvering platform is $\theta_{L} = 0$. To successfully deliver the payload to the maneuvering platform, the UAV's control inputs must ensure that the relative range and its rate become zero, and the LOS angle does not rotate. Specifically, this means $r =\dot{r} = 0$ and $\dot{\theta} = 0$. 

To enforce $r=0$ and $\dot{r}=0$, we again adopt the same sliding manifold for the relative-range dynamics as in the stationary-platform case, as defined in \eqref{eq:sr}. However, the guidance command derived for the stationary case is not applicable to maneuvering platforms. This is because, during the derivation of the stationary-platform guidance law, the platform velocity $v_p$ and the platform radial acceleration $a_{mr}^{p}$ were assumed to be zero. 

The relative-range dynamics for the UAV-platform system is given by
\begin{equation} \label{eq:r_dyn_maneuver}
	\ddot{r}= r \dot{\theta}^2 + \left( a_{p r} - a_{u r} \right),
\end{equation}
where $a_{p r}$ denotes the acceleration of the landing platform. In general, the platform maneuver may not be known a priori and is therefore treated as an external disturbance acting on the system. However, since practical platforms are driven by the actuators, we assume that their maneuvering capability is bounded by a known positive constant, that is, $\lvert a_{p r} \rvert \leq a_{p r}^{\max}$. In the following theorem, we propose a guidance command specifically tailored for delivering a payload to maneuvering platforms using a two-UAV system.
\begin{theorem} \label{thm:man_amr}
	Consider the relative range dynamics as in \eqref{eq:r_dynamics} and the sliding manifold as in \eqref{eq:sr}. If we design the guidance command $a_{ur}$ as
	\begin{equation} \label{ch8_eq:amr_prime}
		a_{ur} = r \dot{\theta}^2 +\mathcal{K}_{3} | \mathcal{S}_{1}|^{\dfrac{\lambda_2 \mathcal{S}_{1}^2}{1+\mu_2 \mathcal{S}_{1}^2}} \sign( \mathcal{S}_{1})+  \mathcal{F}(r)\dot{r} , 
	\end{equation}
	with $\mathcal{K}_{3} > e^{\tfrac{\lambda_2}{2e}}  a_{pr}^{\max}$, then the UAV rendezvous with the landing platform within a fixed time $\mathcal{T}_{r}^{\prime}$ given by
	\begin{align}
		\mathcal{T}_{r}^{\prime} \leq  \dfrac{1}{\mathcal{K}_{3}\left(\Upsilon_2-1\right)}+\dfrac{1}{\mathcal{K}_{3} \mathcal{M}_{2}}+\dfrac{1}{\beta_{1}\left(\Upsilon_1-1\right)}+ \dfrac{1}{\beta_{1} \mathcal{M}_{1}} \label{eq:Tr_prime},
	\end{align}
	where $\mathcal{F}(r)$ and other constants are the same as defined in \eqref{eq:fr}, regardless of initial relative separation $r(0)$.
\end{theorem}
\begin{proof}
	Consider the Lyapunov function candidate $\mathcal{V}_{1} = \mathcal{S}_{1}^2$. The time derivative of $\mathcal{V}_{1}$ along the system trajectories is $\dot{\mathcal{V}}_{1} = 2 \mathcal{S}_{1} \dot{\mathcal{S}}_{1}$. Using the value of $\dot{\mathcal{S}}_{1}$ from \eqref{ch8_eq:dot_sr} while accounting for the platform maneuver, we obtain
	\begin{align}
		\dot{\mathcal{V}}_1 =   2\mathcal{S}_{1} \left(r \dot{\theta}^2 +a_{p r} - a_{u r} + \frac{\beta_{1} \lambda_1  r^{{\dfrac{\lambda_1 r^{2}}{\mu_{1} r^{2}+1}}+1}\left(2 \ln ( r)+\mu_{1} r^{2}+1\right)}{\left(\mu_{1} r^{2}+1\right)^{2}}\dot{r}\right). \label{ch8_eq:dot_v_2_prime}
	\end{align}
	By substituting the value of $a_{u r}$ from \eqref{ch8_eq:amr_prime} into \eqref{ch8_eq:dot_v_2_prime} 
	\begin{align*}
		\dot{\mathcal{V}}_{1} =   2 \mathcal{S}_{1} \left[r \dot{\theta}^2 +a_{p r} - \left(r \dot{\theta}^2 +\mathcal{K}_{3} | \mathcal{S}_{1}|^{\frac{\lambda_2 \mathcal{S}_{1}^2}{1+\mu_2 \mathcal{S}_{1}^2}} \sign( \mathcal{S}_{1})+  \mathcal{F}(r)\dot{r}\right)  + \dfrac{\beta_{1} \lambda_1  r^{{\frac{\lambda_1 r^{2}}{\mu_{1} r^{2}+1}}+1}\left(2 \ln ( r)+\mu_{1} r^{2}+1\right)}{\left(\mu_{1} r^{2}+1\right)^{2}}\dot{r}\right], 
	\end{align*}
	which, after some algebraic simplifications and using the fact that the platform maneuver is bounded by $\lvert a_{mr}^{p} \rvert \leq a_{pr}^{\max}$, results in
	\begin{align}
		\dot{\mathcal{V}}_{1} =   2 \mathcal{S}_{1} \left( -\mathcal{K}_{3}| \mathcal{S}_{1}|^{\dfrac{\lambda_2 \mathcal{S}_{1}^2}{1+\mu_2 \mathcal{S}_{1}^2}} \sign( \mathcal{S}_{1}) +a_{p r} \right) = -2\mathcal{K}_{3}| \mathcal{S}_{1}|^{\dfrac{\lambda_2 \mathcal{S}_{1}^2}{1+\mu_2 \mathcal{S}_{1}^2}+1} \sign( \mathcal{S}_{1}) + 2 a_{p r} \mathcal{S}_{1}\label{ch8_eq:dot_v_3_1}
	\end{align}
	For the case, when $\mathcal{V}_{1} \geq 1$, we have 
	$\frac{\lambda_2 \mathcal{S}_{1}^2}{1+\mu_2 \mathcal{S}_{1}^2}+1 \geq \frac{\lambda_2}{1+\mu_2}+1>2$ and  $\mathcal{S}_{1} \geq 1$. Since $\mathcal{S}_{1} \geq 1$ and $\Upsilon_{2}=\frac{\lambda_2}{1+\mu_2}>1$, implies
	\begin{equation}\label{ch8_eq:dot_v_3_2}
		\dot{\mathcal{V}}_{1} \leq -2 |\mathcal{S}_{1}| \left( \mathcal{K}_{3}| \mathcal{S}_{1}|^{\dfrac{\lambda_2 \mathcal{S}_{1}^2}{1+\mu_2 \mathcal{S}_{1}^2}}  - a_{p r}^{\max} \right) \leq - 2 \left(\mathcal{K}_{3}- a_{p r}^{\max} \right)  \lvert \mathcal{S}_{1}\rvert^{\Upsilon_{2}+1} \leq -2\left(\mathcal{K}_{3}- a_{p r}^{\max} \right) \mathcal{V}_{1}^{\dfrac{\Upsilon_{2}+1}{2}}. 
	\end{equation}
	As $\mathcal{K}_{3}> a_{p r}^{\max}$, again invoking \cite[Lemma 1]{6104367} implies that all the solutions starting from $\{ \mathcal{V}_{1} \geq 1\}$ reach to the set $\{ \mathcal{V}_{1} \leq 1\}$ within a fixed time $\mathcal{T}_{1}^{\prime} \leq \dfrac{1}{\mathcal{K}_{3}\left(\Upsilon_{2}-1\right)}$.
	
	For the case when $\mathcal{V}_{1} \leq  1$, we have $\mathcal{S}_{1} \leq 1$, so the expression of $\dot{\mathcal{V}}_{1}$ given in \eqref{ch8_eq:dot_v_3_1} can be written as
	\begin{equation}\label{ch8_eq:dot_v_3_3}
		\dot{\mathcal{V}}_{1}  = - 2 \mathcal{K}_{3} \lvert \mathcal{S}_{1}\rvert \lvert \mathcal{S}_{1}\rvert^{\dfrac{\lambda_2 \mathcal{S}_{1}^2}{1+\mu_2 \mathcal{S}_{1}^2}} + 2 a_{p r} \mathcal{S}_{1}.
	\end{equation}
	Again notice that $1+\mu_2 \mathcal{S}_{1}^2 \geq 1$, which implies $\min\left( \lvert \mathcal{S}_{1}\rvert^{\frac{\lambda_2 \mathcal{S}_{1}^2}{1+\mu_2 \mathcal{S}_{1}^2}} \right) \geq \min \left(  \lvert \mathcal{S}_{1} \rvert^{\lambda_2  \mathcal{S}_{1} ^2} \right) = e^{\frac{-\lambda_2}{2e}}$, and \eqref{ch8_eq:dot_v_3_3} can be written as 
	\begin{equation}\label{ch8_eq:dot_v_3_4}
		\dot{\mathcal{V}}_{1} \leq  - 2\left(\mathcal{K}_{3}e^{\frac{-\lambda_2}{2e}} -a_{p r}^{\max} \right) \lvert \mathcal{S}_{1}\rvert  \leq -2\left(\mathcal{K}_{3}e^{\frac{-\lambda_2}{2e}} -a_{p r}^{\max} \right) \sqrt{\mathcal{V}_{1}}.
	\end{equation}
	As $\mathcal{K}_{3}e^{\frac{-\lambda_2}{2e}} -a_{p r}^{\max} >0 $, again using the results from \cite[Theorem 4]{doi:10.1137/S0363012997321358} implies that the all solutions staring from $\{ \mathcal{V}_{1} \leq 1\}$ converge to zero within a fixed time $\mathcal{T}_{2}^{\prime}=\dfrac{1}{\mathcal{K}_{3} e^{\frac{-\lambda_2}{2e}}}$. Thus, it follows from \eqref{ch8_eq:dot_v_3_2} and \eqref{ch8_eq:dot_v_3_4}  that all the solutions starting from $\{ \mathcal{V}_{1} \geq 0 \}$ reach to the sliding manifold $\mathcal{S}_{1} =0$ with in a fixed time $\mathcal{T}^{\prime}(s_{0}) \leq \mathcal{T}_1^{\prime} + \mathcal{T}_2^{\prime} \leq \dfrac{1}{\mathcal{K}_{3}\left(\Upsilon_{2}-1\right)}+\dfrac{1}{\mathcal{K}_{3} e^{\frac{-\lambda_{2}}{2 e}}}$ and $\mathcal{S}_{1} \equiv 0$ $\forall$ $t \geq \mathcal{T}^{\prime}(s_{0})$.
	
	Once the system trajectories reach the sliding manifold \eqref{eq:sr}, the reduced-order dynamics can be expressed as
	\begin{equation}
		\dot{r} = - \beta_{1} r^{\dfrac{\lambda_1 r^2}{1+\mu_1 r^2}}.
	\end{equation}
	Following a procedure analogous to the previous analysis, it can be shown that the state $r(t)$, starting from an arbitrary initial condition $r(0)=r_0$, converges to the origin within a fixed time bounded by $\mathcal{T}(r_0) \leq \frac{1}{\beta_{1} \left(\Upsilon_{1}-1\right)} + \frac{1}{\beta_{1} e^{-\frac{\lambda_{1}}{2 e}}}$.
	Consequently, the closed-loop system described by \eqref{eq:r_dyn_maneuver}, \eqref{eq:sr}, and \eqref{ch8_eq:amr_prime} reaches the equilibrium point $(0,0)$ in fixed time. The overall convergence time is given by $\mathcal{T}_{r}^{\prime} = \mathcal{T}^{\prime}(s_0) + \mathcal{T}(r_0)$, 
	where $\mathcal{T}^{\prime}(s_0)$ denotes the reaching time to the sliding manifold and $\mathcal{T}(r_0)$ corresponds to the convergence time along the manifold. Therefore, both the relative range and its rate converge to zero within the prescribed fixed time $\mathcal{T}_r^{\prime} $. Consequently, the UAV lands on the manuring platform with zero radial velocity. This completes the proof.
\end{proof}

Note that we choose the same sliding manifold for both maneuvering and stationary platforms to nullify the relative separation between the equivalent agent and the landing platform. This choice is motivated by the fact that the objective of driving the relative range and its rate to zero remains the same regardless of platform movements. We, however, choose a different sliding manifold for the LOS angle, defined as $\mathcal{S}_{3} = \dot{\theta}$. This distinction arises because the UAV system can only deliver the payload with a landing angle of zero. Therefore, unlike the previous case, there is no requirement to regulate the LOS angle to a specific value. Instead, it suffices to ensure that the LOS does not rotate, that is, $\dot{\theta} = 0$. In other words, the LOS angle must remain constant.  The dynamics of the LOS angle is given by \eqref{eq:los_dynamics} as
\begin{equation*}
	\ddot{\theta}=-\dfrac{2\dot{r}\dot{\theta}}{r} + \dfrac{1}{r} \left( a_{p \theta}-a_{u  \theta} \right),
\end{equation*}
where the term $ a_{p \theta}$ denotes the acceleration of the maneuvering platform, which is assumed to be upper bounded by a known positive constant, that is, $ \lvert a_{p \theta} \rvert \leq  a_{p \theta}^{\max} $. Next, our objective is to design the UAV's guidance command, $a_{u  \theta}$, such that a target set $\mathscr{D}=\{\mathcal{D} \rvert\; \dot{\theta}=0 \; \forall \; t \geq t_3 \}$, where $\mathcal{J}$ is a set of relevant states and $t_3$ is some fixed time instant, is achieved. The essence of the design is presented in the following theorem.
\begin{theorem}\label{thm:man_amtheta}
	Consider the LOS angle dynamics given in \eqref{eq:los_dynamics} and the sliding manifold chosen as $\mathcal{S}_{3} = \dot{\theta}$. If the guidance command $a_{u\theta}$ is designed as
	\begin{align}
		a_{u\theta} = - 2\dot{r}\dot{\theta}
		+ \mathcal{K}_4 r |\dot{\theta}|^{\dfrac{\lambda_5 \dot{\theta}^2}{1+\mu_5 \dot{\theta}^2}}
		\sign(\dot{\theta}), \label{ch8_eq:am_theta_prime}
	\end{align}
	where $\mathcal{K}_{4}>\sup_{t\geq 0}\left(\frac{a_{p\theta}^{\max}}{r}\right) e^{\frac{\lambda_5}{2e}}$ and $\lambda_5,\mu_5>0$ are chosen such that
	$\Upsilon_5 \coloneqq \dfrac{\lambda_5}{1+\mu_5} > 1$, then the LOS joining the equivalent agent and the landing platform becomes non-maneuvering, i.e., $\dot{\theta}=0$, within a fixed time bounded by
	\begin{align}
		\mathcal{T}_{\theta}^{\prime}
		\leq
		\dfrac{1}{\mathcal{K}_{4}(\Upsilon_5-1)}
		+
		\dfrac{1}{\mathcal{K}_{4} e^{-\frac{\lambda_{5}}{2 e}}}. \label{ch8_eq:T_theta_prime}
	\end{align}
\end{theorem}
\begin{proof}
	Consider a Lyapunov function candidate as $\mathcal{V}_{3}=\dot{\theta}^2$. The time derivative of the Lyapunov function candidate along the state trajectories can be obtained as $\dot{\mathcal{V}}_{3}=  2 \dot{\theta} \ddot{\theta}$, which on substituting $\ddot{\theta}$ from \eqref{eq:los_dynamics}, results in
	\begin{align}
		\dot{\mathcal{V}}_{3} =  2\dot{\theta} \left[-\dfrac{2\dot{r}\dot{\theta}}{r} + \dfrac{1}{r} \left( a_{p \theta}-a_{u  \theta} \right) \right]. \label{ch8_eq:w_dot_prime_1}
	\end{align}
	With the guidance command given in \eqref{ch8_eq:am_theta_prime}, the expression in \eqref{ch8_eq:w_dot_prime_1} becomes
	\begin{align*}
		\dot{\mathcal{V}}_{3} =  2\dot{\theta} \left[ -\dfrac{2\dot{r}\dot{\theta}}{r}+ \dfrac{a_{p \theta}}{r} - \dfrac{1}{r} \left(-2\dot{r}\dot{\theta} +\mathcal{K}_4 r |\dot{\theta}|^{\dfrac{\lambda_5 {\dot{\theta}}^2}{1+\mu_5 {\dot{\theta}}^2}} \sign(\dot{\theta})\right)   \right] ,   
	\end{align*} 
	which, after some algebraic simplifications, results in
	\begin{align}
		\dot{\mathcal{V}}_{3} = -2\mathcal{K}_4 \dot{\theta}  |\theta|^{\dfrac{\lambda_5 {\theta}^2}{1+\mu_5 {\theta^2}}} \sign(\theta) + \dfrac{a_{p \theta} \dot{\theta}}{r}. \label{ch8_eq:w_dot_prime_2}
	\end{align}
	
	Now consider the case when $\mathcal{V}_{3}=\dot{\theta}^2 \geq 1$. This implies $\dfrac{\lambda_5 {\dot{\theta}}^2}{1+\mu_5 {\dot{\theta}^2}}+1\geq \dfrac{\lambda_5}{1+\mu_5}+1>2$. Since $\lvert \dot{\theta} \rvert \geq 1$ and $\dfrac{\lambda_5}{1+\mu_5} > 1$, one may write \eqref{ch8_eq:w_dot_prime_2} as
	\begin{align}
		\dot{\mathcal{V}}_{3} =  -2 \left(\mathcal{K}_{4}  - \dfrac{a_{p\theta}^{\max}}{r}\right) \lvert \dot{\theta} \rvert^{\Upsilon_{5}+1} \leq -2 \left(\mathcal{K}_{4}  - \dfrac{a_{p\theta}^{\max}}{r}\right) \mathcal{V}_{3}^{\dfrac{\Upsilon_{5}+1}{2}}. \label{ch8_eq:w_dot_prime_5}
	\end{align}
	For the case when $\mathcal{V}_{3}=\dot{\theta}^2 \leq 1$, one may write \eqref{ch8_eq:w_dot_prime_2} as
	\begin{align}
		\dot{\mathcal{V}}_{3} =    - 2 \mathcal{K}_{4} |\dot{\theta}|  |\theta|^{\dfrac{\lambda_5 {\dot{\theta}}^2}{1+\mu_5 {\dot{\theta}^2}}} + \dfrac{a_{p \theta} \dot{\theta}}{r}. \label{ch8_eq:w_dot_prime_6}
	\end{align}
	Also, for $\mathcal{V}_{3}\leq 1$, $1+\mu_5\dot{\theta}^2 \geq 1$ and $\lvert \theta \rvert \leq 1$, which together implies $\min \left(|\dot{\theta}|^{\dfrac{\lambda_5 {\dot{\theta}}^2} {1+\mu_5 {\dot{\theta}^2}}}\right) \geq $ $\min \left( \lvert \dot{\theta} \rvert^{\lambda_5 \dot{\theta}^2}\right) = e^{\dfrac{-\lambda_5}{2e}}$. Using this relation, we write the expression in \eqref{ch8_eq:w_dot_prime_6} as
	\begin{align}
		\dot{\mathcal{V}}_{3} \leq  -2 \left( \mathcal{K}_4  e^{\frac{-\lambda_5}{2e}} - \dfrac{a_{p\theta}^{\max}}{r} \right) |\theta| \leq  -2 \left( \mathcal{K}_4  e^{\frac{-\lambda_5}{2e}} - \dfrac{a_{p\theta}^{\max}}{r} \right) \mathcal{V}_3^{\frac{1}{2}} . \label{ch8_eq:w_dot_prime_7}
	\end{align}
	Since $\mathcal{K}_{4}>\sup_{t\geq 0}\left(\frac{a_{p\theta}^{\max}}{r}\right) e^{\frac{\lambda_5}{2e}}$, by a similar line of arguments as in \Cref{thm:man_amr}, it follows from \eqref{ch8_eq:w_dot_prime_5} and \eqref{ch8_eq:w_dot_prime_7} that the LOS rate $\dot{\theta}$ nullifies to zero within a fixed time $\mathcal{T}_{\theta}^{\prime}$.
\end{proof}

It is important to note that the guidance laws presented in \Cref{thm:stat_amr,thm:stat_amtheta,thm:man_amr,thm:man_amtheta} are in the transformed coordinate system and the mapping is given in \eqref{ch8_eq:am_transformation}. As this is a bijective mapping, one can calculate the components of control inputs required by the equivalent agent $U$ by inverting \eqref{ch8_eq:am_transformation}.

\subsection{Link Orientation Control}\label{sec:link_orientation}
As discussed earlier, the link can rotate about the $z$-axis. Therefore, in addition to achieving accurate payload delivery, it is essential to regulate the orientation of the rigid link connecting the UAVs. Uncontrolled link rotation during payload delivery may lead to misalignment and degrade accuracy. Thus, we design a control law to stabilize the link orientation to a desired value within a fixed time. Let $\xi \in \mathbb{R}$ denote the angular orientation of the rigid link about the $z$-axis, measured with respect to an inertial frame. The rotational dynamics of the link are modeled as
\begin{equation} \label{eq:orinet_dyan}
	\dot{\xi} = \omega, \; \dot{\omega} = \alpha,
\end{equation}
where $\omega$ and $\alpha$ denote the angular velocity and angular acceleration (control input), respectively. We define the orientation error as $e_{\xi} \coloneqq \xi - \xi_d$, where $\xi_d$ is the desired orientation. To ensure fixed-time convergence, we introduce the following sliding manifold
\begin{align}\label{eq:s_xi}
	\mathcal{S}_{\xi} = \omega + \beta_3 |e_{\xi}|^{\frac{\lambda_6 e_{\xi}^2}{1+\mu_6 e_{\xi}^2}} \sign(e_{\xi}),
\end{align}
where $\beta_3>0$, $\lambda_6>0$, $\mu_6>0$ are design parameters satisfying $\Upsilon_6 \coloneqq \frac{\lambda_6}{1+\mu_6} > 1$. One can obtain the derivative of the sliding manifold $\mathcal{S}_{\xi}$ by following a similar procedure as in \Cref{lem:sr} as
\begin{align}\label{eq:s_xi_dot}
	\dot{\mathcal{S}}_{\xi} = \alpha + \mathcal{F}(e_{\xi}) \dot{e}_{\xi},
\end{align}
where $\mathcal{F}(e_{\xi})$ is a bounded nonlinear function, given by
\begin{equation} \label{eq:f_e_xi}
	\mathcal{F}(e_{\xi}) =
	\dfrac{\beta_{3} \lambda_3 
		\left( 2\ln|e_{\xi}| + \mu_{6} e_{\xi}^{2} + 1 \right)}
	{\left(1+\mu_{6} e_{\xi}^{2}\right)^{2}} e_{\xi}^{\left(\dfrac{\lambda_6 e_{\xi}^{2}}{1+\mu_{6} e_{\xi}^{2}}\right)+1}.
\end{equation}

\begin{theorem}
	Consider the link orientation dynamics given by \eqref{eq:orinet_dyan} and the sliding manifold chosen as \eqref{eq:s_xi}. If the control input is designed as
	\begin{equation}\label{eq:alpha}
		\alpha = - \mathcal{F}(e_{\xi}) \omega 
		- \mathcal{K}_5 |\mathcal{S}_{\xi}|^{\frac{\lambda_7 \mathcal{S}_{\xi}^2}{1+\mu_7 \mathcal{S}_{\xi}^2}} \sign(\mathcal{S}_{\xi}),
	\end{equation}
	where $\mathcal{K}_5>0$, $\lambda_7>0$, $\mu_7>0$ satisfy $\Upsilon_7 \coloneqq \frac{\lambda_7}{1+\mu_7} > 1$, then the link orientation angle $\xi$ will converge to its desired value $\xi_{d}$ within a fixed time, upper bounded by
	\begin{equation}
		\mathcal{T}_{\xi} \leq  
		\dfrac{1}{\mathcal{K}_5(\Upsilon_7-1)} 
		+ \dfrac{1}{\mathcal{K}_5 \mathcal{M}_5}
		+ \dfrac{1}{\beta_3(\Upsilon_6-1)} 
		+ \dfrac{1}{\beta_3 \mathcal{M}_6},
	\end{equation}
	where $\mathcal{M}_5=e^{-\frac{\lambda_7}{2e}}$ and $\mathcal{M}_6=e^{-\frac{\lambda_6}{2e}}$.
\end{theorem}
\begin{proof}
	Consider the Lyapunov function $\mathcal{V}_{\xi} = \mathcal{S}_{\xi}^2$. The time derivative of $\mathcal{V}_{\xi}$ can be obtained using \eqref{eq:s_xi_dot} as
	\begin{equation}
		\dot{\mathcal{V}}_{\xi} = 2 \mathcal{S}_{\xi} \dot{\mathcal{S}}_{\xi} =  2 \mathcal{S}_{\xi} \left( \omega + \beta_3 |e_{\xi}|^{\frac{\lambda_6 e_{\xi}^2}{1+\mu_6 e_{\xi}^2}} \sign(e_{\xi}) \right),
	\end{equation}
	which on, substituting the value of control input $\alpha$ from \eqref{eq:alpha} results in
	\begin{equation}
		\dot{\mathcal{V}}_{\xi}
		= -2\mathcal{K}_5 |\mathcal{S}_{\xi}|^{\frac{\lambda_7 \mathcal{S}_{\xi}^2}{1+\mu_7 \mathcal{S}_{\xi}^2}+1}.
	\end{equation}
	Using arguments similar to those in \Cref{thm:stat_amr}, one can show the error $e_{\xi}$ converges to zero within a fixed-time $\mathcal{T}_{\xi}$. Consequently, the link will be in the desired direction within the same time $\mathcal{T}_{\xi}$. This completes the proof.
\end{proof}

\section{Performance Evaluations} \label{ch8_sec:load_simulation}
In this section, we evaluate the efficacy of the proposed guidance strategies (\Cref{thm:stat_amr,thm:stat_amtheta,thm:man_amr,thm:man_amtheta}) for scenarios where the UAV system is required to deliver a payload to both stationary and maneuvering platforms. To respect the physical limitations of the actuators, we adopt the strategy used in \cite{doi:10.2514/1.G005098}, with $k_M$, $k_\alpha$, and $k_\omega$ chosen as $0.8$, $0.15$, and $0.045$, respectively. The maximum allowable speed for each vehicle is constrained to $25 \ \si{ m/s}$, and the maximum acceleration in each direction is chosen as $25 \ \si{m/s^2}$. The length of the rigid payload is $L = 2\ \si{m}$, and the parameter $\sigma$ is set to $0.5$, meaning the equivalent agent is located at an equal distance from both UAVs $A$ and $B$. The desired link orientation angle is specified as $90^\circ$, with link controller gains chosen as $\beta_3=0.2$, $\lambda_6=3.5$, $\mu_6=2.4$, $\lambda_7=2.1$,  $\mu_7=1.0$, $\mathcal{K}_{5}=2.5$. For better visualization, animations for our results are available at \href{https://www.youtube.com/playlist?list=PLGPUUU7qKUSA}{https://www.youtube.com/playlist?list=PLGPUUU7qKUSA}.
\subsection{Payload delivery on a stationary platform}
We first evaluate the performance of the proposed guidance strategy (\Cref{thm:stat_amr,thm:stat_amtheta})  when delivering the payload to a stationary platform from various approach directions. The landing platform is located at the origin ($0,0\ \si{m}$). On the other hand, the equivalent agent is located at a radial distance of $100\ \si{m}$ and at a LOS angle of $120^\circ$ with respect to the platform, with an initial heading angle of $135^\circ$. The desired landing angle ($\theta_{L}$) is chosen form the set $\{-45^{\circ},\ -90^{\circ},\ -135^{\circ},\ -180^{\circ}\}$. The design parameters are set to $\beta_1=0.1$, $\lambda_1=3.5$, $\mu_1=2.4$, $\lambda_2=2.1$,  $\mu_2=1.0$, $\mathcal{K}_{1}=0.5$, $\mathcal{K}_2=5.0$.

Under the proposed guidance laws (presented in \Cref{thm:stat_amr,thm:stat_amtheta}), we show the performance through \Cref{fig:stat}. One may notice from \Cref{fig:stat_path} that the UAV delivers the payload at the desired landing angles for all four cases. It can be observed from \Cref{fig:stat_amu} that the control input demand is relatively higher during the phase when sliding mode is not imposed on the sliding manifolds $\mathcal{S}_{i}$ for $i=1,2$. However, once sliding mode is enforced on the chosen manifolds, the control demand reduces significantly and ultimately converges to zero (see \Cref{fig:stat_surface}). The individual accelerations of UAV $A$ and UAV $B$ are depicted in \Cref{fig:stat_ama,fig:stat_amb}. \Cref{fig:stat_r_psi_u} demonstrates that shortly after the sliding mode occurs, the heading angle converges to its desired value ($-\theta_{L}$) and the relative range successfully nullifies to zero. Consequently, the UAV delivers the payload to the stationary platform from the desired direction. \Cref{fig:stat_vu} shows that the equivalent agent initially increases its linear speed to nullify the distance to the platform, before smoothly decelerating to rest precisely at the moment of delivery. The corresponding linear velocities allocated to the individual UAVs are shown in \Cref{fig:stat_va,fig:stat_vb}, and a similar behavior can be observed. The performance of the link controller is shown in \Cref{fig:Stat_Link_States}. Finally, the performance of the link orientation controller is shown in \Cref{fig:Stat_Link_States}. It can be observed that the angle $\zeta$ converges to the desired orientation of the link, which is $90^\circ$, within a fixed time for all cases. Note that the link orientation is set to be perpendicular to the heading angle. However, it can be chosen based on mission requirements.
\begin{figure}[!ht]
	\centering
	\begin{subfigure}{0.33\linewidth}
		\centering
		\includegraphics[width=\linewidth]{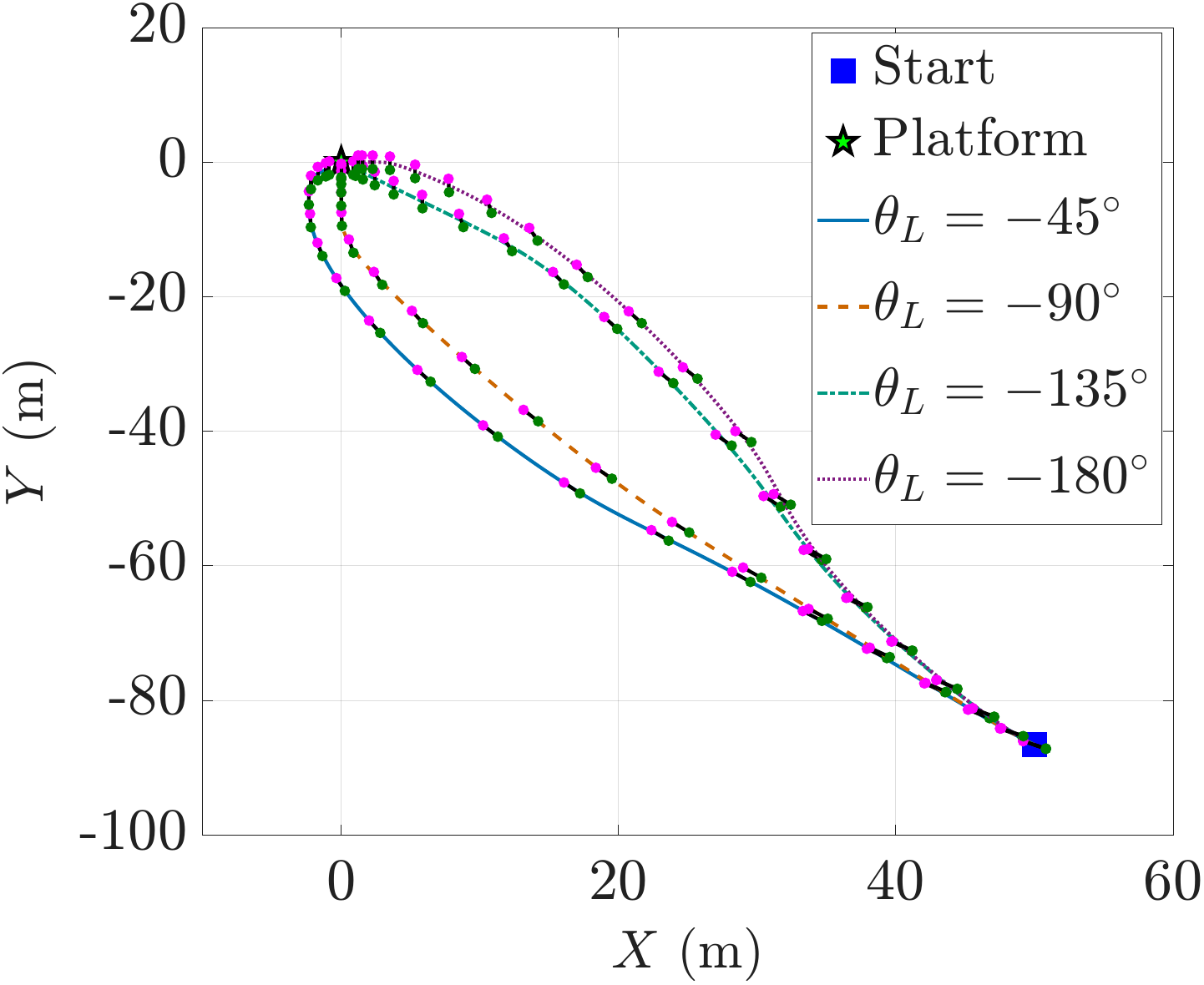}
		\caption{Trajectory.}
		\label{fig:stat_path}
	\end{subfigure}%
	\begin{subfigure}{0.33\linewidth}
		\centering
		\includegraphics[width=\linewidth]{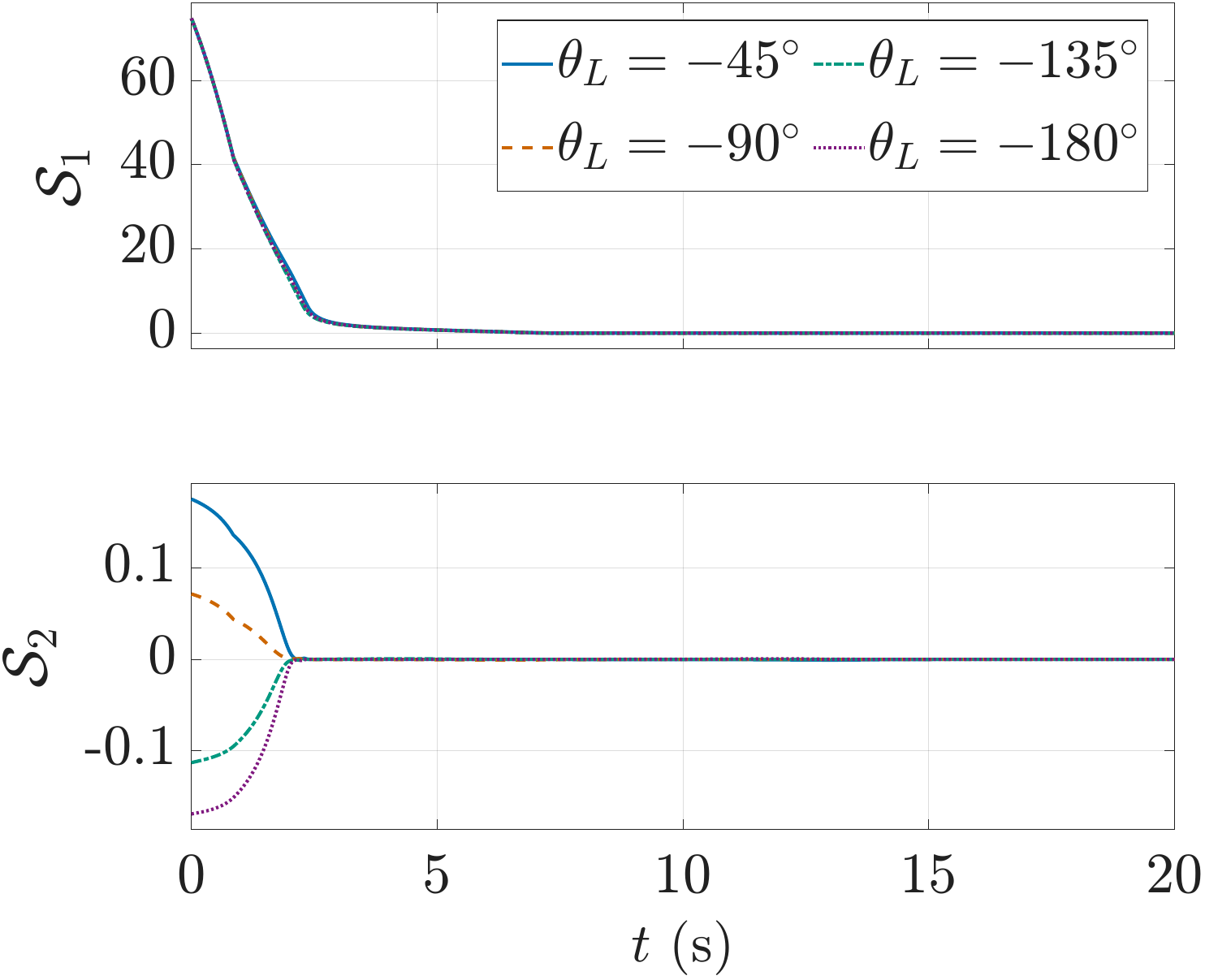}
		\caption{Sliding surfaces.}
		\label{fig:stat_surface}
	\end{subfigure}%
	\begin{subfigure}{0.33\linewidth}
		\centering
		\includegraphics[width=\linewidth]{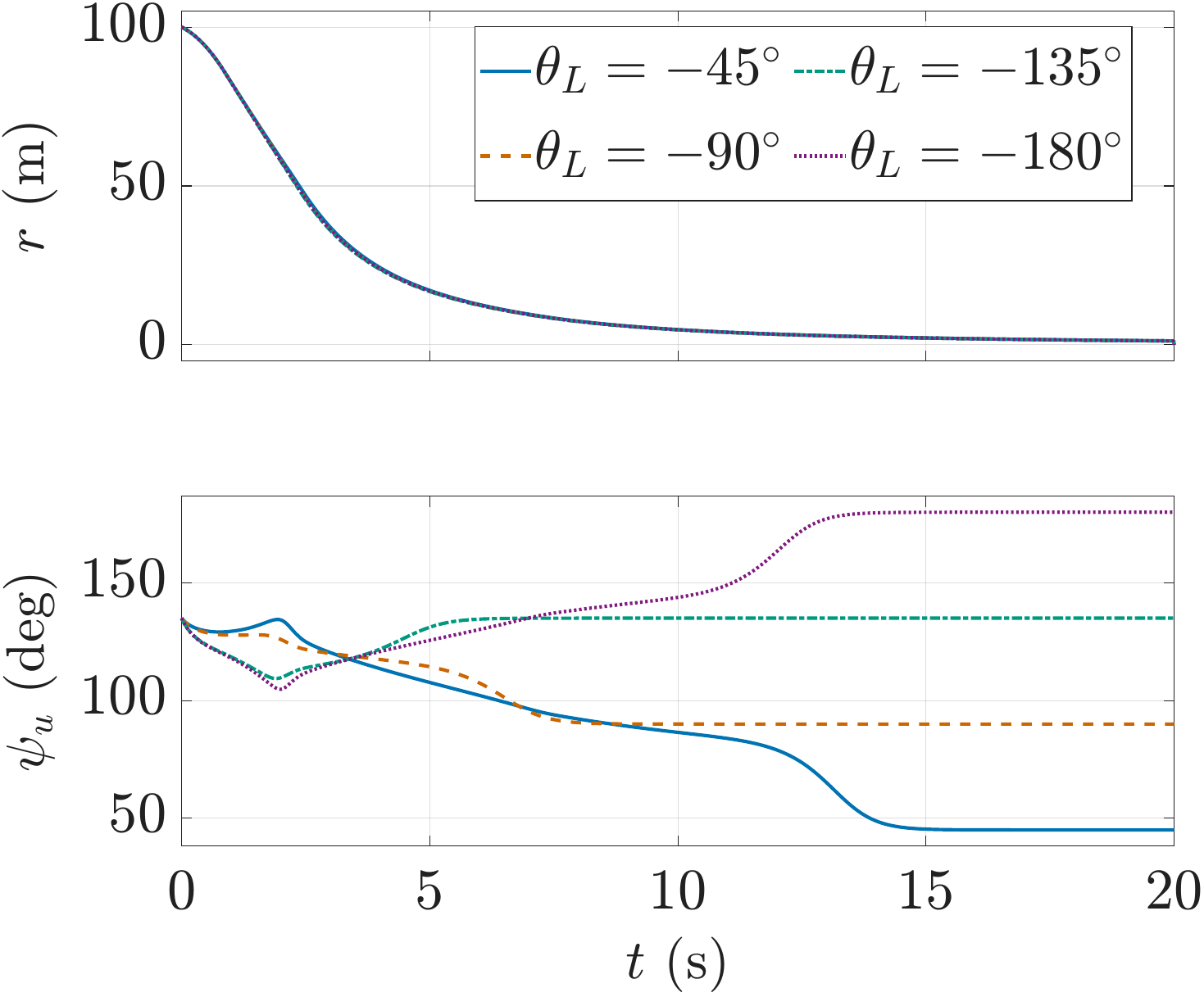}
		\caption{Relative range and heading angle.}
		\label{fig:stat_r_psi_u}
	\end{subfigure}
	\begin{subfigure}{0.33\linewidth}
		\centering
		\includegraphics[width=\linewidth]{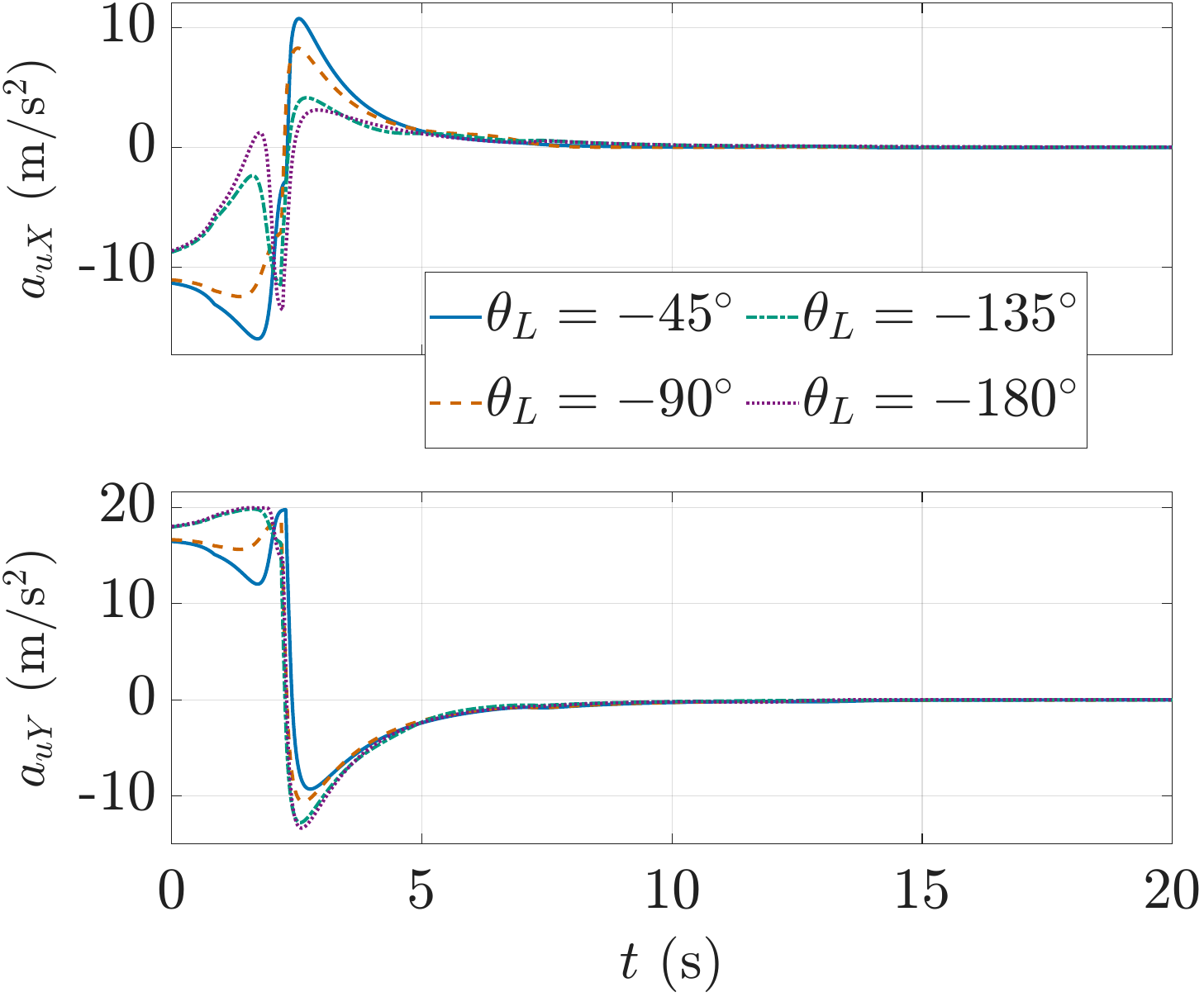}
		\caption{Control input to the equivalent agent U.}
		\label{fig:stat_amu}
	\end{subfigure}%
	\begin{subfigure}{0.33\linewidth}
		\centering
		\includegraphics[width=\linewidth]{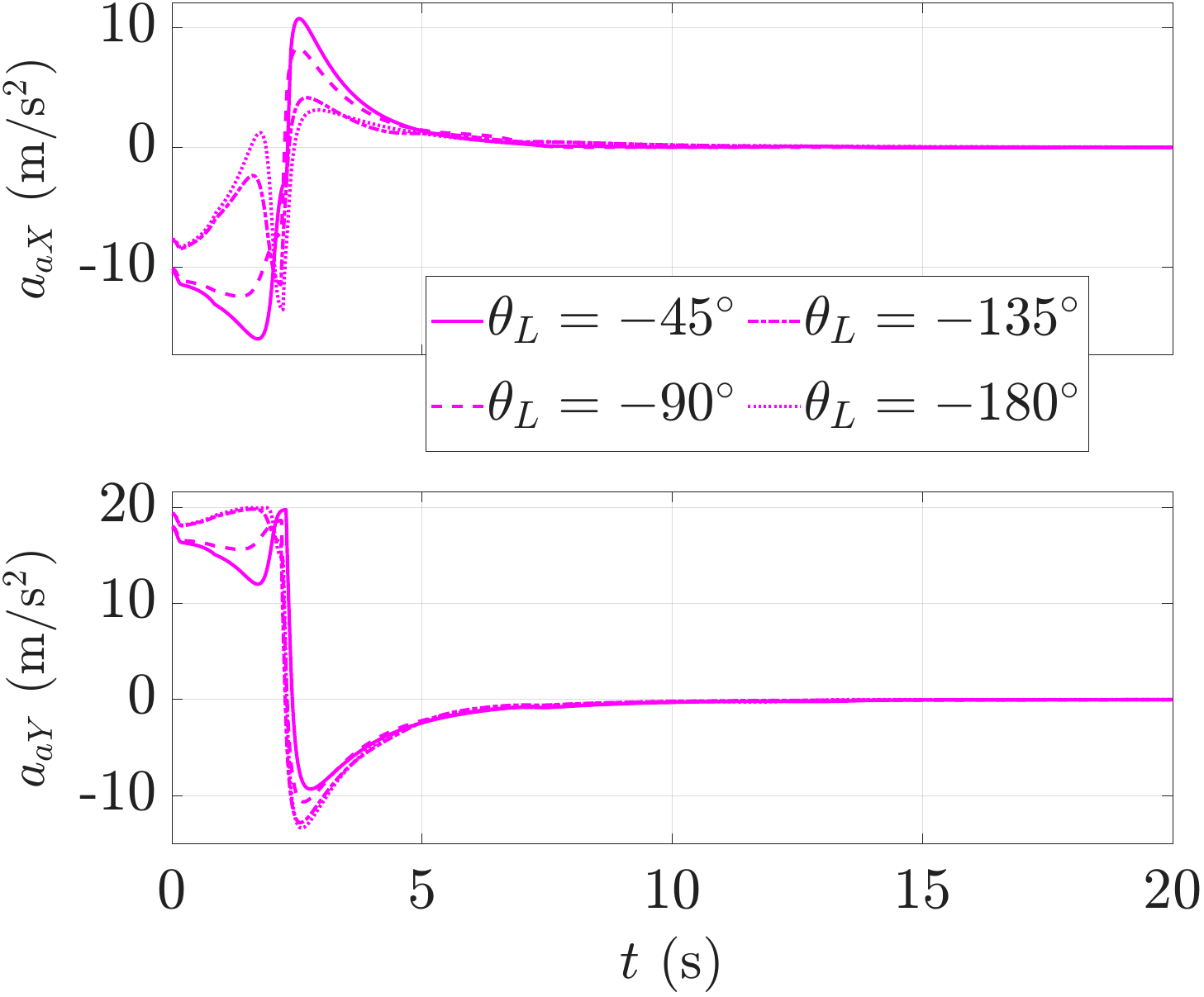}
		\caption{Control input to the UAV A.}
		\label{fig:stat_ama}
	\end{subfigure}%
	\begin{subfigure}{0.33\linewidth}
		\centering
		\includegraphics[width=\linewidth]{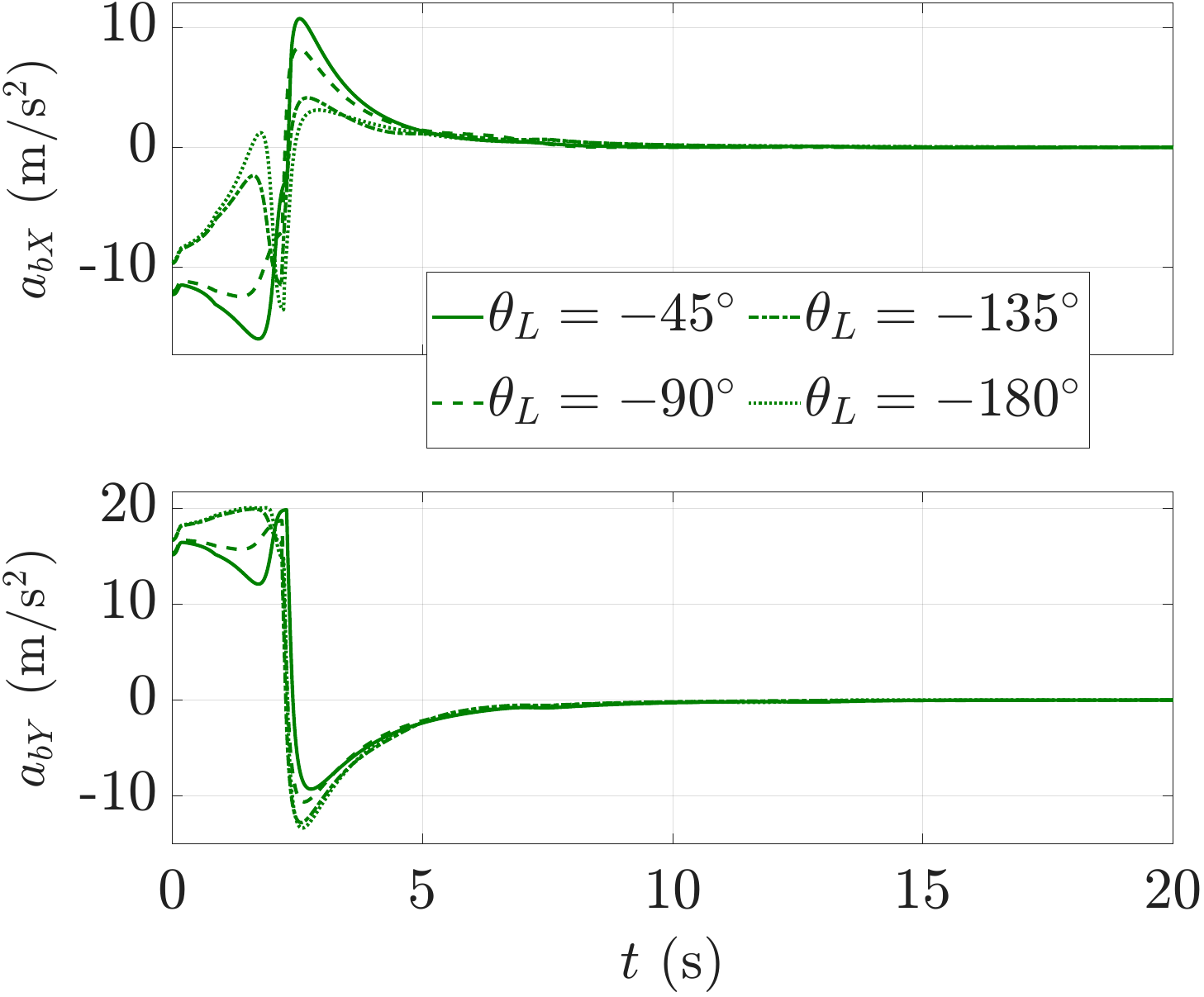}
		\caption{Control input to the UAV B.}
		\label{fig:stat_amb}
	\end{subfigure}
	\begin{subfigure}{0.33\linewidth}
		\centering
		\includegraphics[width=\linewidth]{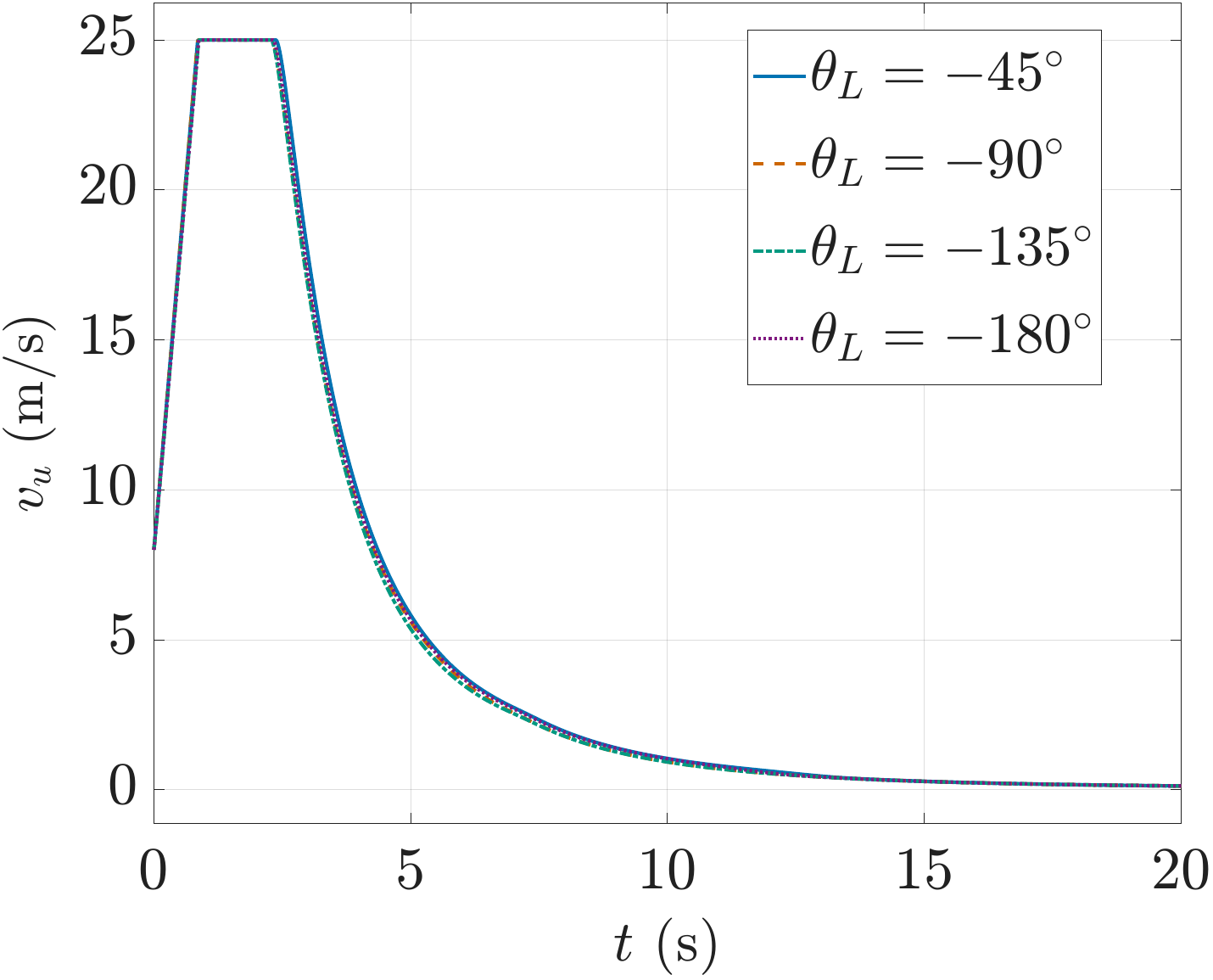}
		\caption{Speed of the equivalent agent U.}
		\label{fig:stat_vu}
	\end{subfigure}%
	\begin{subfigure}{0.33\linewidth}
		\centering
		\includegraphics[width=\linewidth]{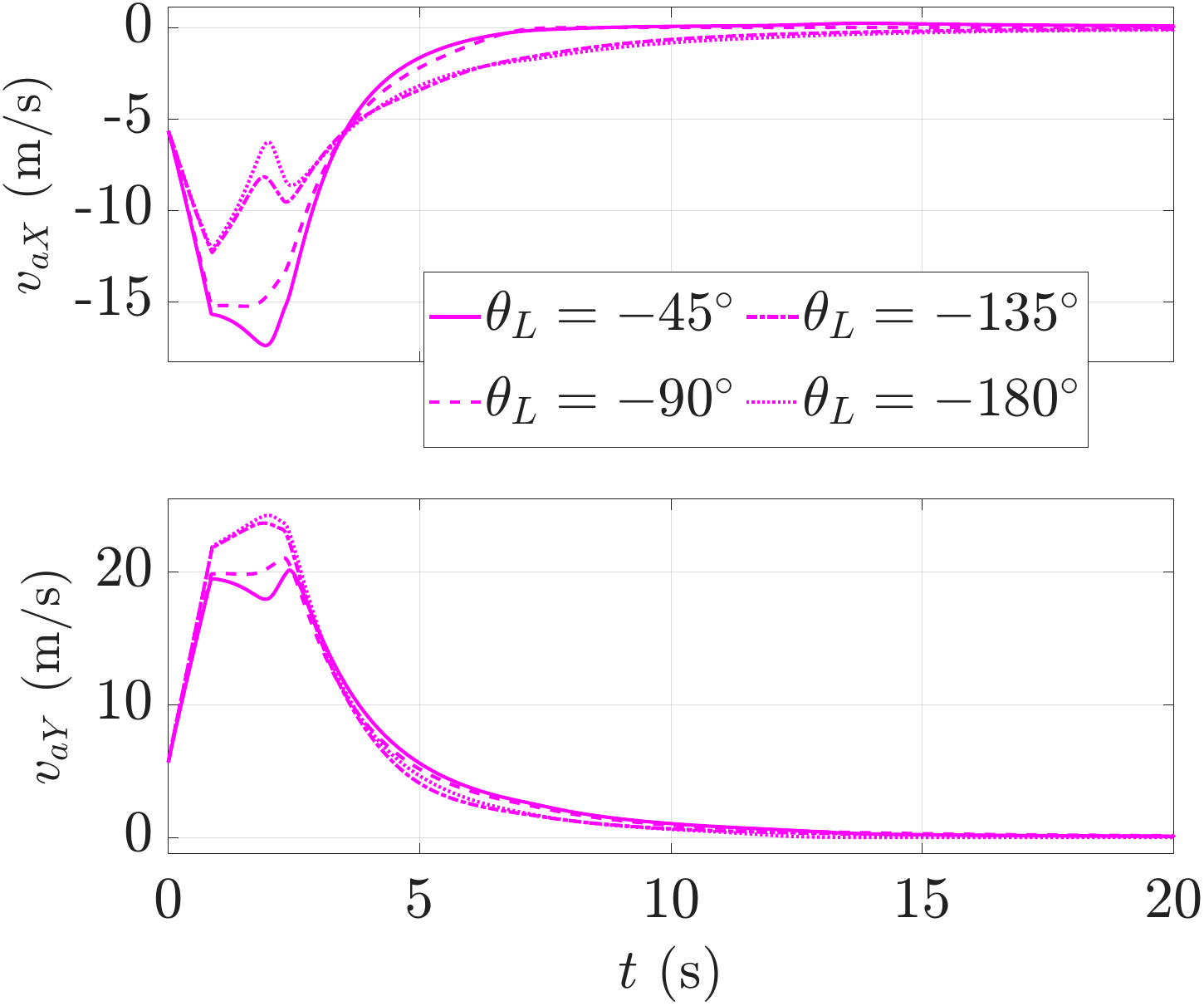}
		\caption{Velocity components of the UAV A.}
		\label{fig:stat_va}
	\end{subfigure}%
	\begin{subfigure}{0.33\linewidth}
		\centering
		\includegraphics[width=\linewidth]{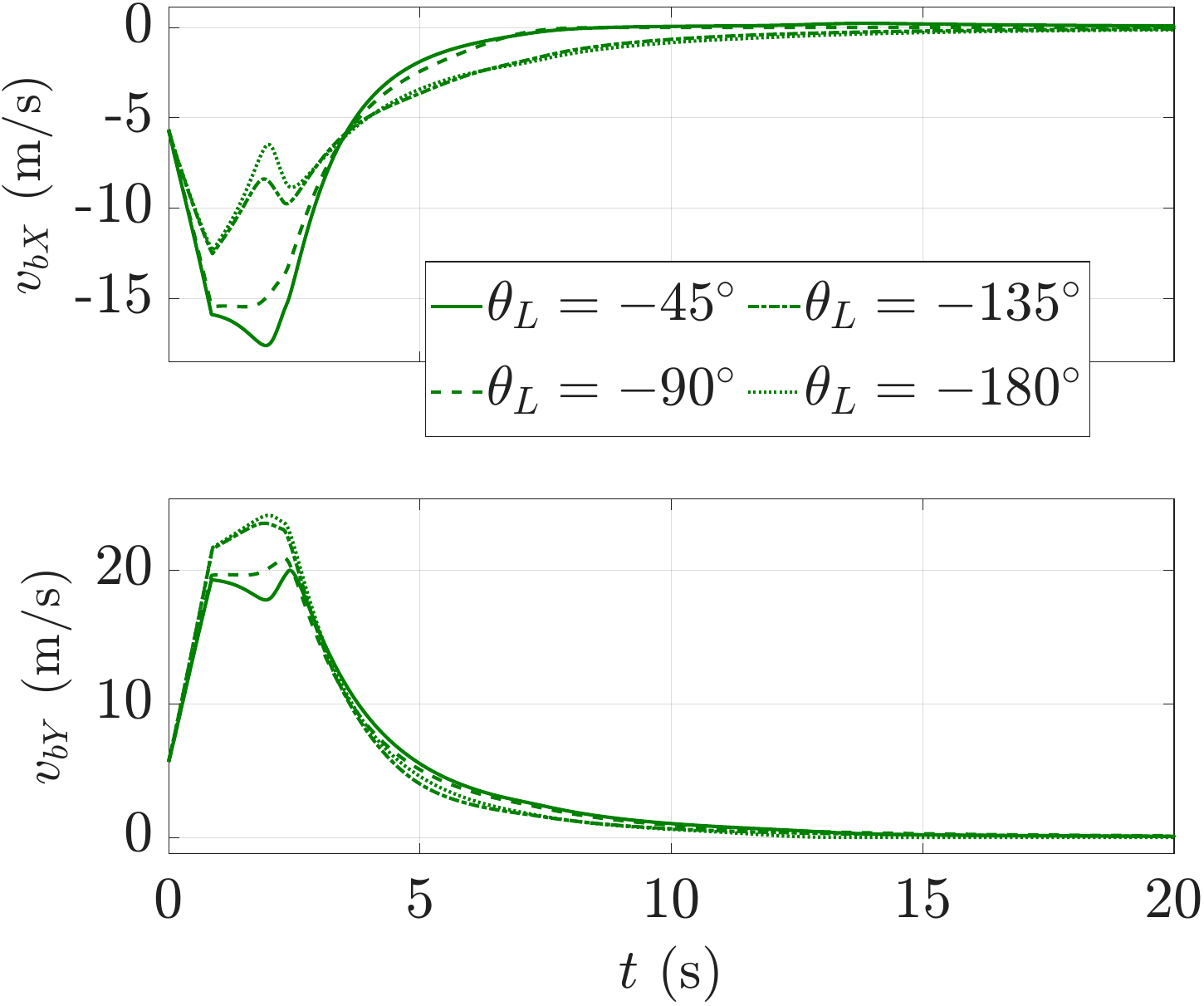}
		\caption{Velocity components of the UAV B.}
		\label{fig:stat_vb}
	\end{subfigure}
	\begin{subfigure}{0.33\linewidth}
		\centering
		\includegraphics[width=\linewidth]{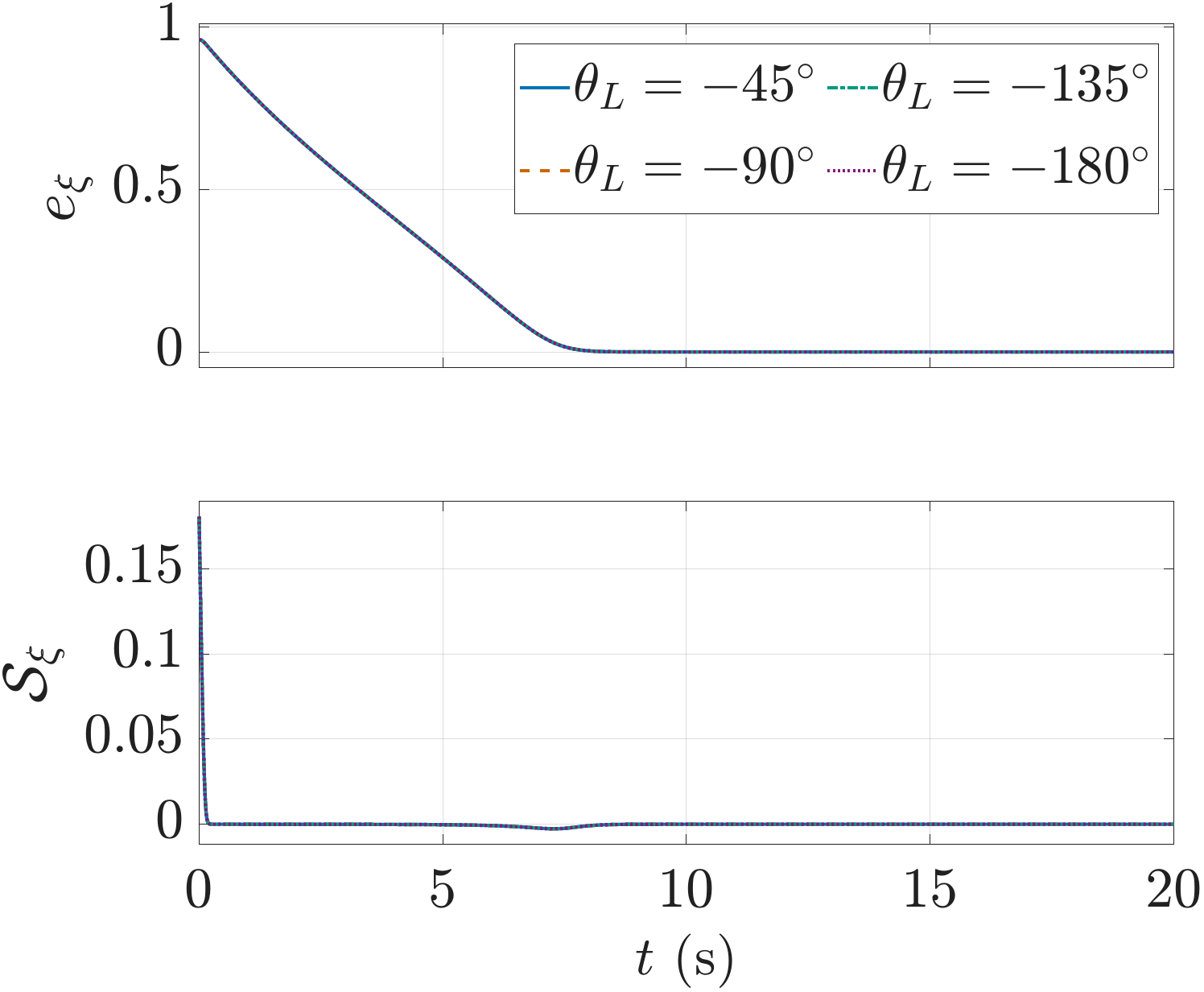}
		\caption{Error and sliding surface of the link.}
		\label{fig:Stat_Link_Sliding}
	\end{subfigure}%
	\begin{subfigure}{0.33\linewidth}
		\centering
		\includegraphics[width=\linewidth]{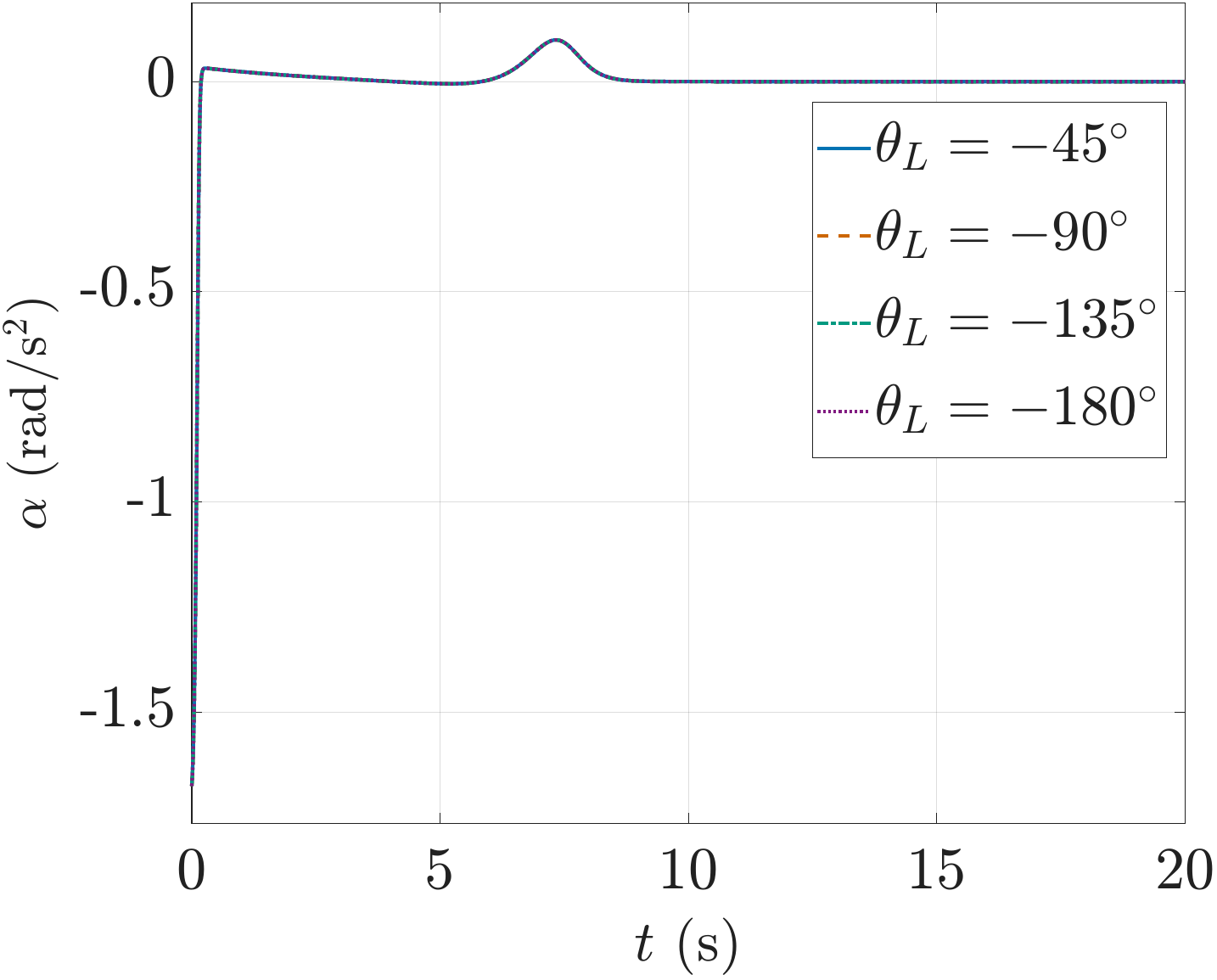}
		\caption{Control input to the link.}
		\label{fig:Stat_Link_Input}
	\end{subfigure}%
	\begin{subfigure}{0.33\linewidth}
		\centering
		\includegraphics[width=\linewidth]{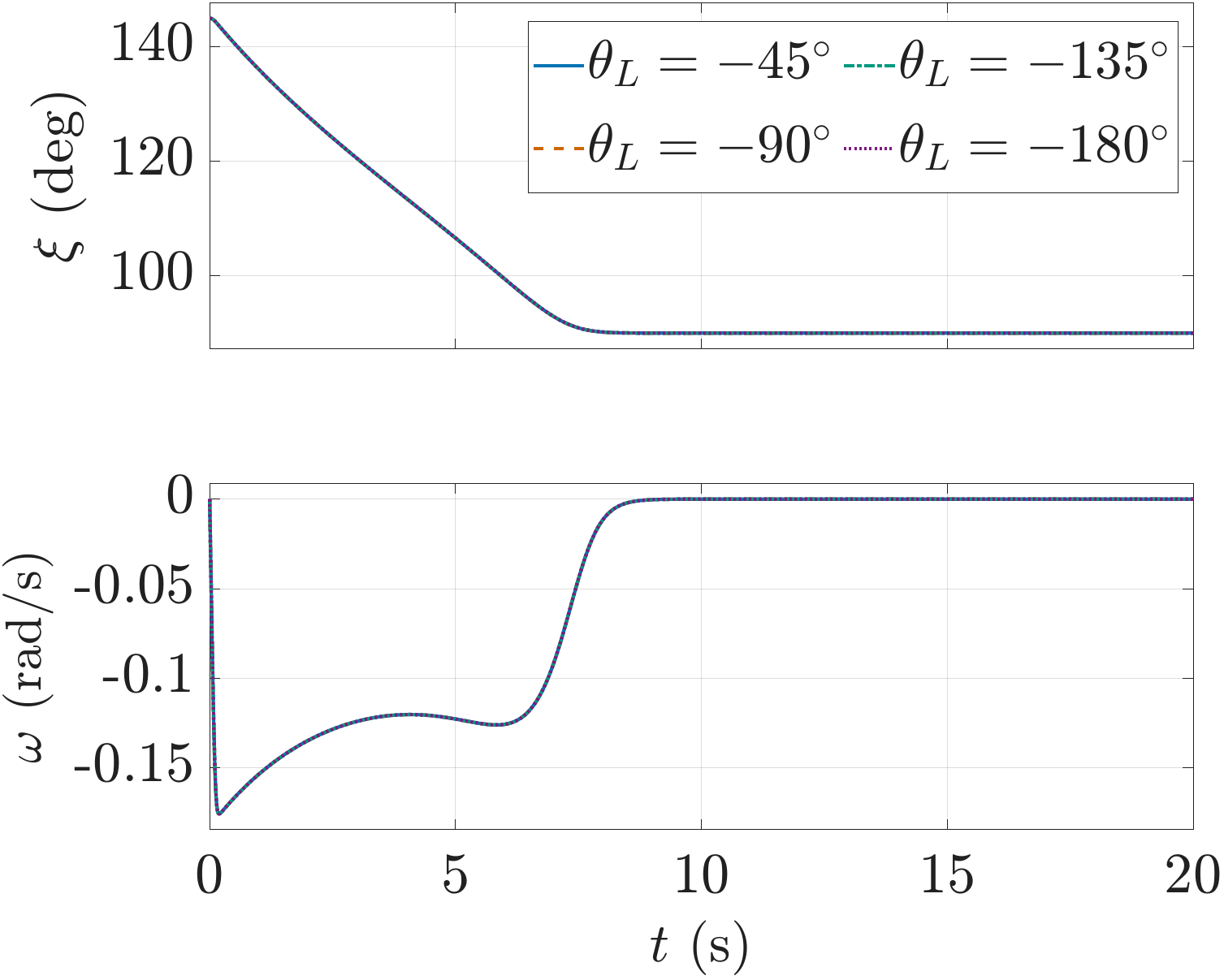}
		\caption{Link states.}
		\label{fig:Stat_Link_States}
	\end{subfigure}
	\caption{UAV delivering payload on a stationary platform.}
	\label{fig:stat}
\end{figure}

To further assess the robustness of the proposed guidance strategy to variations in the initial conditions and the prescribed landing angle, we conduct two Monte Carlo simulation studies, each consisting of 500 runs. In the first case, the initial position and heading of the UAV system are kept fixed, while the desired landing angle is uniformly varied over $\theta_{L}\in[-180^\circ,180^\circ]$. In the second case, the landing angle is fixed at $\theta_{L}=-90^\circ$, while the initial conditions are uniformly varied over $r(0)\in[40,160]$ m, $\theta(0)\in[-180^\circ,180^\circ]$, $\psi_u(0)\in[-180^\circ,180^\circ]$, and $v_u(0)\in[2,15]$ m/s. The remaining engagement parameters are kept identical to those used for the stationary case. For each case, the UAV trajectories and the terminal landing-angle error, range error, and link-orientation error are computed, as shown in \Cref{fig:mone_carlo_diff_landing_angle,fig:mone_carlo_diff_initial_condition}. It can be observed from \Cref{fig:case1_AllTrajectories,fig:case2_AllTrajectories} that the UAV successfully delivers the payload to the landing platform for all considered landing angles and initial configurations. Furthermore, one may notice from \Cref{fig:mone_carlo_diff_landing_angle,fig:mone_carlo_diff_initial_condition} that the landing-angle error, range error, and link-orientation error converge to zero at the time of landing for all the considered configurations.
\begin{figure}[!ht]
	\centering
	\begin{subfigure}{0.45\linewidth}
		\centering
		\includegraphics[width=\linewidth]{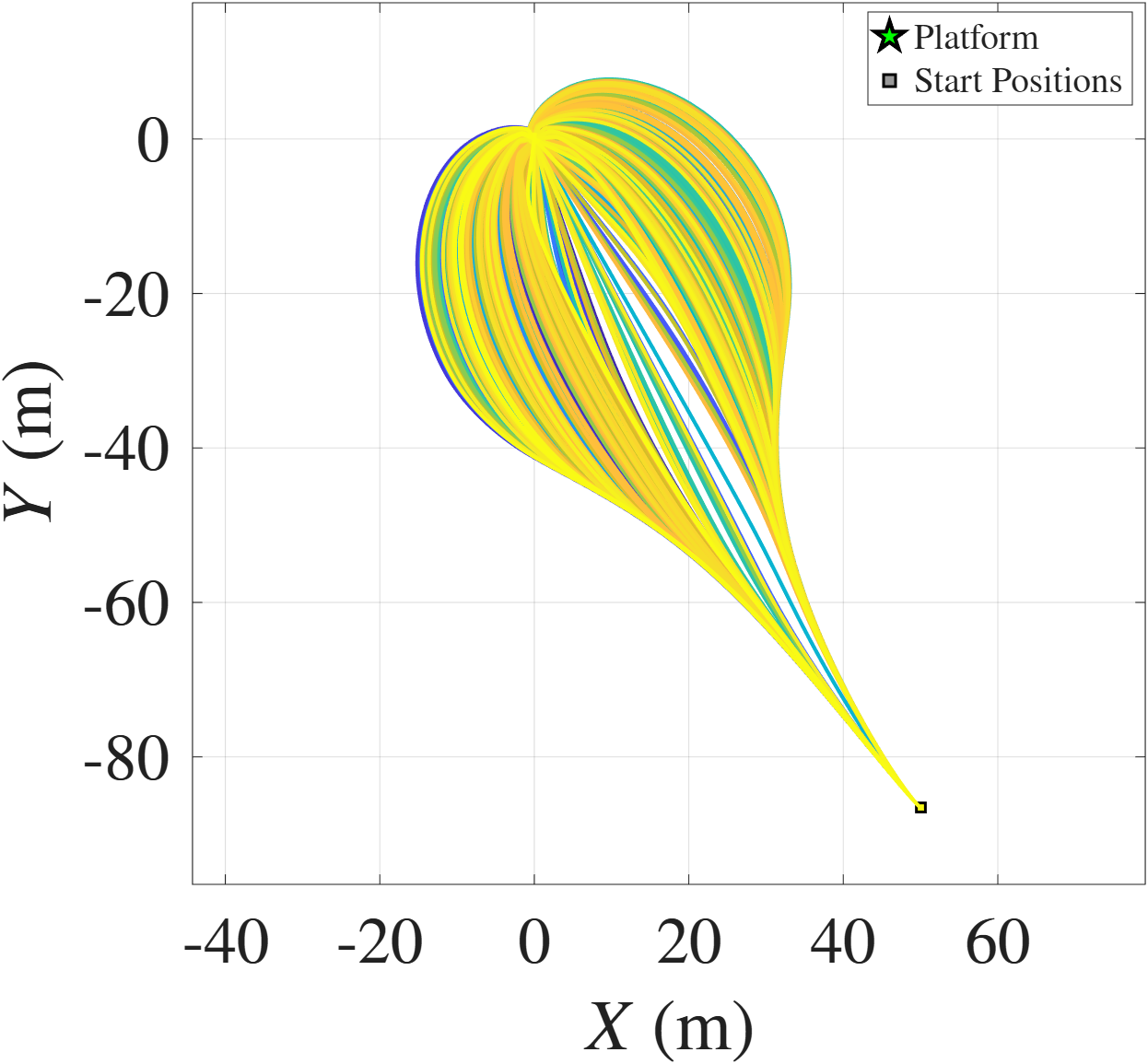}
		\caption{Trajectories.}
		\label{fig:case1_AllTrajectories}
	\end{subfigure}%
	\begin{subfigure}{0.45\linewidth}
		\centering
		\includegraphics[width=\linewidth]{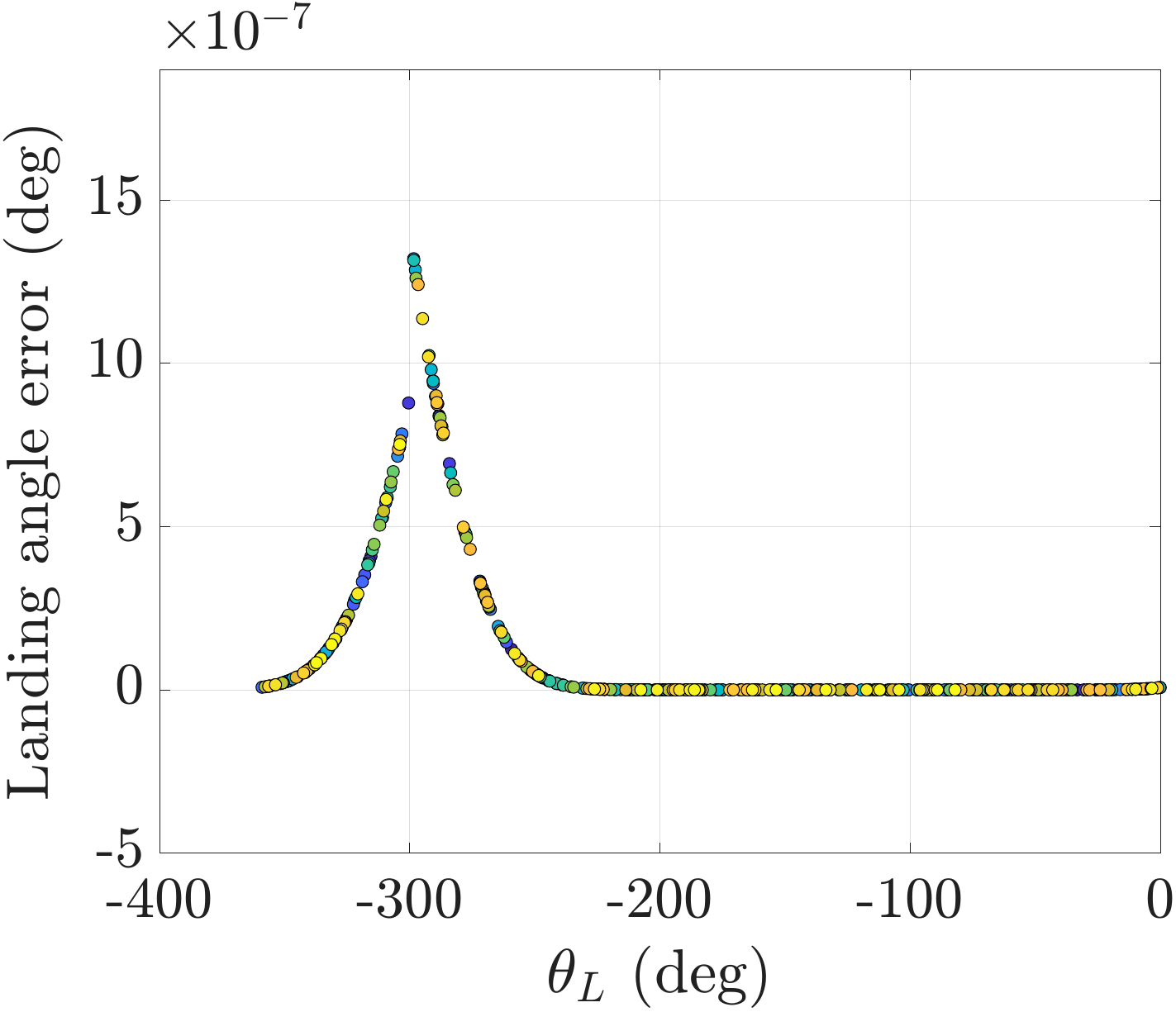}
		\caption{Landing angle error.}
		\label{fig:case1_LandingAngleError}
	\end{subfigure}
	\begin{subfigure}{0.45\linewidth}
		\centering
		\includegraphics[width=\linewidth]{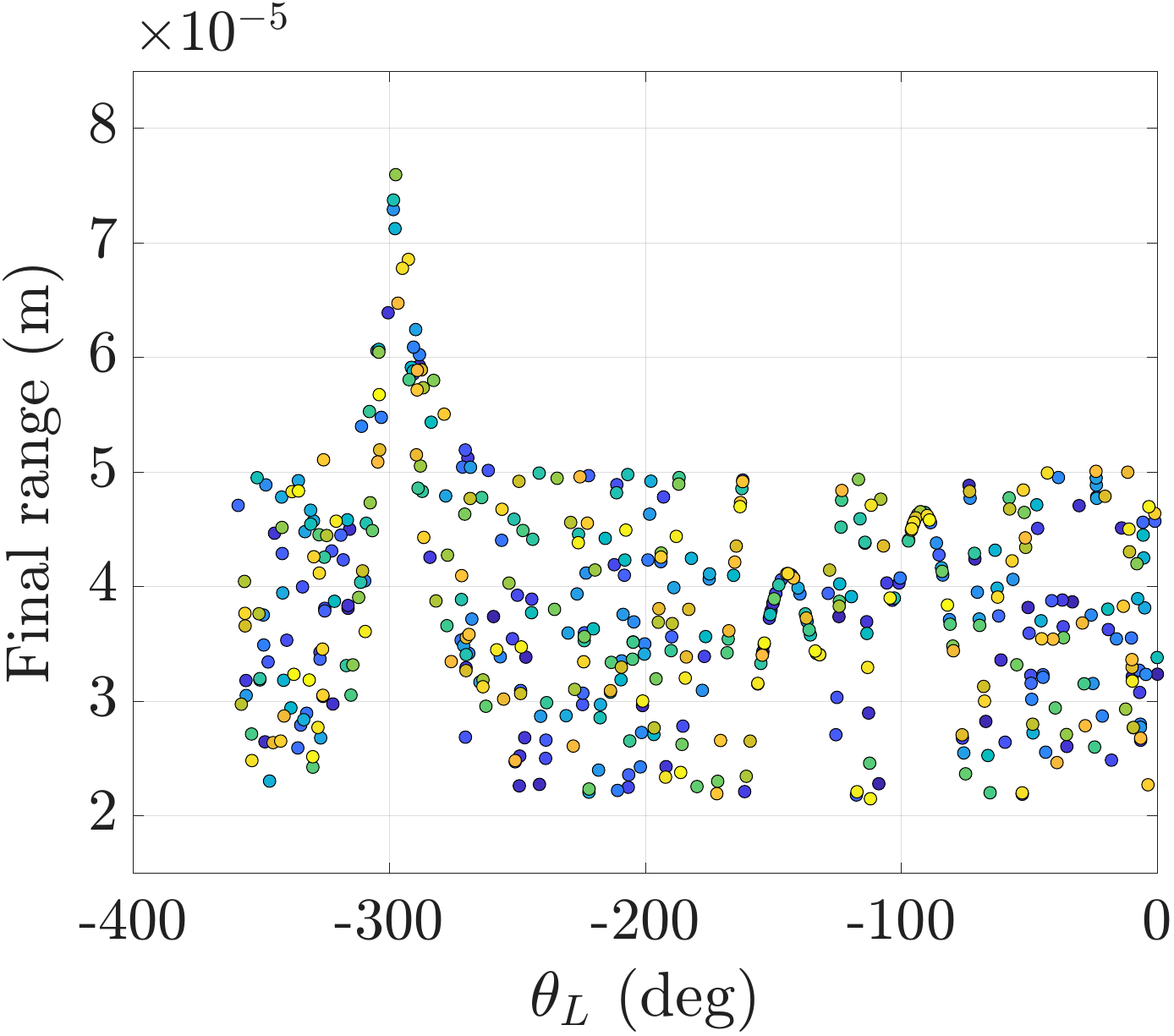}
		\caption{Final value of range.}
		\label{fig:case1_ClosestApproach}
	\end{subfigure}%
	\begin{subfigure}{0.45\linewidth}
		\centering
		\includegraphics[width=\linewidth]{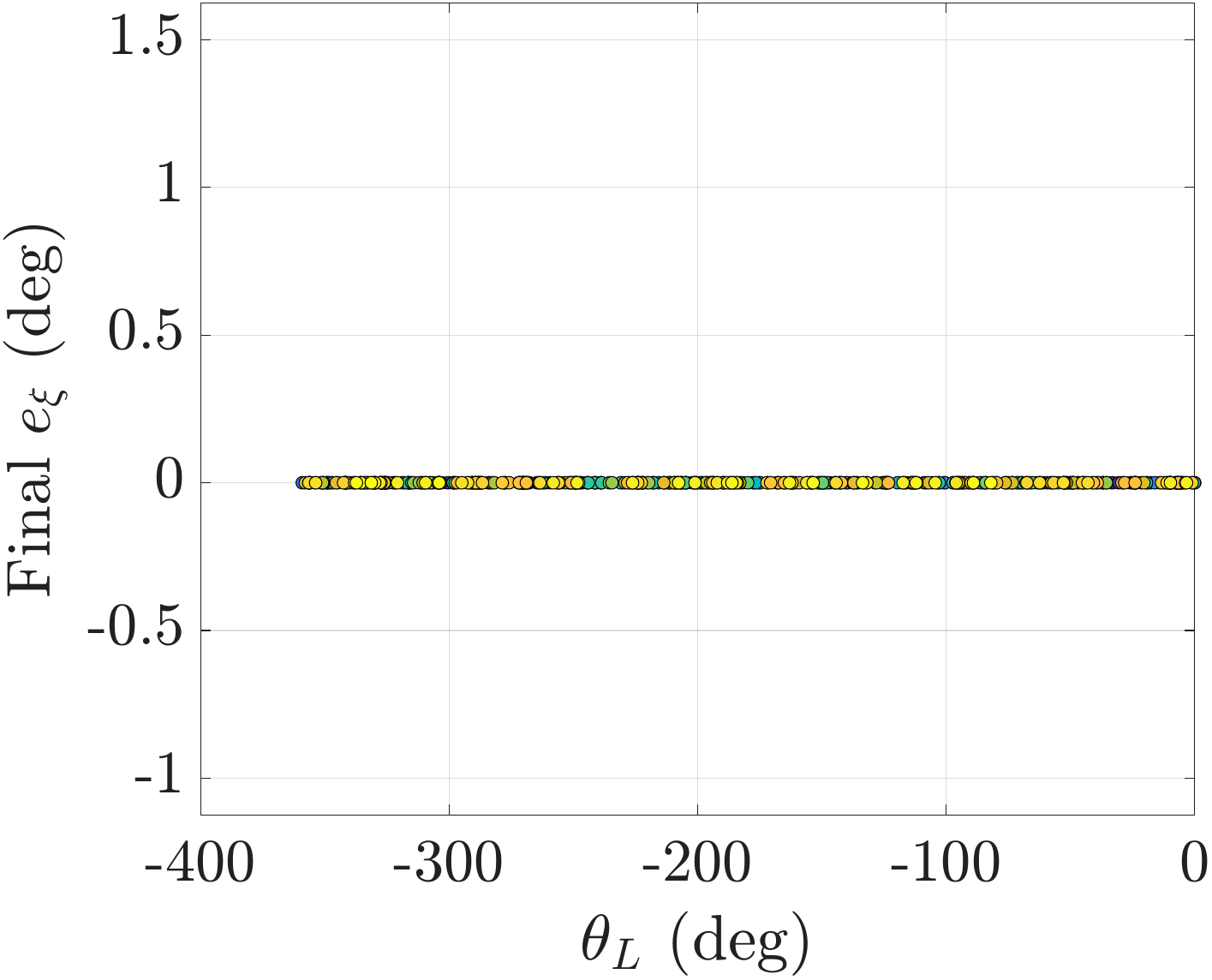}
		\caption{Final value of link orientation error.}
		\label{fig:case1_LinkOrientationError}
	\end{subfigure}
	\caption{Monte Carlo study: UAV delivering payload on a stationary platform with different landing angles.}
	\label{fig:mone_carlo_diff_landing_angle}
\end{figure}
\begin{figure}[!ht]
	\centering
	\begin{subfigure}{0.45\linewidth}
		\centering
		\includegraphics[width=\linewidth]{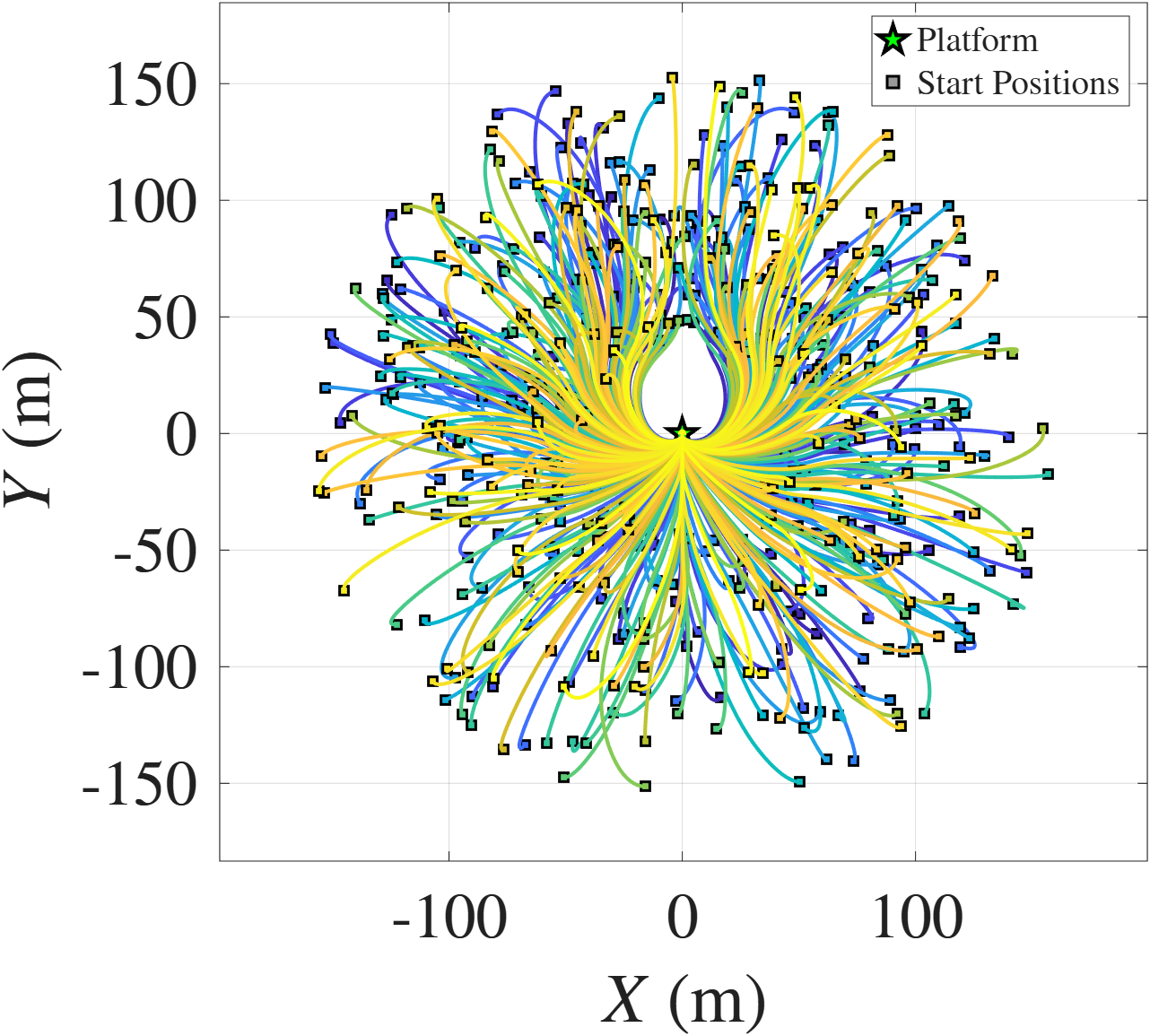}
		\caption{Trajectories.}
		\label{fig:case2_AllTrajectories}
	\end{subfigure}%
	\begin{subfigure}{0.45\linewidth}
		\centering
		\includegraphics[width=\linewidth]{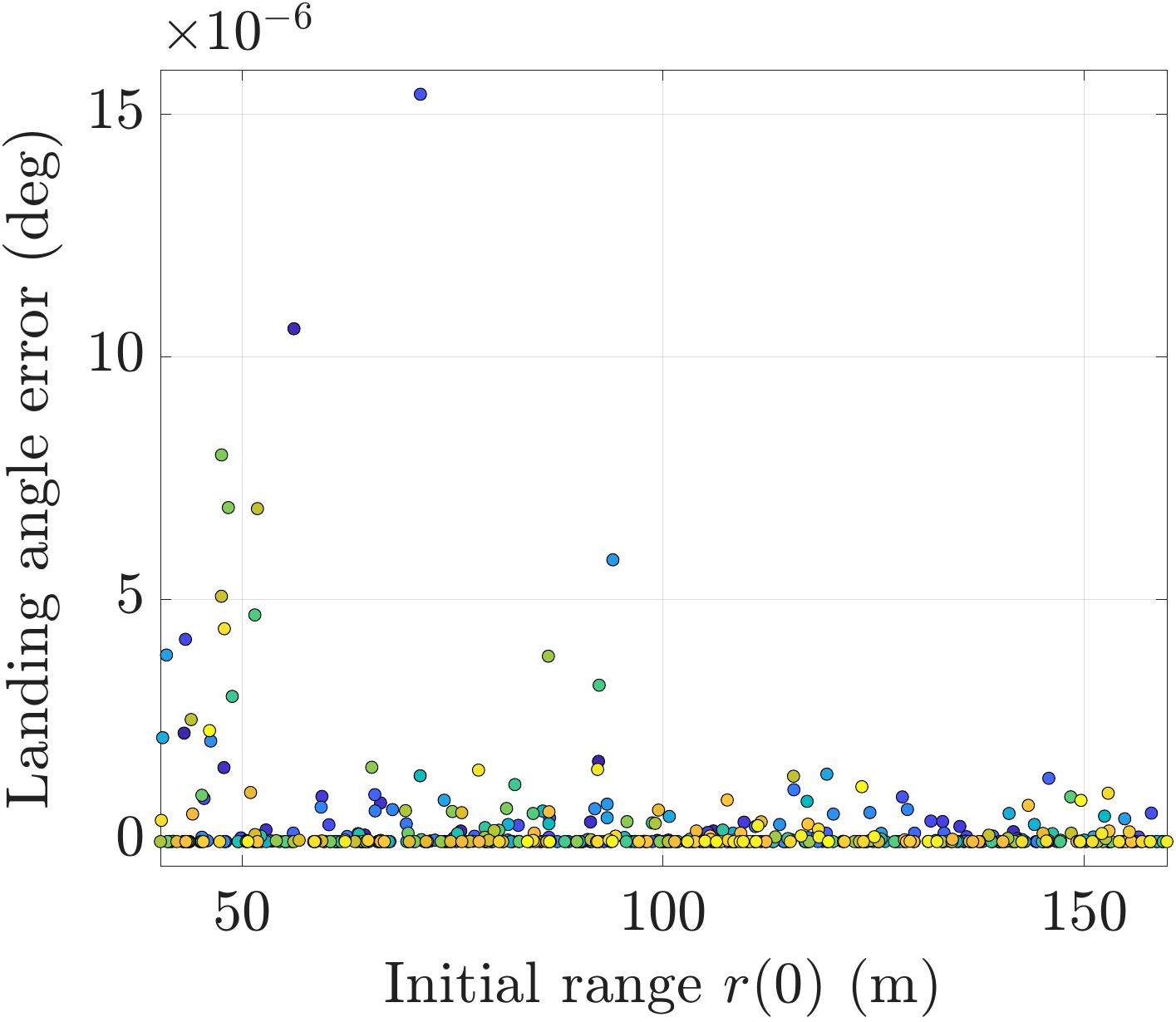}
		\caption{Landing angle error.}
		\label{fig:case2_LandingAngleError}
	\end{subfigure}
	\begin{subfigure}{0.45\linewidth}
		\centering
		\includegraphics[width=\linewidth]{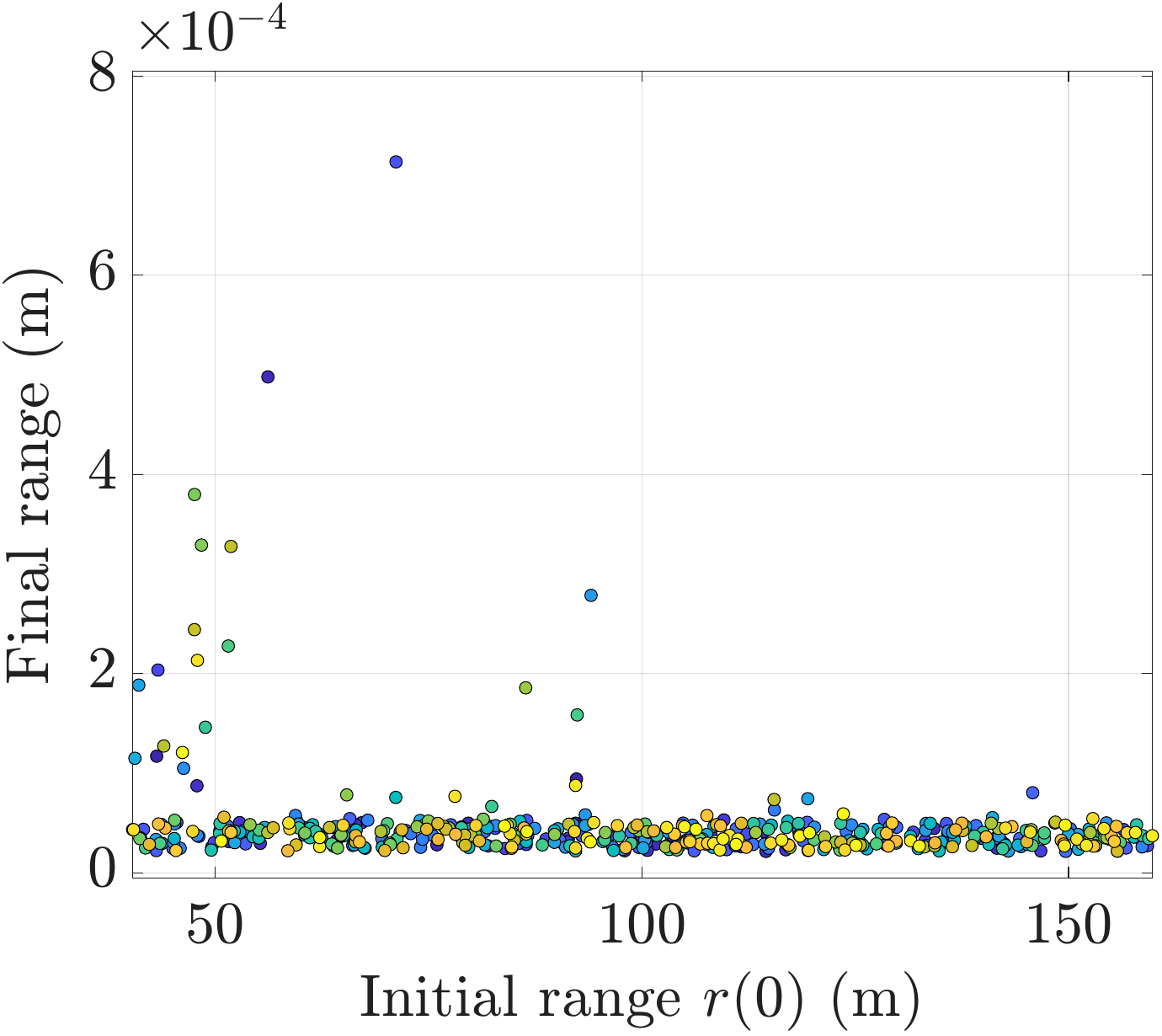}
		\caption{Final value of range.}
		\label{fig:case2_ClosestApproach}
	\end{subfigure}%
	\begin{subfigure}{0.45\linewidth}
		\centering
		\includegraphics[width=\linewidth]{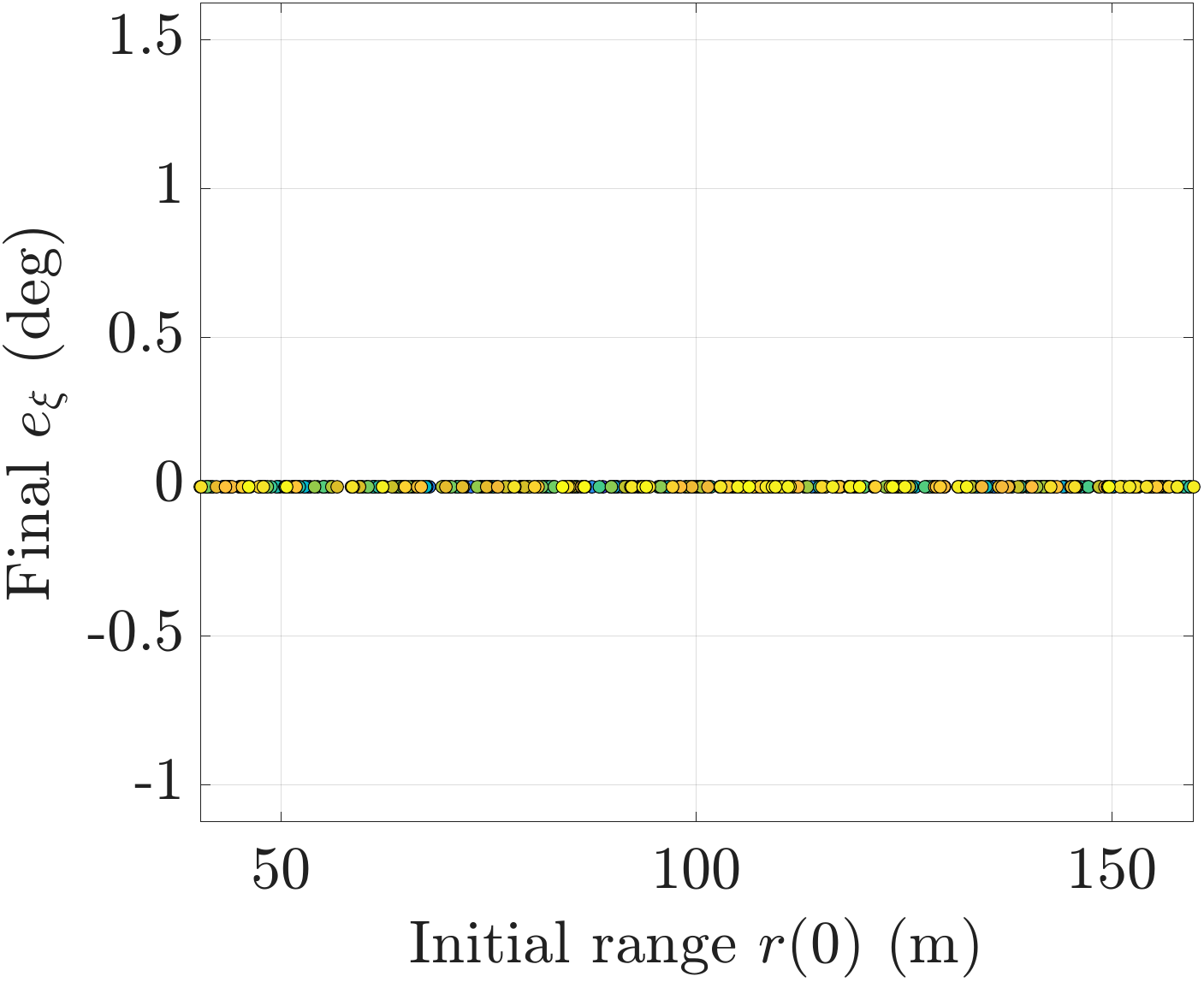}
		\caption{Final value of link orientation error.}
		\label{fig:case2_LinkOrientationError}
	\end{subfigure}
	\caption{Monte Carlo study: UAV delivering payload on a stationary platform from different initial configurations.}
	\label{fig:mone_carlo_diff_initial_condition}
\end{figure}
\subsection{Payload delivery on a maneuvering platform}
Next, we evaluate the performance of the proposed cooperative payload delivery strategy for a scenario where the landing platform is undergoing various active maneuvers. In this case, the platform is initially located at $(-30,100)\ \si{m}$ with an initial speed of $2\ \si{m/s}$ and a heading angle of $235^\circ$. On the other hand, the equivalent agent is initially positioned at $(5,5)\ \si{m}$, which is traveling at a speed of $5\ \si{m/s}$ with a heading angle of $60^\circ$. The initial orientation of the rigid link connecting the UAVs is set to $145^\circ$ with an initial angular velocity of $0\ \si{rad/s}$. As established theoretically, the only feasible landing angle for a maneuvering platform is $0^\circ$. Consequently, the desired landing angle is chosen as $0^\circ$. The value of design parameters $\mathcal{K}_{3}$ and $\mathcal{K}_{4}$ are chosen as $\mathcal{K}_{4} = 3$ and $\mathcal{K}_{6}=10$, respectively, while other settings remain identical to those used in the stationary case.

To test the robustness of the proposed strategy, the platform is subjected to a series of five distinct maneuvering profiles, denoted by M1--M5: M1 (Balanced maneuver): Lateral acceleration of $0.5\ \si{m/s^2}$, longitudinal acceleration of $0.5 \ \si{m/s^2}$. M2 (High longitudinal, low lateral): Lateral acceleration of $0.25 \ \si{m/s^2}$, longitudinal acceleration of $0.75 \ \si{m/s^2}$. M3 (High lateral, low longitudinal): Lateral acceleration of 0.75 \ \si{m/s^2}, longitudinal acceleration of $0.25 \ \si{m/s^2}$. M4 (Pure longitudinal): Lateral acceleration of $0.0 \ \si{m/s^2}$, longitudinal acceleration of $0.75 \ \si{m/s^2}$ (straight-line acceleration). M5 (Pure lateral): Lateral acceleration of $-0.5 \ \si{m/s^2}$, longitudinal acceleration of $0.0 \ \si{m/s^2}$ (pure turning maneuver). Note that these scenarios represent varying degrees of lateral and longitudinal accelerations to simulate complex platform turns.

Under the proposed guidance strategy (presented in \Cref{thm:man_amr,thm:man_amtheta}), we depict the performance through \Cref{fig:man}. \Cref{fig:man_Trajectory} demonstrates that the UAV system successfully delivers the payload for all cases, regardless of the maneuver of the platform. Similar to the stationary case, one can notice that $\mathcal{S}_{1}$  and $\mathcal{S}_{3} = \dot{\theta}$ converges to zero within a fixed-time. Consequently, the relative separation between the UAV system and the landing platform becomes zero, and the heading angle of the UAV converges with the heading angle of the platform to deliver the payload at the landing angle of zero degrees. Unlike the stationary case, once the UAV lands on the moving platform, it synchronizes its speed and heading with the platform's, as shown in \Cref{fig:man_Range_Heading,fig:man_Velocity}. As the link controller parameters, the initial configuration is kept the same as the stationary case, we observe a similar tracking performance.
\begin{figure}[!ht]
	\centering
	\begin{subfigure}{0.33\linewidth}
		\centering
		\includegraphics[width=\linewidth]{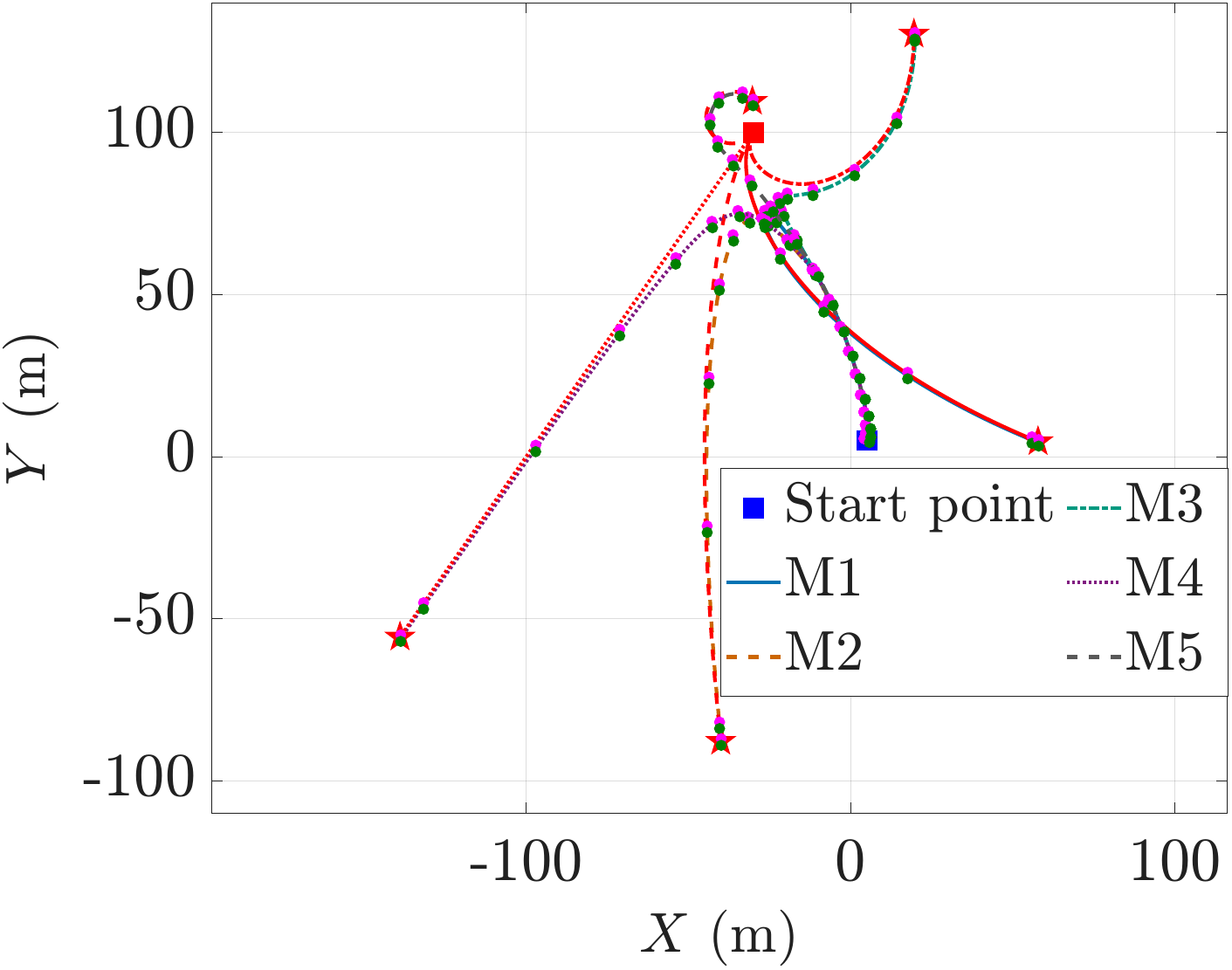}
		\caption{Trajectory.}
		\label{fig:man_Trajectory}
	\end{subfigure}%
	\begin{subfigure}{0.33\linewidth}
		\centering
		\includegraphics[width=\linewidth]{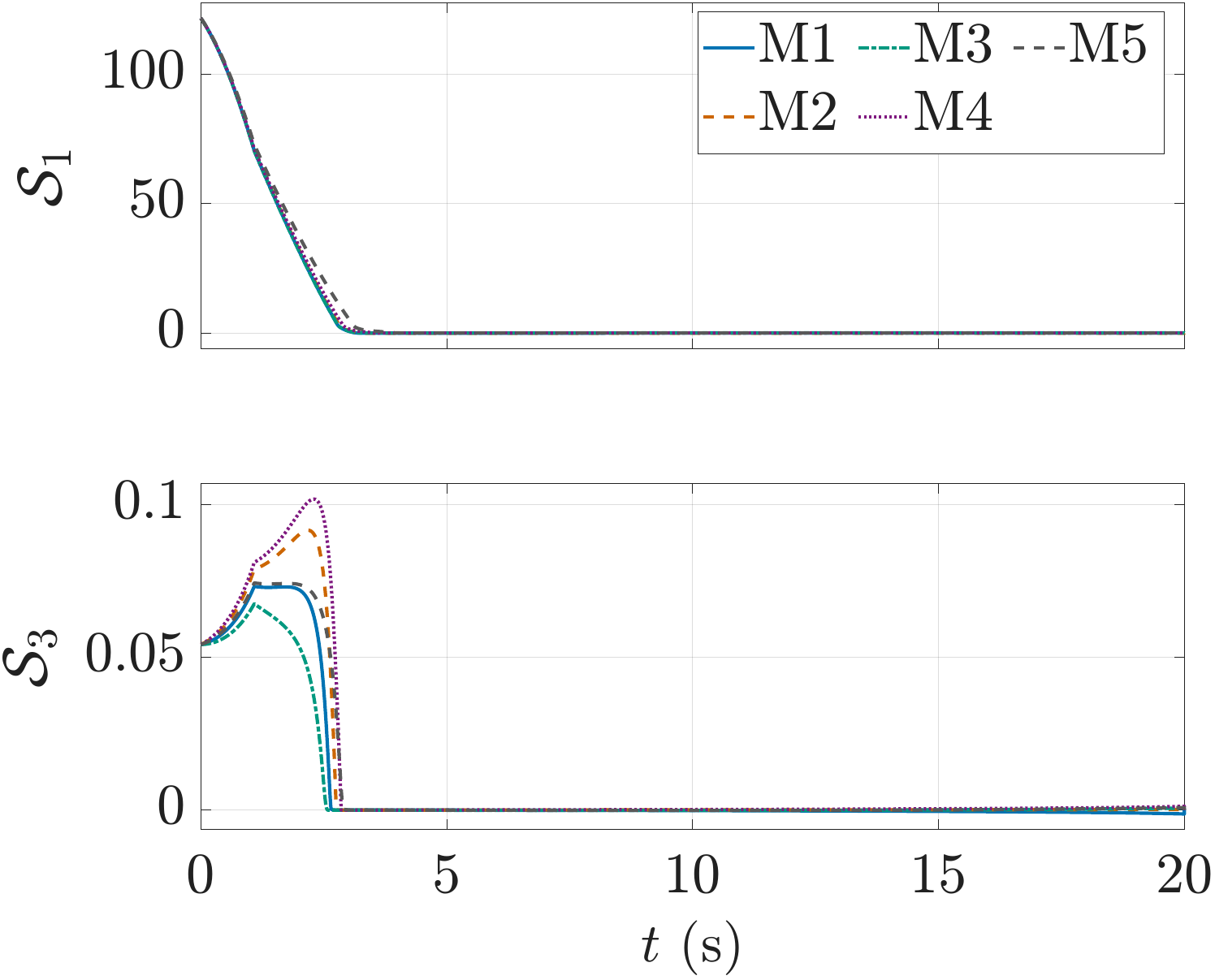}
		\caption{Sliding surfaces.}
		\label{fig:man_Sliding_Surfaces}
	\end{subfigure}%
	\begin{subfigure}{0.33\linewidth}
		\centering
		\includegraphics[width=\linewidth]{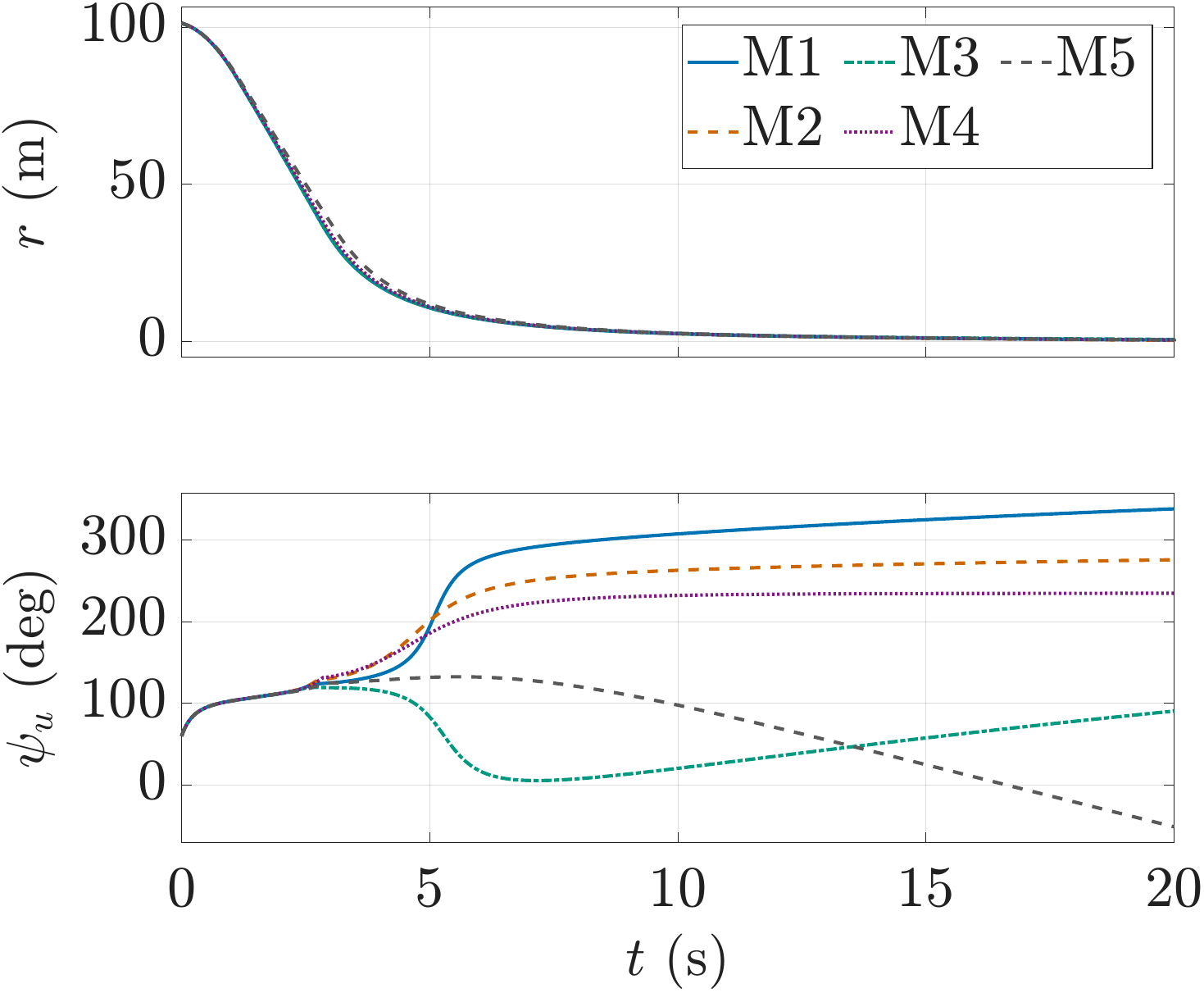}
		\caption{Relative range and heading angle.}
		\label{fig:man_Range_Heading}
	\end{subfigure}
	\begin{subfigure}{0.33\linewidth}
		\centering
		\includegraphics[width=\linewidth]{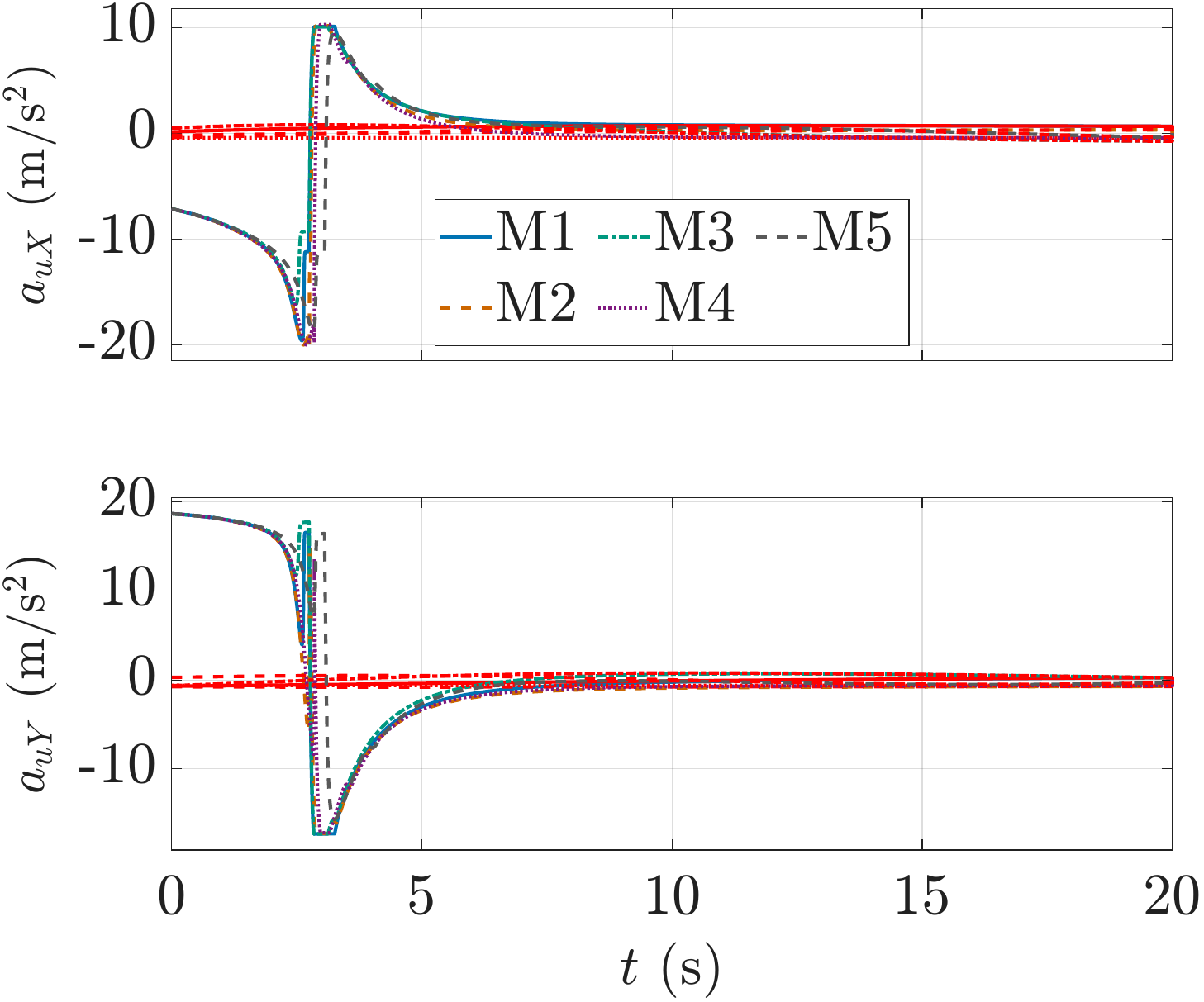}
		\caption{Control input to the equivalent agent U.}
		\label{fig:man_Acc_Virtual}
	\end{subfigure}%
	\begin{subfigure}{0.33\linewidth}
		\centering
		\includegraphics[width=\linewidth]{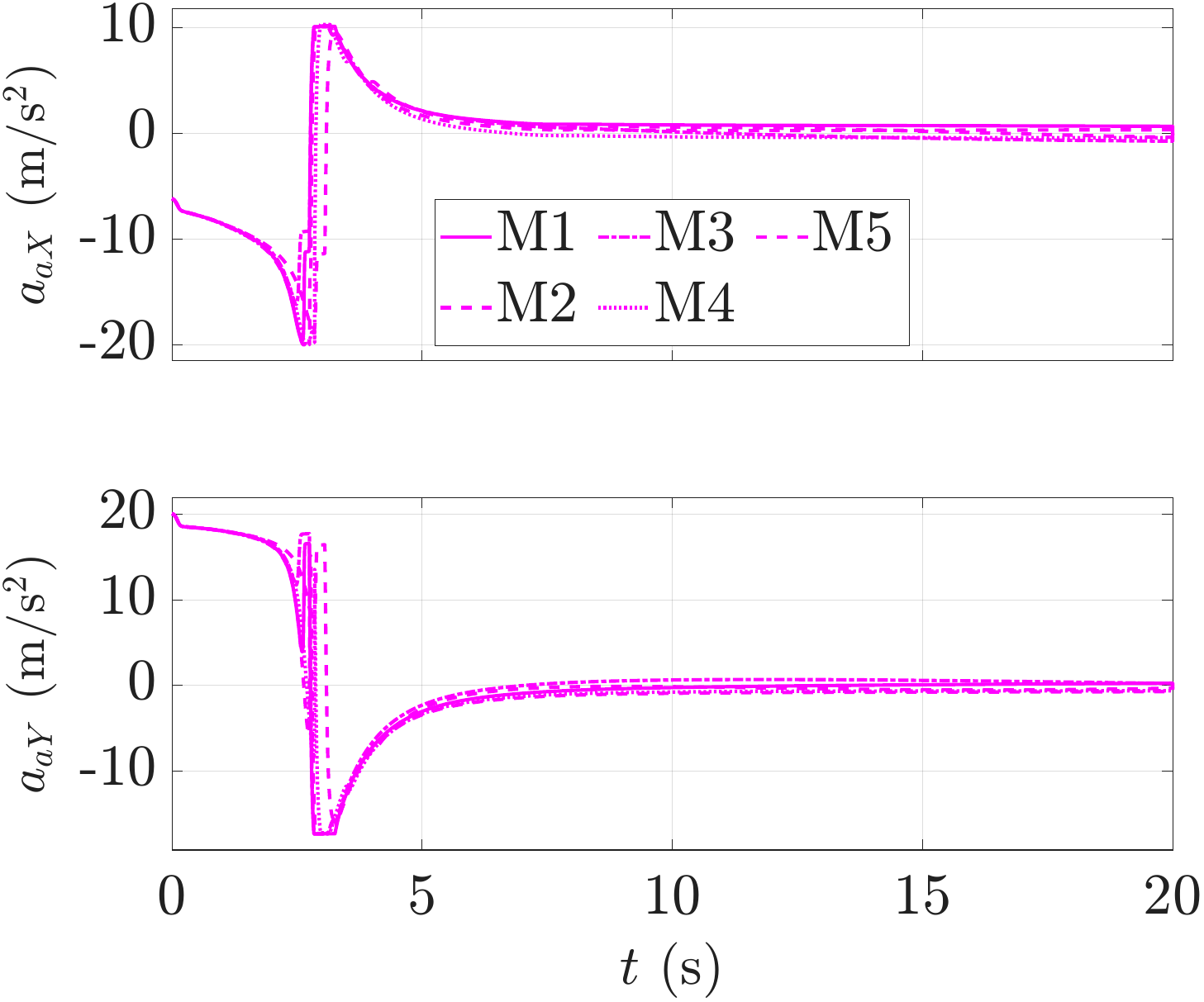}
		\caption{Control input to the UAV A.}
		\label{fig:man_Acc_UAV_A}
	\end{subfigure}%
	\begin{subfigure}{0.33\linewidth}
		\centering
		\includegraphics[width=\linewidth]{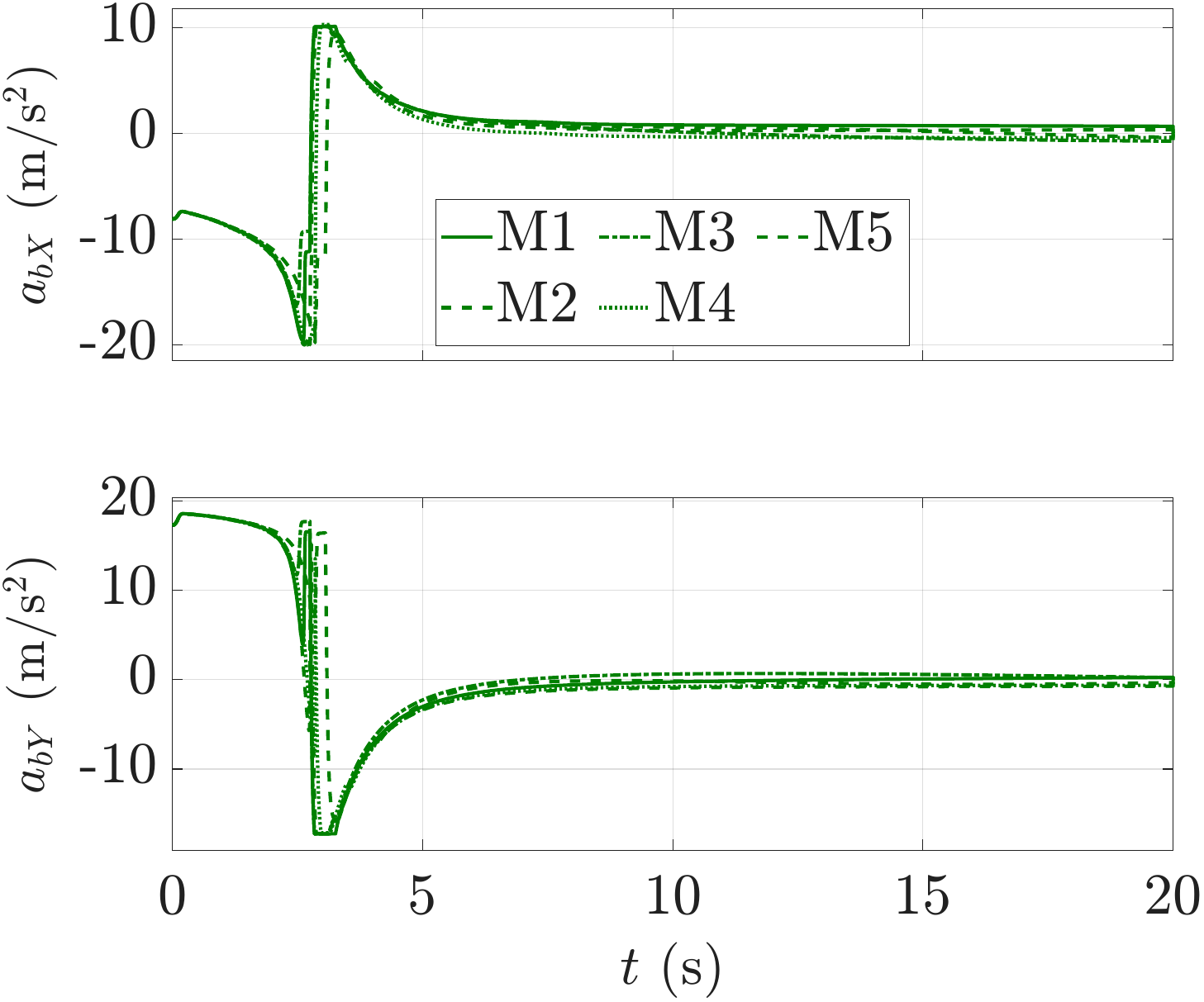}
		\caption{Control input to the UAV B.}
		\label{fig:man_Acc_UAV_B}
	\end{subfigure}
	\begin{subfigure}{0.33\linewidth}
		\centering
		\includegraphics[width=\linewidth]{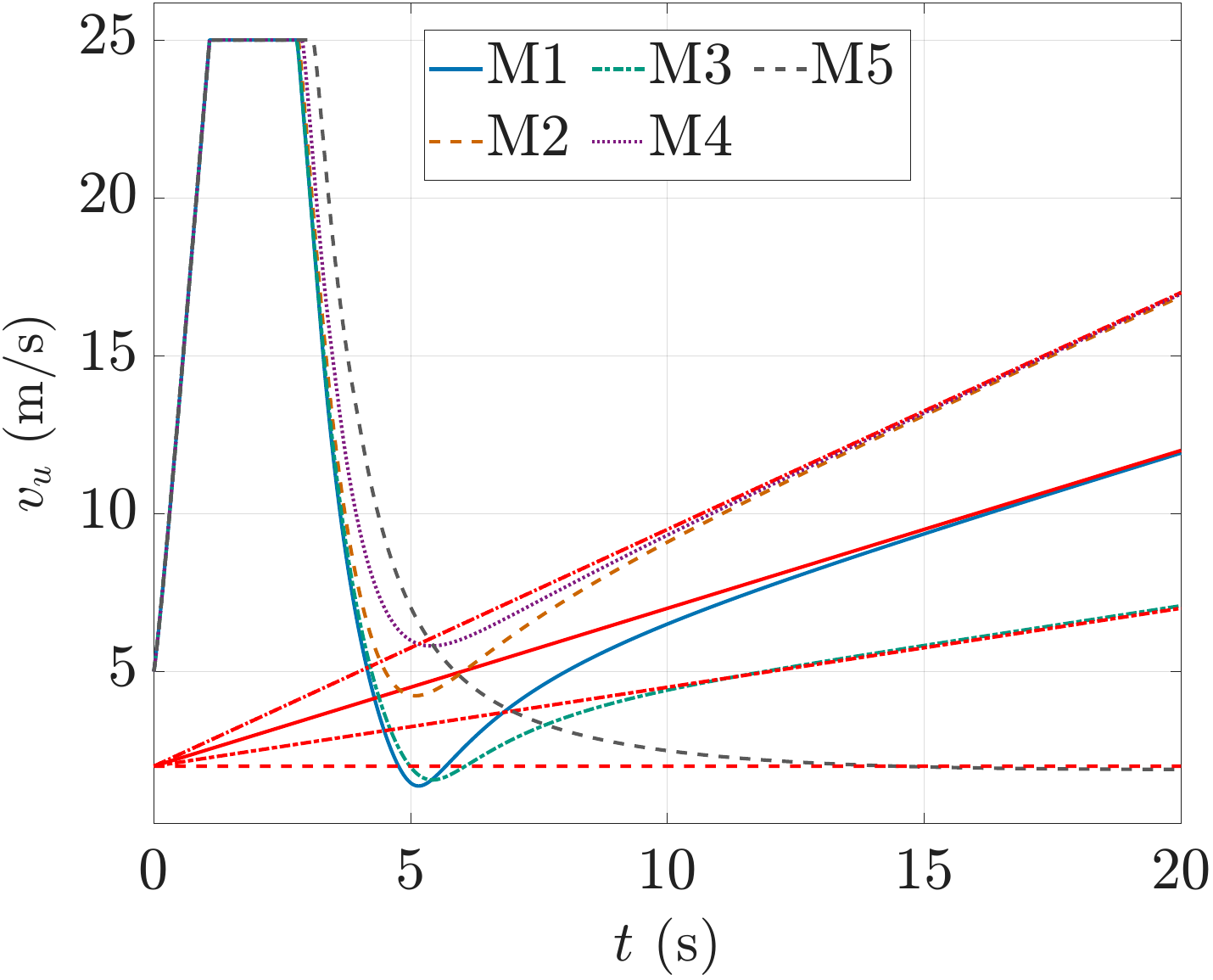}
		\caption{Speed of the equivalent agent U.}
		\label{fig:man_Velocity}
	\end{subfigure}%
	\begin{subfigure}{0.33\linewidth}
		\centering
		\includegraphics[width=\linewidth]{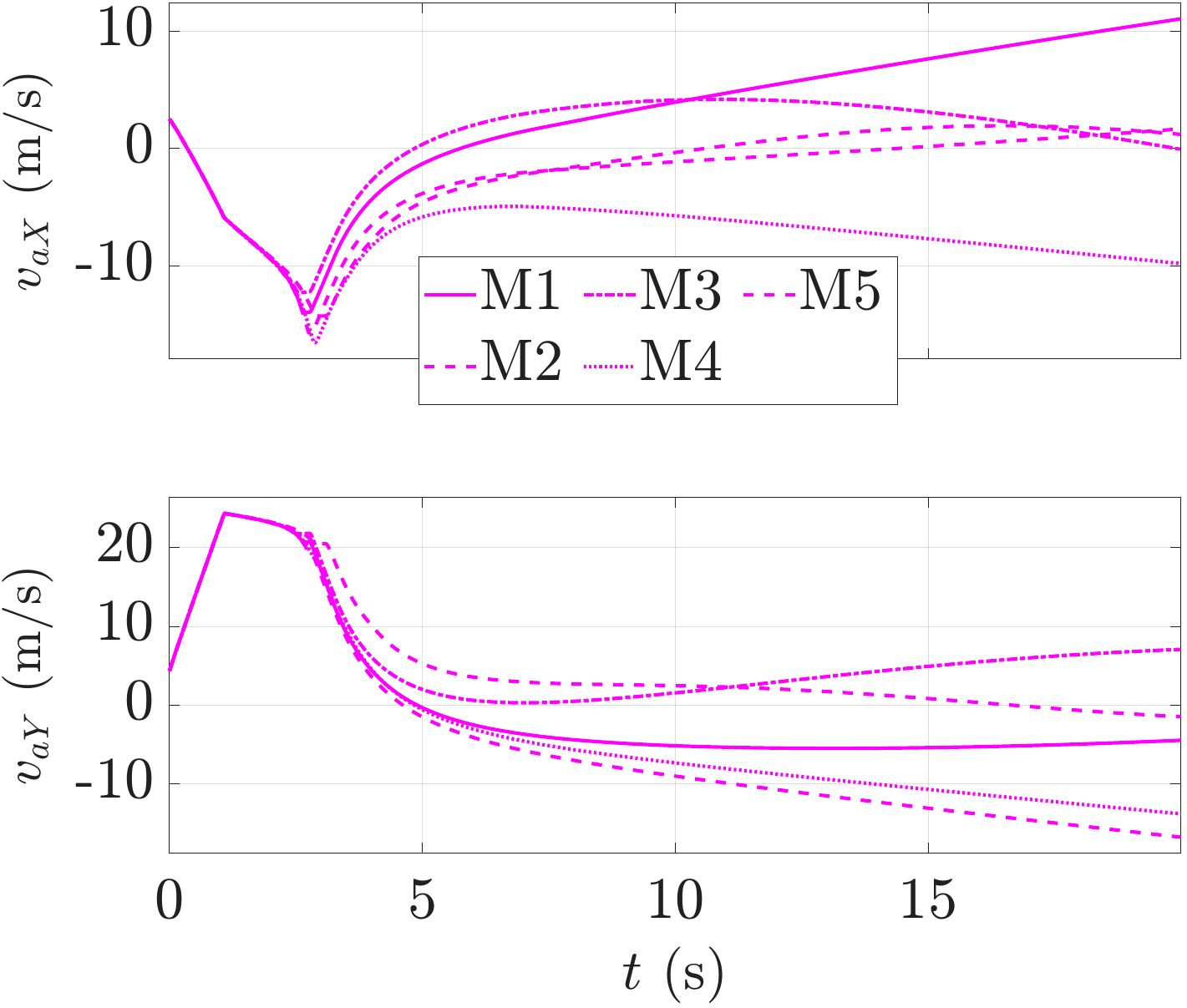}
		\caption{Velocity components of the UAV A.}
		\label{fig:man_Vel_UAV_A}
	\end{subfigure}%
	\begin{subfigure}{0.33\linewidth}
		\centering
		\includegraphics[width=\linewidth]{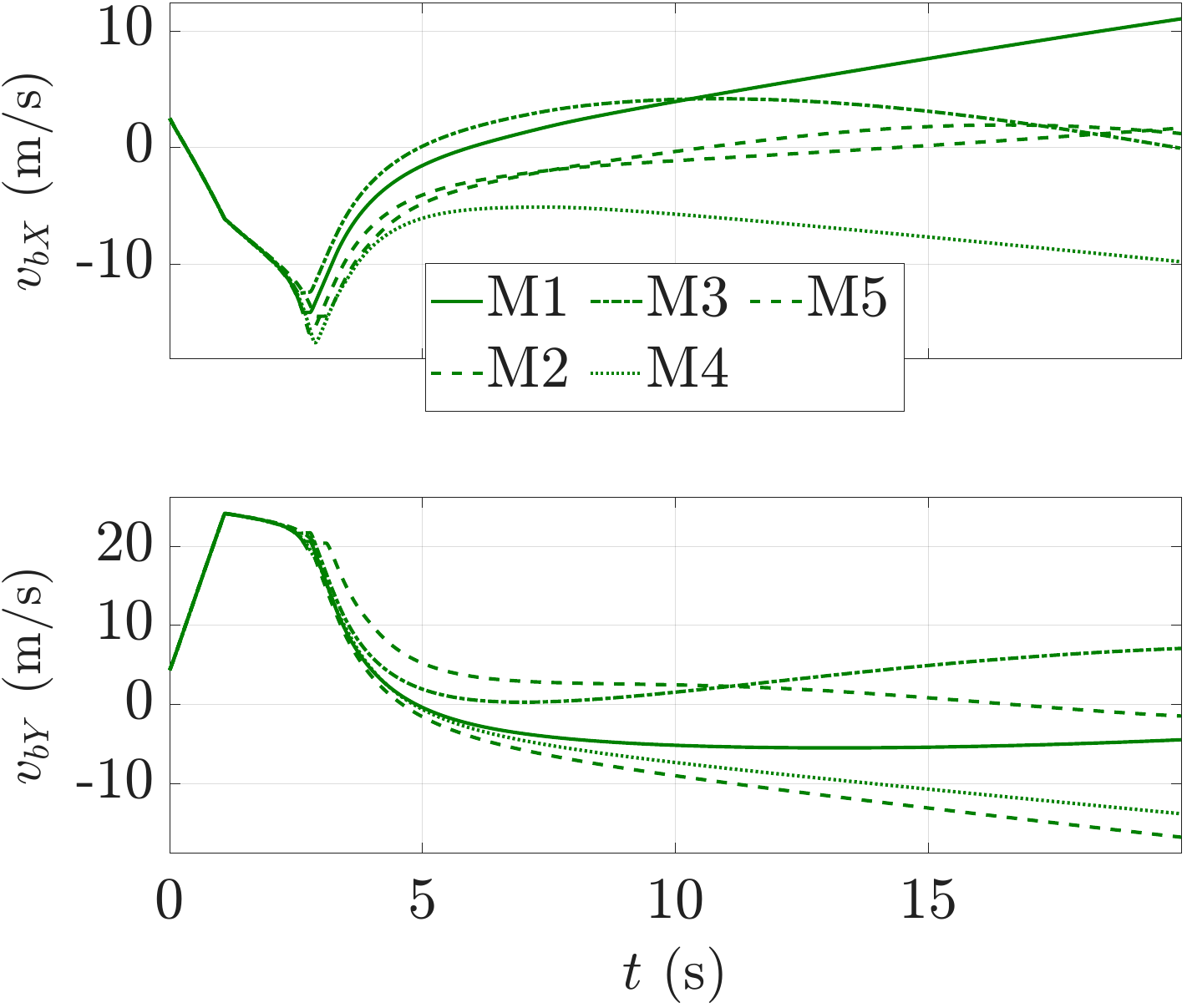}
		\caption{Velocity components of the UAV B.}
		\label{fig:man_Vel_UAV_B}
	\end{subfigure}
	\begin{subfigure}{0.33\linewidth}
		\centering
		\includegraphics[width=\linewidth]{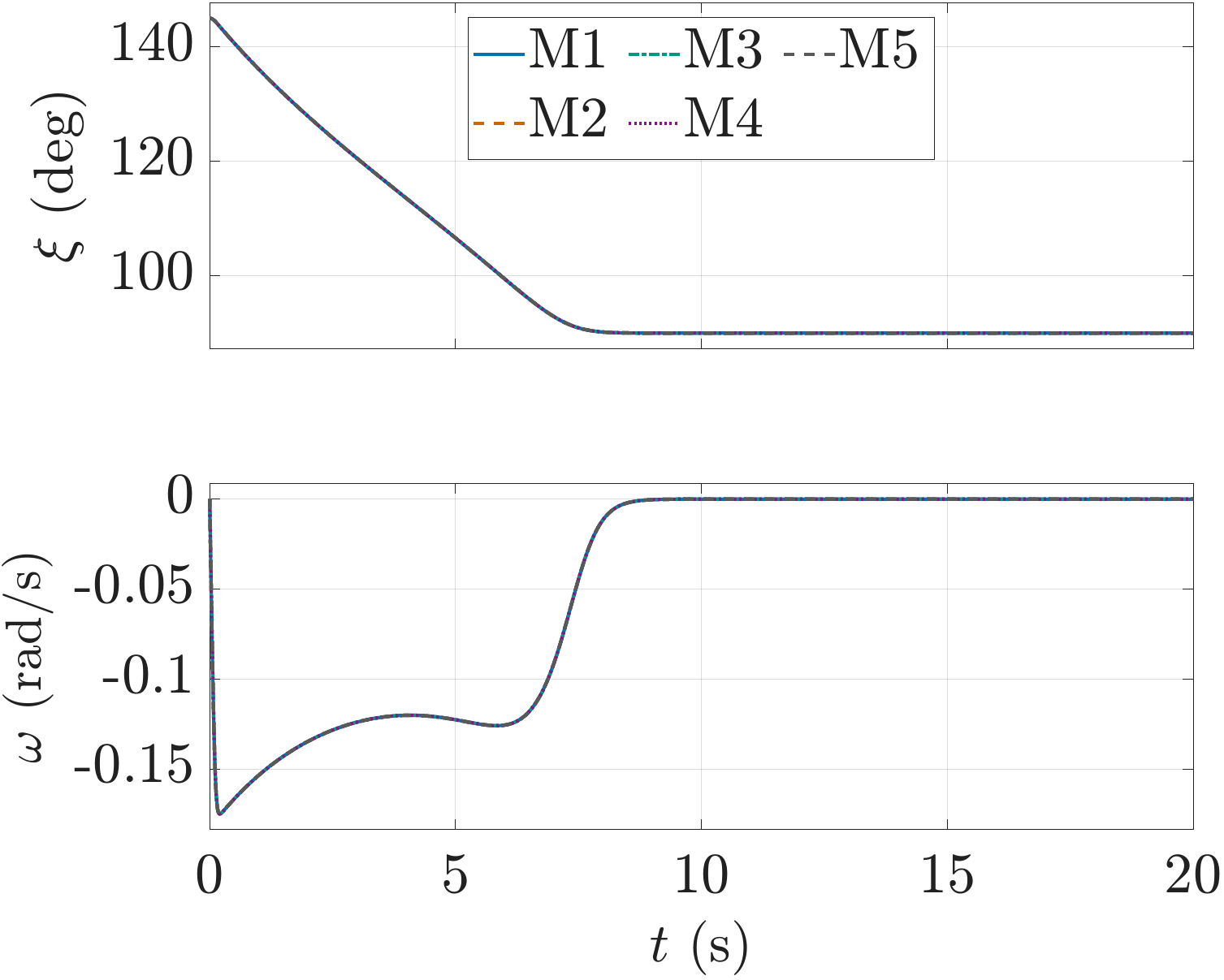}
		\caption{Link states.}
		\label{fig:man_Link_States}
	\end{subfigure}%
	\begin{subfigure}{0.33\linewidth}
		\centering
		\includegraphics[width=\linewidth]{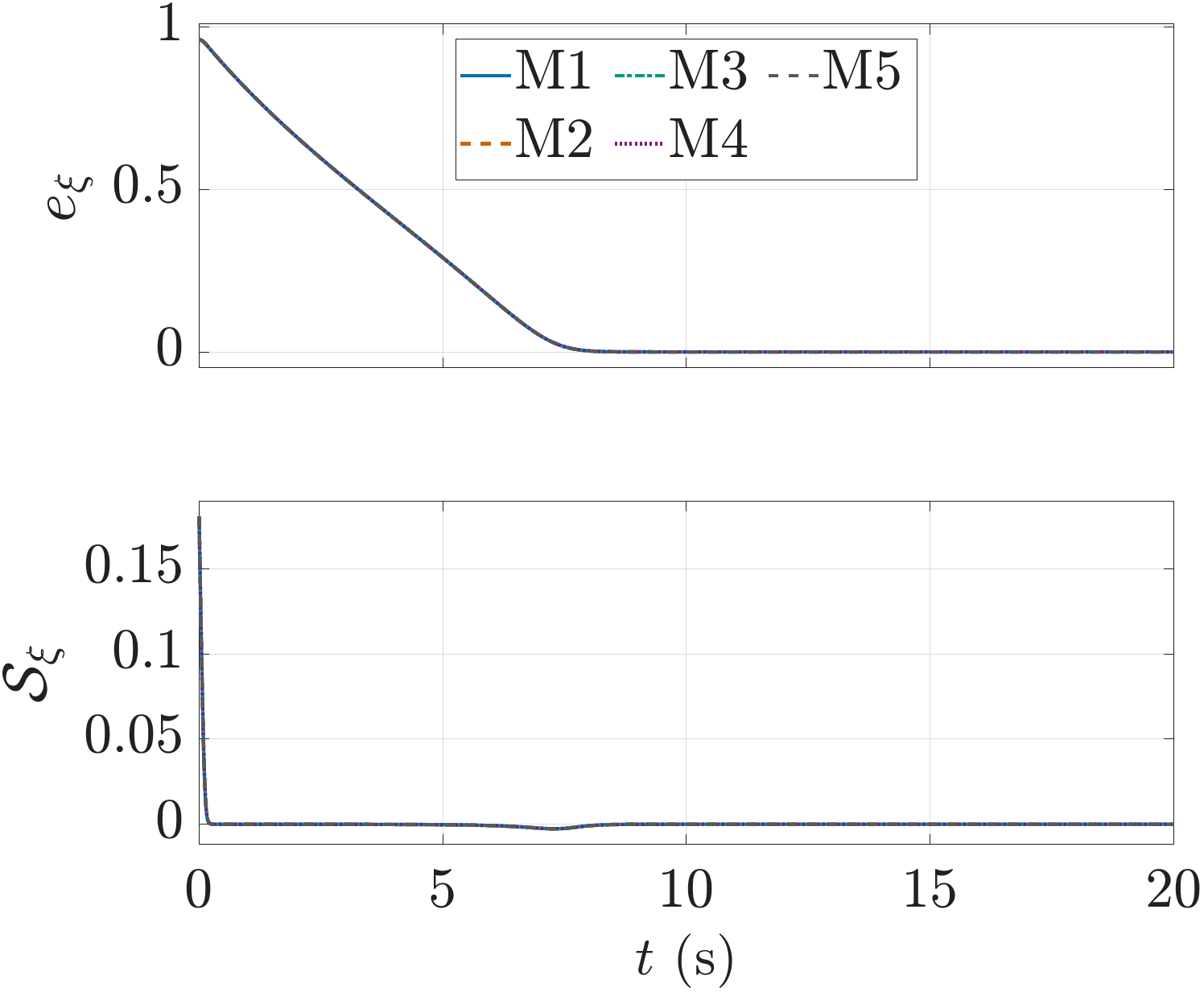}
		\caption{Error and sliding surface of the link.}
		\label{fig:man_Link_Sliding}
	\end{subfigure}%
	\begin{subfigure}{0.33\linewidth}
		\centering
		\includegraphics[width=\linewidth]{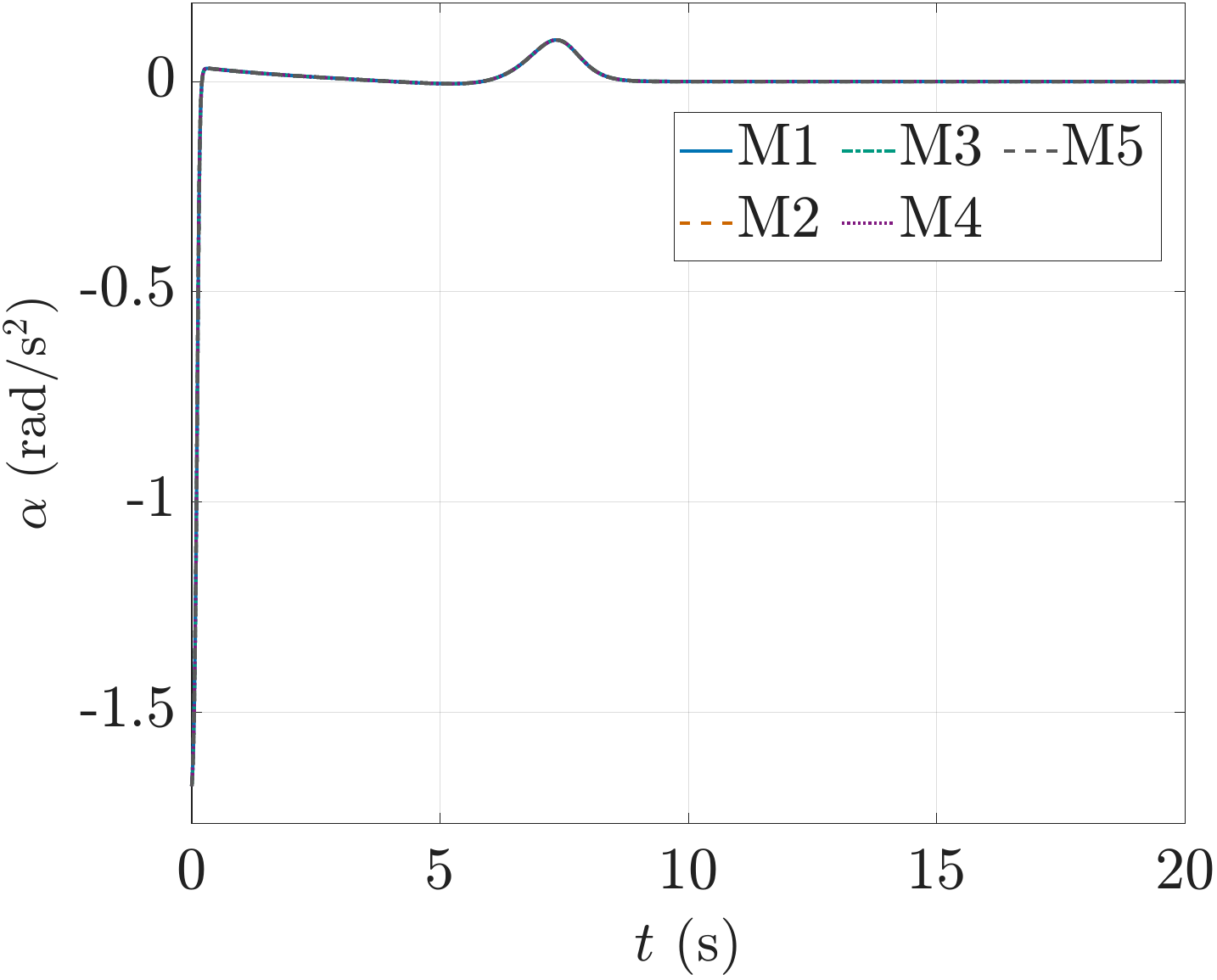}
		\caption{Control input to the link.}
		\label{fig:man_Link_Input}
	\end{subfigure}
	\caption{UAV delivering payload maneuvering platforms.}
	\label{fig:man}
\end{figure}
\subsection{Processor-in-loop simulation}
To assess the feasibility of real-time implementation and evaluate the computational efficiency of the proposed fixed-time cooperative guidance strategy, a Processor-in-the-Loop (PIL) simulation framework is developed. The closed-loop PIL architecture, illustrated in \Cref{fig:pil_diagram}, consists of a Host PC and an embedded target hardware. The coupled two-UAV-payload kinematics, platform dynamics, and relative engagement models are simulated on the Host PC running MATLAB/Simulink. The proposed fixed-time sliding mode guidance laws (\Cref{thm:stat_amr,thm:stat_amtheta,thm:man_amr,thm:man_amtheta}) and the link orientation controller (\Cref{sec:link_orientation}) are automatically converted to C/C++ code and deployed onto a Raspberry Pi embedded microprocessor using the MATLAB Simulink support package for Raspberry Pi hardware. 
\begin{figure}[!ht]
\centering
\begin{tikzpicture}[
	>={Stealth[scale=1.2]},
	block/.style={rectangle, draw=black, thick, fill=blue!5, text width=4cm, align=center, minimum height=1.5cm, rounded corners},
	plant/.style={rectangle, draw=black, thick, fill=green!5, text width=4.5cm, align=center, minimum height=2cm, rounded corners},
	comms/.style={rectangle, draw=gray, dashed, thick, text width=2.5cm, align=center, minimum height=1cm},
	line/.style={draw, thick, ->}
	]
	\node[plant] (pc) at (0,0) {\textbf{Host PC\\ (MATLAB/Simulink)}\\ \vspace{0.2cm} {\small Two-UAV system dynamics}};
	\node[block] (mcu) at (8.5,0) {\textbf{Raspberry Pi}\\ \vspace{0.2cm} {\small Guidance Laws}};
	\node[comms] (tx) at (4.25, 1.5) {Ethernet (UDP) \\ (State Data)};
	\node[comms] (rx) at (4.25, -1.5) {Ethernet (UDP) \\ (Control Commands)};
	\draw[line] (pc.north) |- (tx.west);
	\draw[line] (tx.east) -| (mcu.north) node[pos=0.8, right] {($r, \theta, v_r, v_\theta, \xi$)};
	\draw[line] (mcu.south) |- (rx.east);
	\draw[line] (rx.west) -| (pc.south) node[pos=0.5, left] {($a_{aX}, a_{aY}, a_{bX}, a_{bY}, \alpha$)};
	\node[draw, dashed, gray, fit=(tx) (rx), inner sep=0.5cm, label=above:{\textit{Communication Interface}}] {};
\end{tikzpicture}
\caption{Closed-loop architecture of the Processor-in-the-Loop (PIL) simulation.}
\label{fig:pil_diagram}
\end{figure}
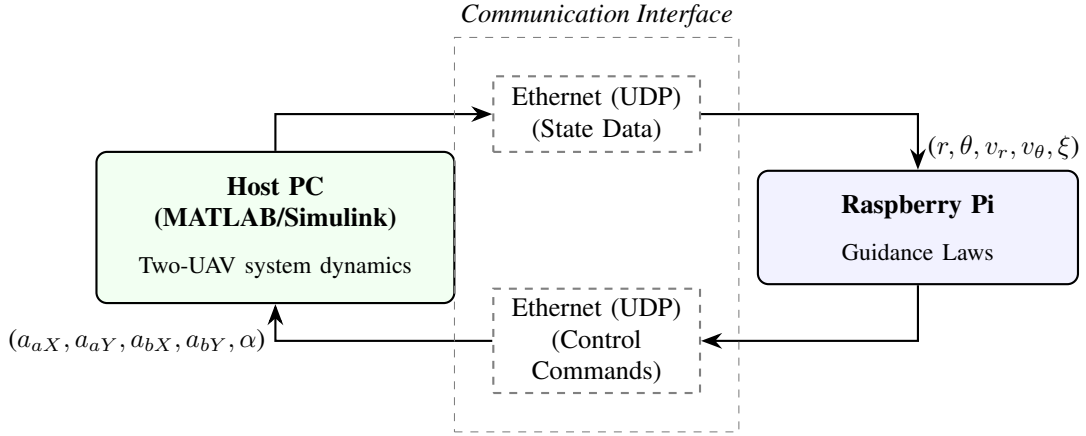

During execution, the Host PC continuously calculates the current engagement states $\mathcal{C}$ (including $r$, $\theta$, $v_r$, $v_\theta$, and link orientation $\xi$) and transmits them to the Raspberry Pi over an Ethernet (UDP) connection. The embedded target computes the required equivalent agent acceleration commands ($a_{uX}, a_{uY}$) and the link angular acceleration ($\alpha$), and subsequently maps them to individual UAV commands ($a_{aX}, a_{aY}, a_{bX}, a_{bY}$) using the allocation strategy defined in \eqref{eqn:compamab}. These computed commands are then transmitted back to the Host PC via Ethernet (UDP) to propagate the plant states for the next time step. The sampling time is fixed as $T_{s} = 0.01$ s. Thus, the available computation time per control cycle is $10$ ms. The initial engagement geometry and controller parameters are kept identical to those used for the desired landing condition of $\theta_{\rm L}=-90^\circ$.

The comparative performance of the pure Simulink and PIL simulations is presented in \Cref{fig:PIL}. As shown in \Cref{fig:PIL}, the UAV successfully delivers the payload to the stationary platform while exhibiting behavior consistent with that observed in the pure Simulink simulation, even when the controller is executed on the Raspberry Pi hardware. The discrepancies between the pure simulation and Raspberry Pi implementation are defined as $z_{1}=r_{\rm PIL}-r_{\rm SIL}$, $z_{2}=e_{\theta,\rm PIL}-e_{\theta,\rm SIL}$, and $z_{3}=e_{\xi,\rm PIL}-e_{\xi,\rm SIL}$. As shown in the \Cref{fig:RPI_Stat_SIL_PIL_Error.eps}, these discrepancies remain small throughout the maneuver and converge to zero. It can be observed from \Cref{fig:RPI_Stat_Computation_Time} that the average execution time of the control loop on the Raspberry Pi is $0.0073$ ms, which is substantially shorter than the $10$ ms sampling interval. The measured worst-case execution time (WCET) is $0.867$ ms, while $99\%$ of the execution times remain below $0.0277$ ms. This, in turn, demonstrates that the proposed guidance strategy can be executed within the available sampling period, which makes it a potential candidate for real-time implementation.
\begin{figure}[!ht]
\centering
\begin{subfigure}{0.33\linewidth}
	\centering
	\includegraphics[width=\linewidth]{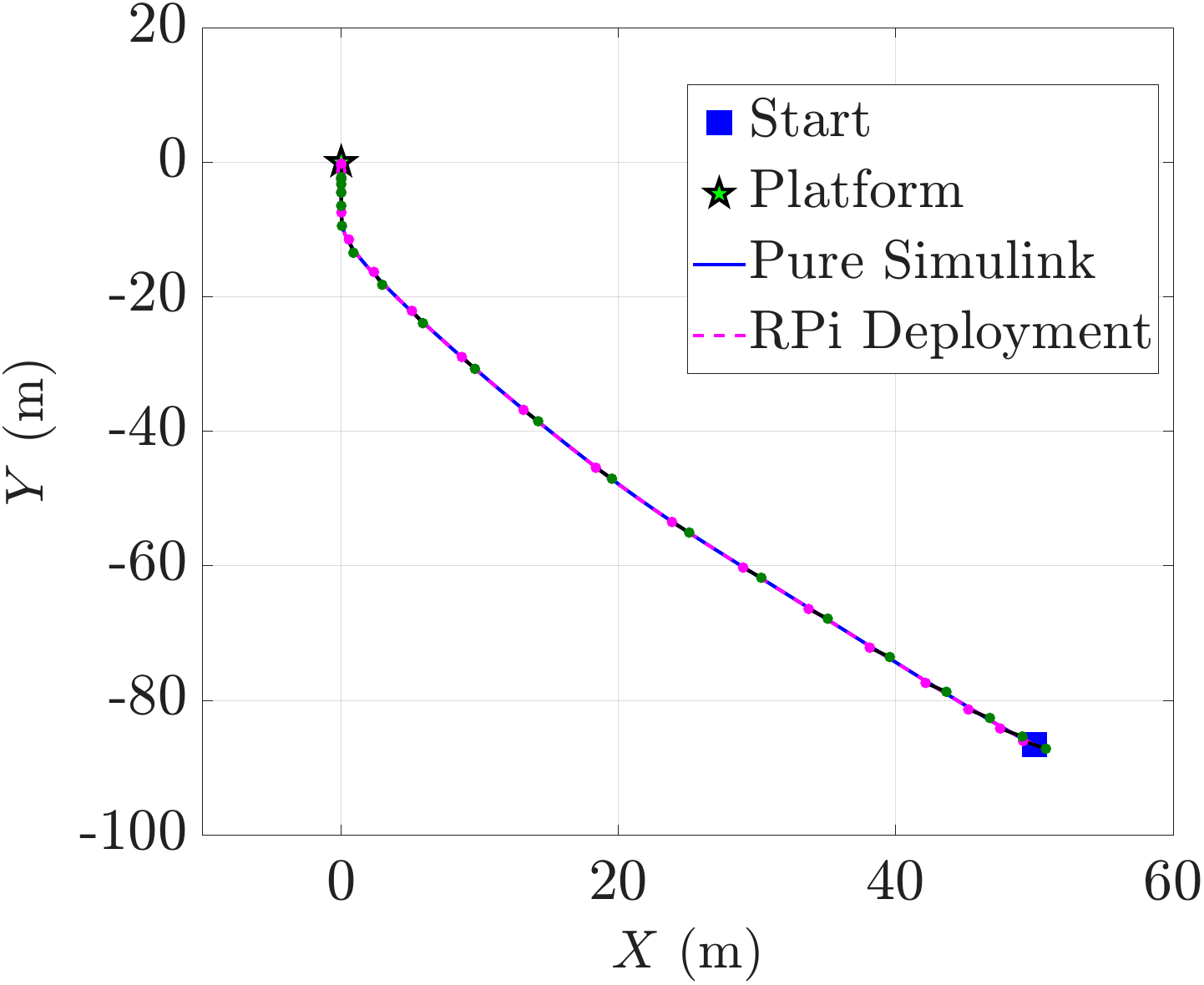}
	\caption{Trajectory.}
	\label{fig:RPI_Stat_Trajectory}
\end{subfigure}%
\begin{subfigure}{0.33\linewidth}
	\centering
	\includegraphics[width=\linewidth]{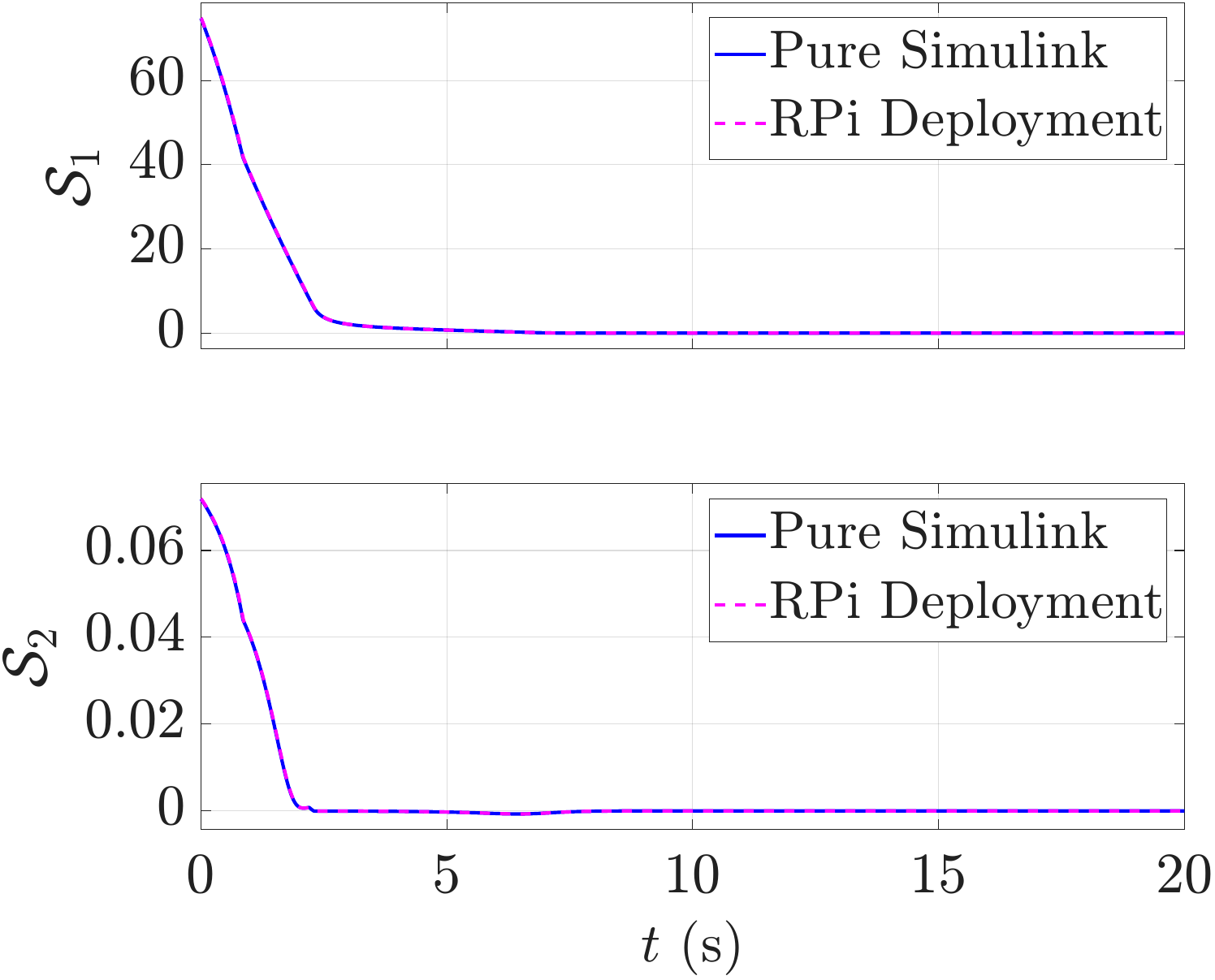}
	\caption{Sliding surfaces.}
	\label{fig:RPI_Stat_Sliding_Surfaces}
\end{subfigure}%
\begin{subfigure}{0.33\linewidth}
	\centering
	\includegraphics[width=\linewidth]{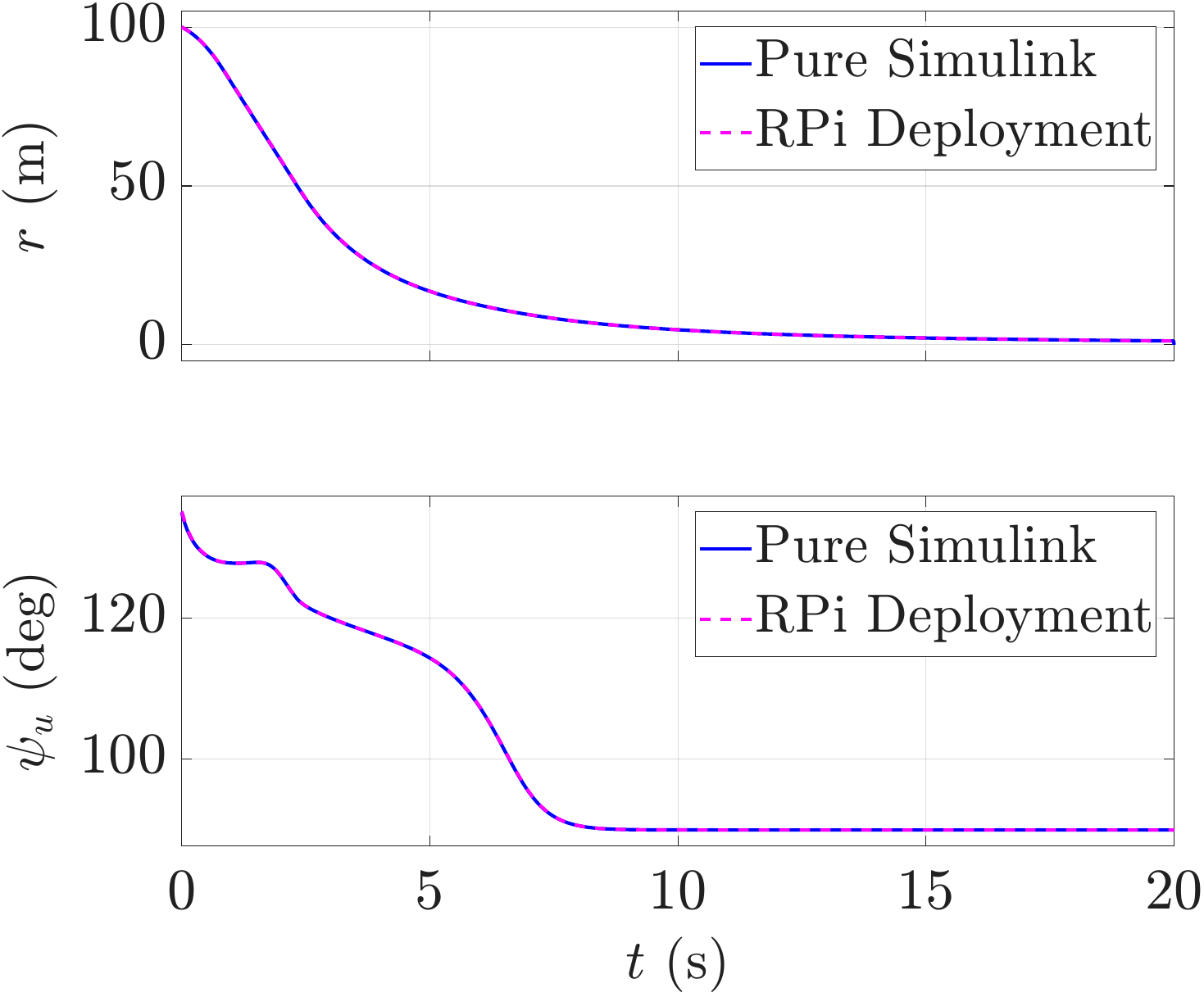}
	\caption{Relative range and heading angle.}
	\label{fig:RPI_Stat_Range_Heading}
\end{subfigure}
\begin{subfigure}{0.33\linewidth}
	\centering
	\includegraphics[width=\linewidth]{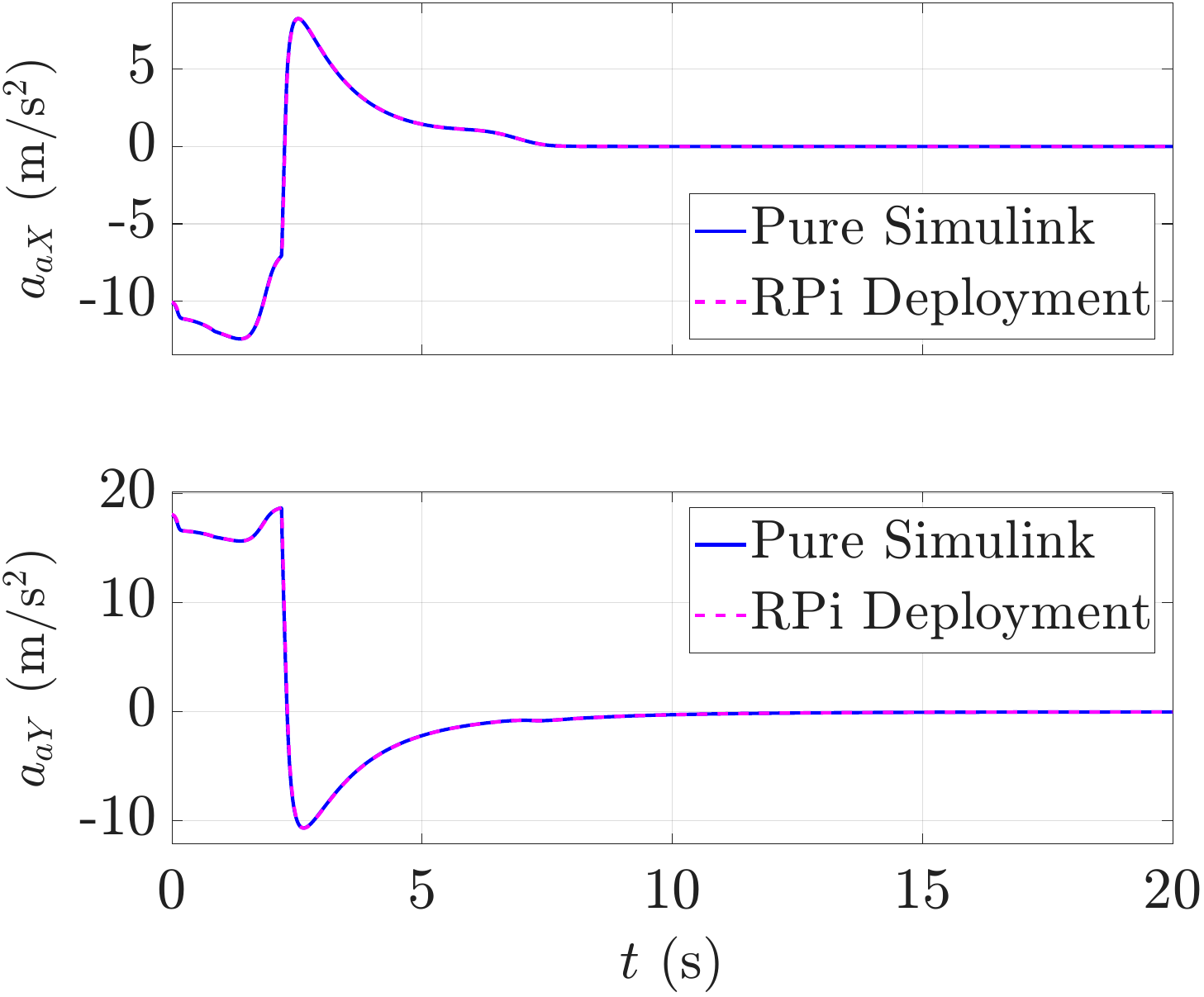}
	\caption{Control input to the equivalent agent U.}
	\label{fig:RPI_Stat_Acc_UAV_A}
\end{subfigure}%
\begin{subfigure}{0.33\linewidth}
	\centering
	\includegraphics[width=\linewidth]{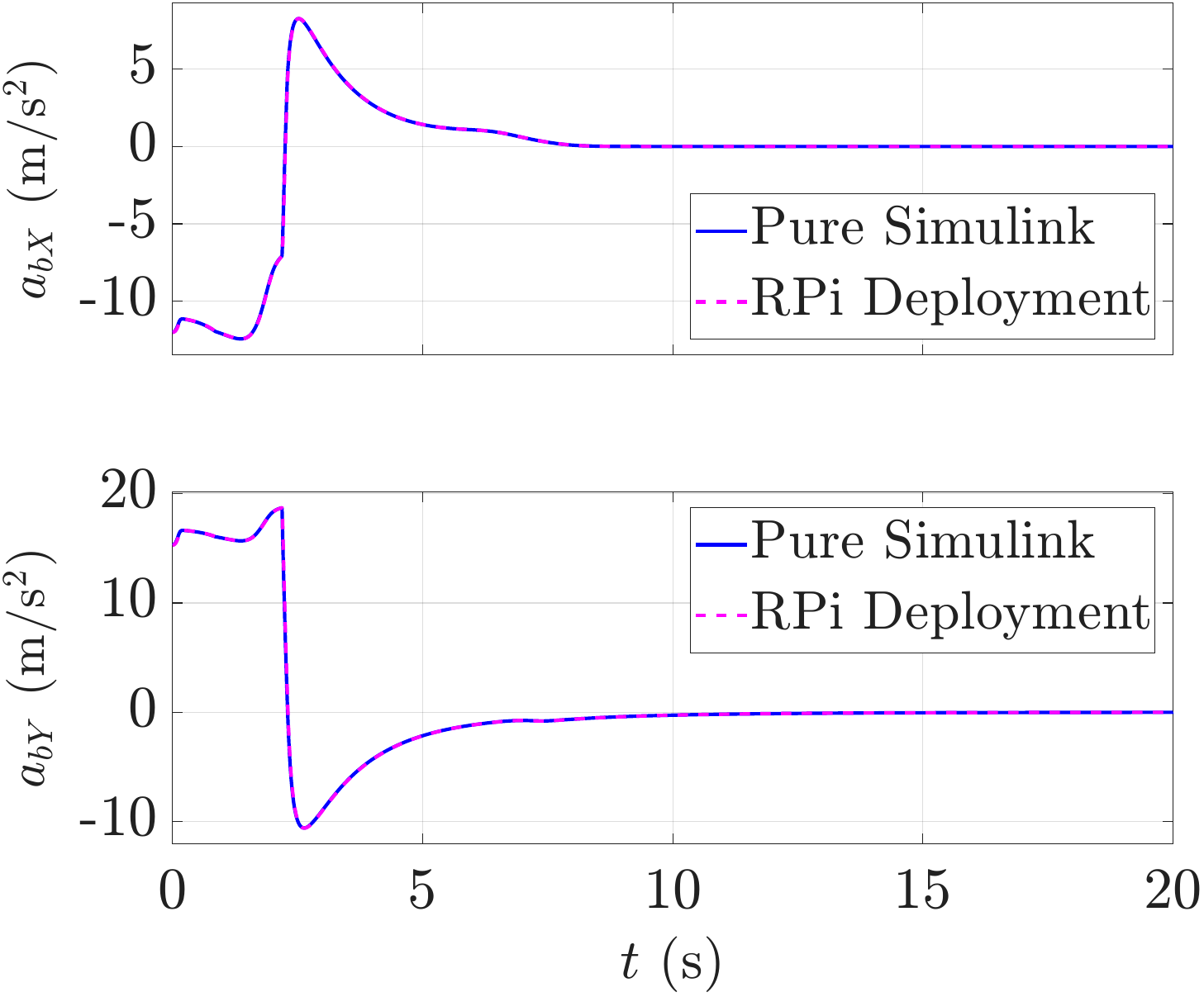}
	\caption{Control input to the UAV B.}
	\label{fig:RPI_Stat_Acc_UAV_B}
\end{subfigure}%
\begin{subfigure}{0.33\linewidth}
	\centering
	\includegraphics[width=\linewidth]{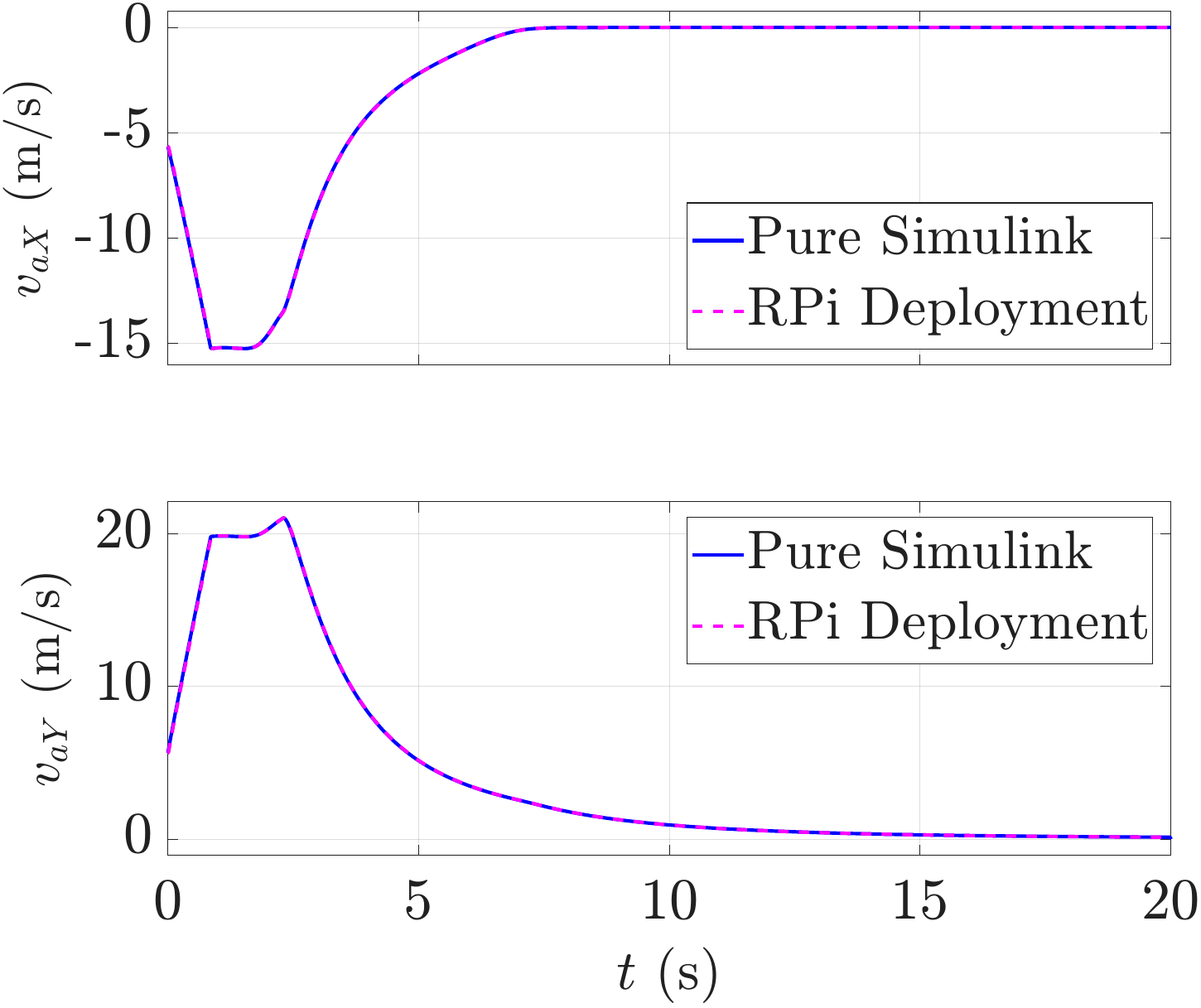}
	\caption{Velocity components of the UAV A.}
	\label{fig:RPI_Stat_Vel_UAV_A}
\end{subfigure}
\begin{subfigure}{0.33\linewidth}
	\centering
	\includegraphics[width=\linewidth]{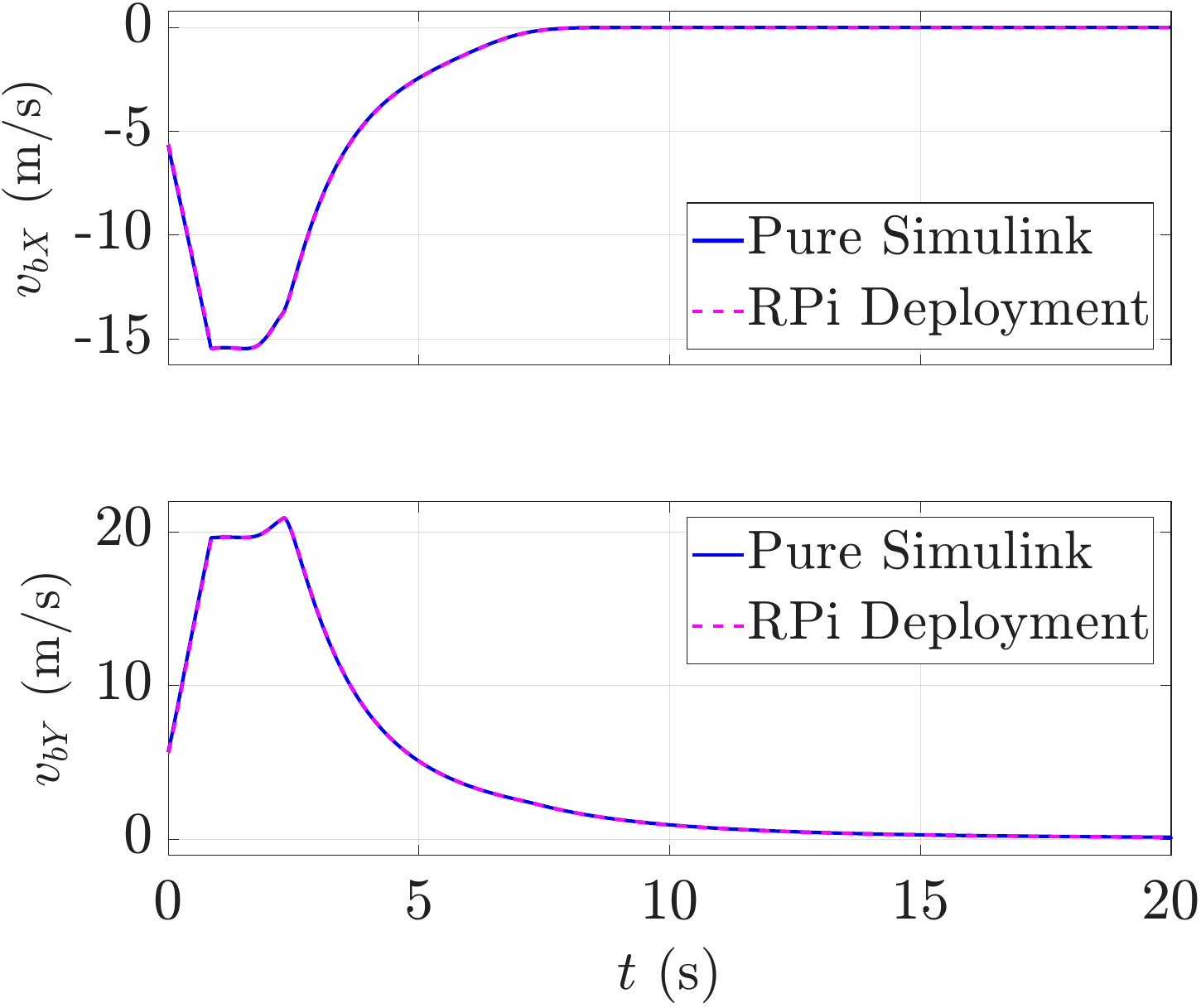}
	\caption{Velocity components of the UAV B.}
	\label{fig:RPI_Stat_Vel_UAV_B}
\end{subfigure}%
\begin{subfigure}{0.33\linewidth}
	\centering
	\includegraphics[width=\linewidth]{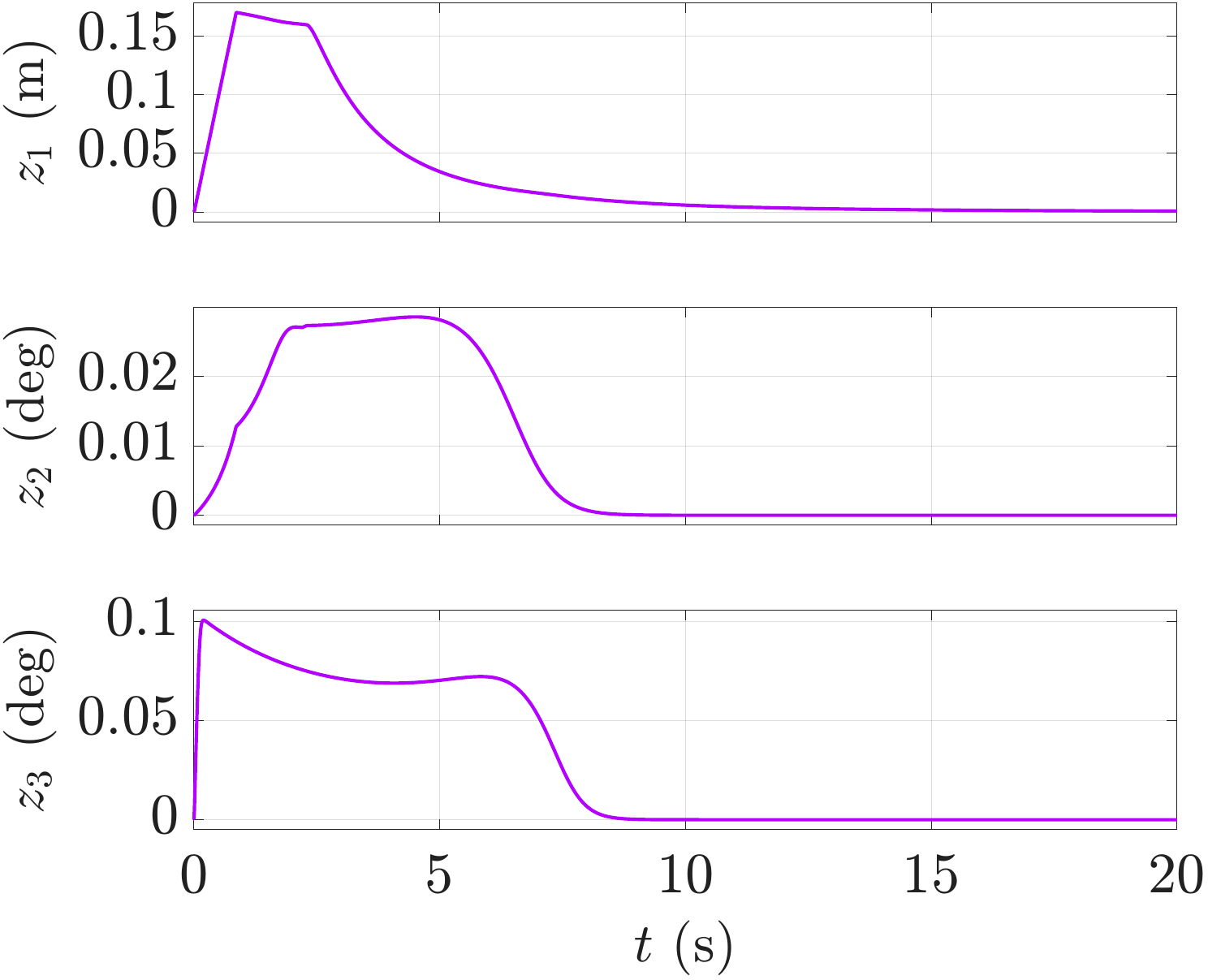}
	\caption{Error discrepancy b/w Simulink and PIL.}
	\label{fig:RPI_Stat_SIL_PIL_Error.eps}
\end{subfigure}%
\begin{subfigure}{0.33\linewidth}
	\centering
	\includegraphics[width=\linewidth]{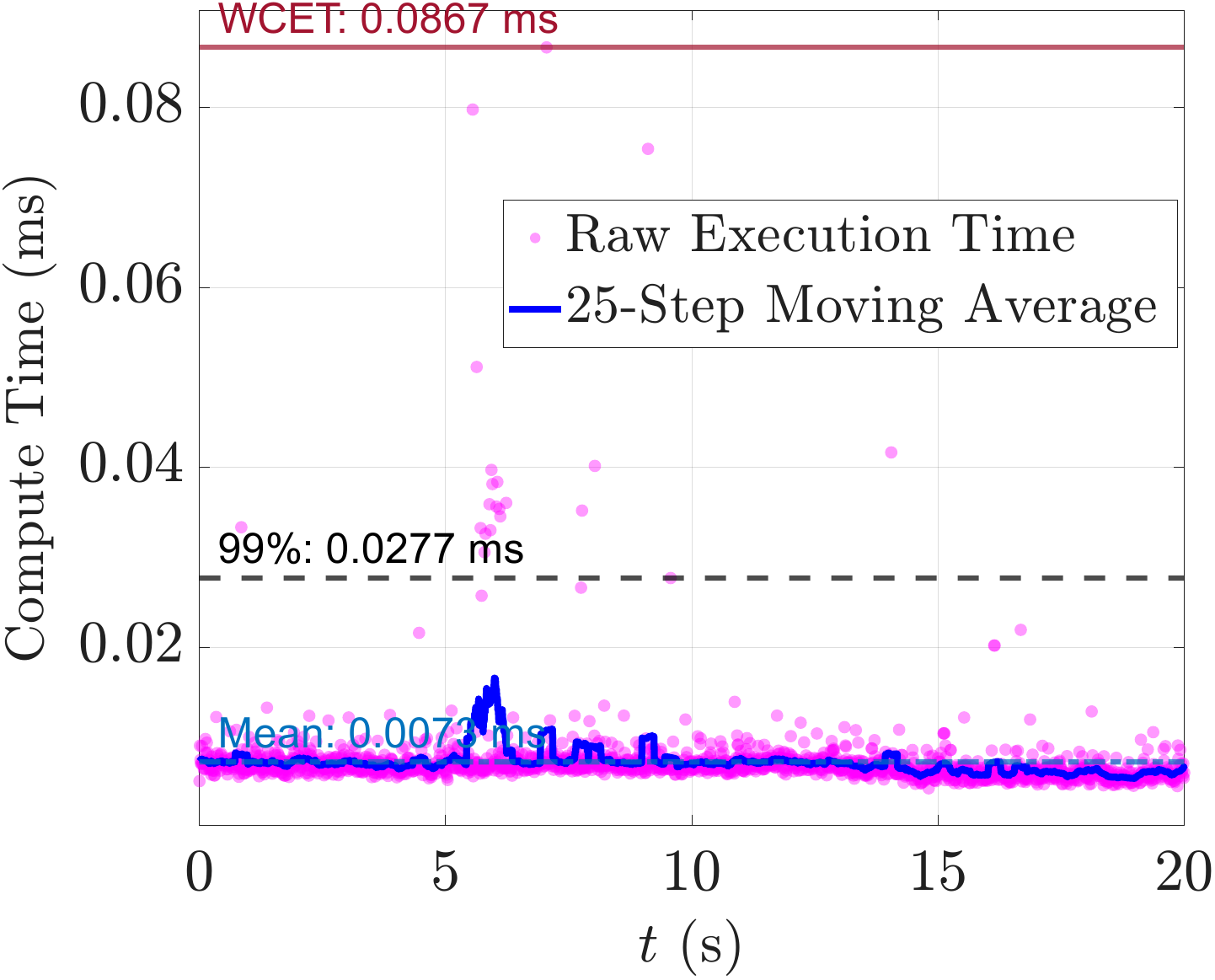}
	\caption{Measured Raspberry Pi execution time.}
	\label{fig:RPI_Stat_Computation_Time}
\end{subfigure}
\caption{PIL validation with Raspberry Pi.}
\label{fig:PIL}
\end{figure}
\subsection{Comparative performance analysis}
To benchmark the efficacy of the proposed strategy, we compare its performance with that of the method discussed in \cite{doi:10.2514/1.G005098}. The simulation setup is kept the same as in the stationary case with four distinct landing angles. The quantitative assessment of the controllers is carried out using eight distinct performance indices: Convergence times ($\mathcal{T}(s_{r 0})$ and $\mathcal{T}(s_{\varrho 0})$): The time required for the range sliding surface ($\mathcal{}{S}_1$) and the LOS angle sliding surface ($\mathcal{S}_2$) to reach and stay within a boundary threshold ($\epsilon = 0.001$) around the sliding manifold. Tracking accuracy ($\mathrm{RMSE}_r$ and $\mathrm{RMSE}_\varrho$): The root mean square error for the relative range and LOS angle, measuring the total deviation from the desired approach trajectory. Virtual control effort of equivalent agent ($\mathrm{IACE}$ and $\mathrm{ISCE}$): The integral of absolute control effort ($\int \sqrt{a_{UX}^2 + a_{UY}^2} \mathrm{d}t$) and the integral of squared control effort ($\int (a_{UX}^2 + a_{UY}^2) \mathrm{d}t$) represent the aggregate acceleration commands generated for the virtual center of mass (CoM). Individual UAV control energy ($E_A$ and $E_B$): The total physical control energy expended by each UAV, calculated as $\int [(a_{aX})^2 + (a_{aY})^2]\mathrm{d}t$ for UAV A and similarly for UAV B.

The results of this comparative study are presented in \Cref{fig:stat_comp} as a series of bar charts and a normalized radar graph. One can observe in \Cref{fig:Stat_Trajectory_rao} that the UAV successfully delivers the payload to the stationary platform in all four scenarios, even when using the guidance strategy discussed in \cite{doi:10.2514/1.G005098}. The sliding surfaces and linear speed profiles of the equivalent agent are shown in \Cref{fig:Stat_Sliding_Surfaces_rao,fig:Stat_Velocity_rao}. One can observe that sliding mode is imposed on the sliding surfaces $\mathcal{S}_{1}$ and $\mathcal{S}_{2}$ at distinct times in all cases. As sliding mode is imposed on the sliding surfaces in all cases, the UAV system, in fact, delivers the payload at the desired landing angles. \Cref{fig:Metric_T_Slide_S1,fig:Metric_T_Slide_S2} illustrate the convergence times for the sliding surfaces. Across all landing angles, the proposed controller achieves faster convergence to the sliding manifolds than the method in \cite{doi:10.2514/1.G005098}. Additionally, the convergence time remains almost uniform for the proposed method, whereas it varies for the method in \cite{doi:10.2514/1.G005098}. Furthermore, the tracking performance, depicted in \Cref{fig:Metric_RMSE_r,fig:Metric_RMSE_th}, reveals that the proposed method consistently yields lower $\mathrm{RMSE}_r$ and $\mathrm{RMSE}_\varrho$ values, which ensure tighter trajectory adherence and minimize deviation during the critical approach phase.  As shown in \Cref{fig:Metric_IACE,fig:Metric_ISCE}, the proposed method demonstrates a noticeable reduction in both the IACE and ISCE for the equivalent agent. When this virtual effort is mapped to the physical UAVs through the control allocation scheme, the individual energy consumption indices ($E_A$ and $E_B$) are shown in \Cref{fig:Metric_E_U1,fig:Metric_E_U2}. One can observe that both UAV $A$ and UAV $B$ require significantly less energy with the proposed strategy than the method in \cite{doi:10.2514/1.G005098}. This, in turn, improves the UAV system's endurance. To provide a holistic view of the comparative assessment, a normalized radar chart is presented in \Cref{fig:Metric_Spider}. For this graph, the average values of all eight metrics are normalized such that the maximum (worst) value for each axis is scaled to $1.0$. A smaller enclosed polygon area indicates superior overall performance. The proposed method visually occupies a strictly smaller footprint than the \cite{doi:10.2514/1.G005098} across all axes. This demonstrates that the proposed cooperative delivery strategy achieves faster convergence and higher tracking precision while simultaneously demanding less control energy from the physical UAVs.
\begin{figure}[!ht]
	\centering
	\begin{subfigure}{0.33\linewidth}
		\centering
		\includegraphics[width=\linewidth]{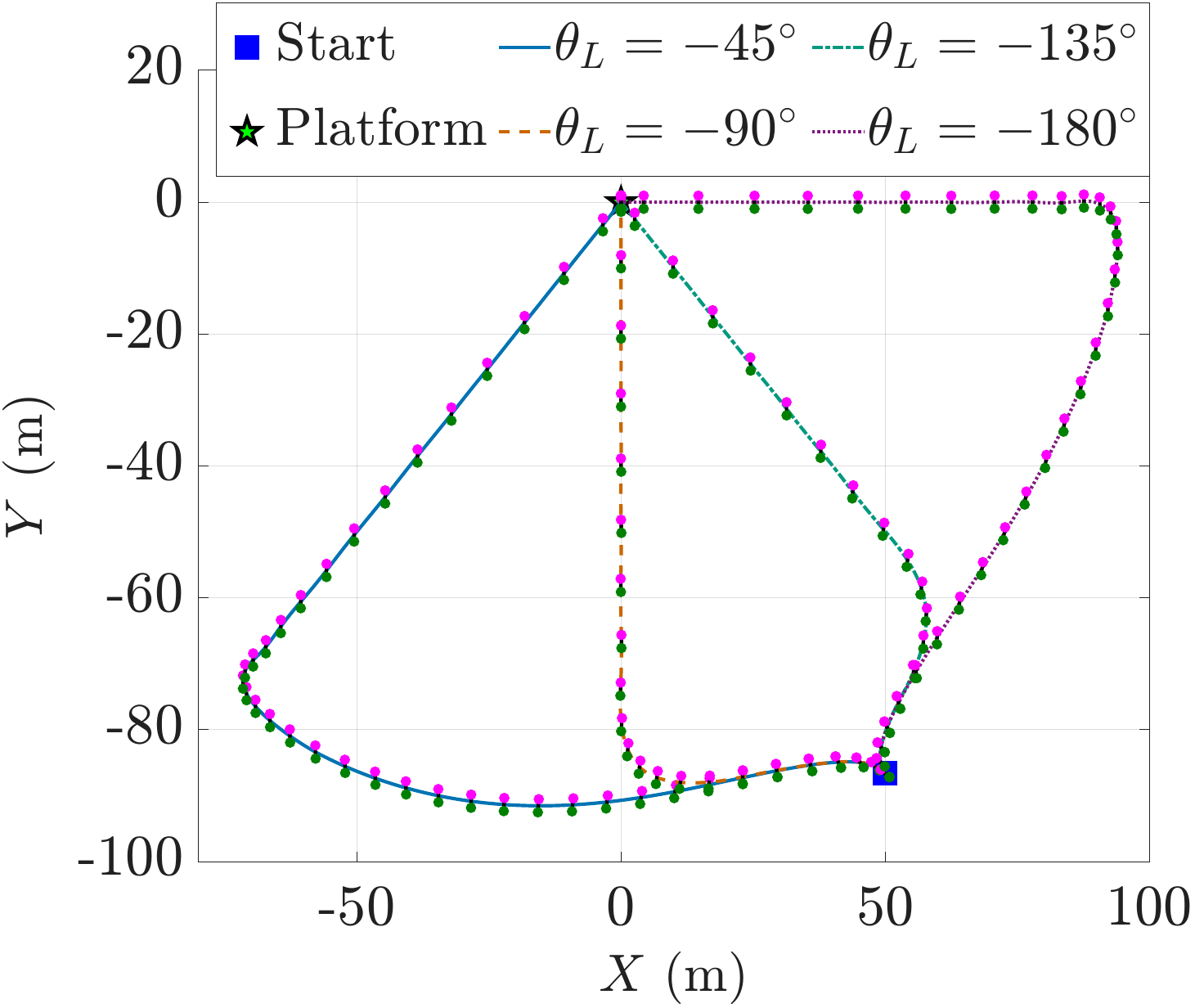}
		\caption{Trajectory.}
		\label{fig:Stat_Trajectory_rao}
	\end{subfigure}%
	\begin{subfigure}{0.33\linewidth}
		\centering
		\includegraphics[width=\linewidth]{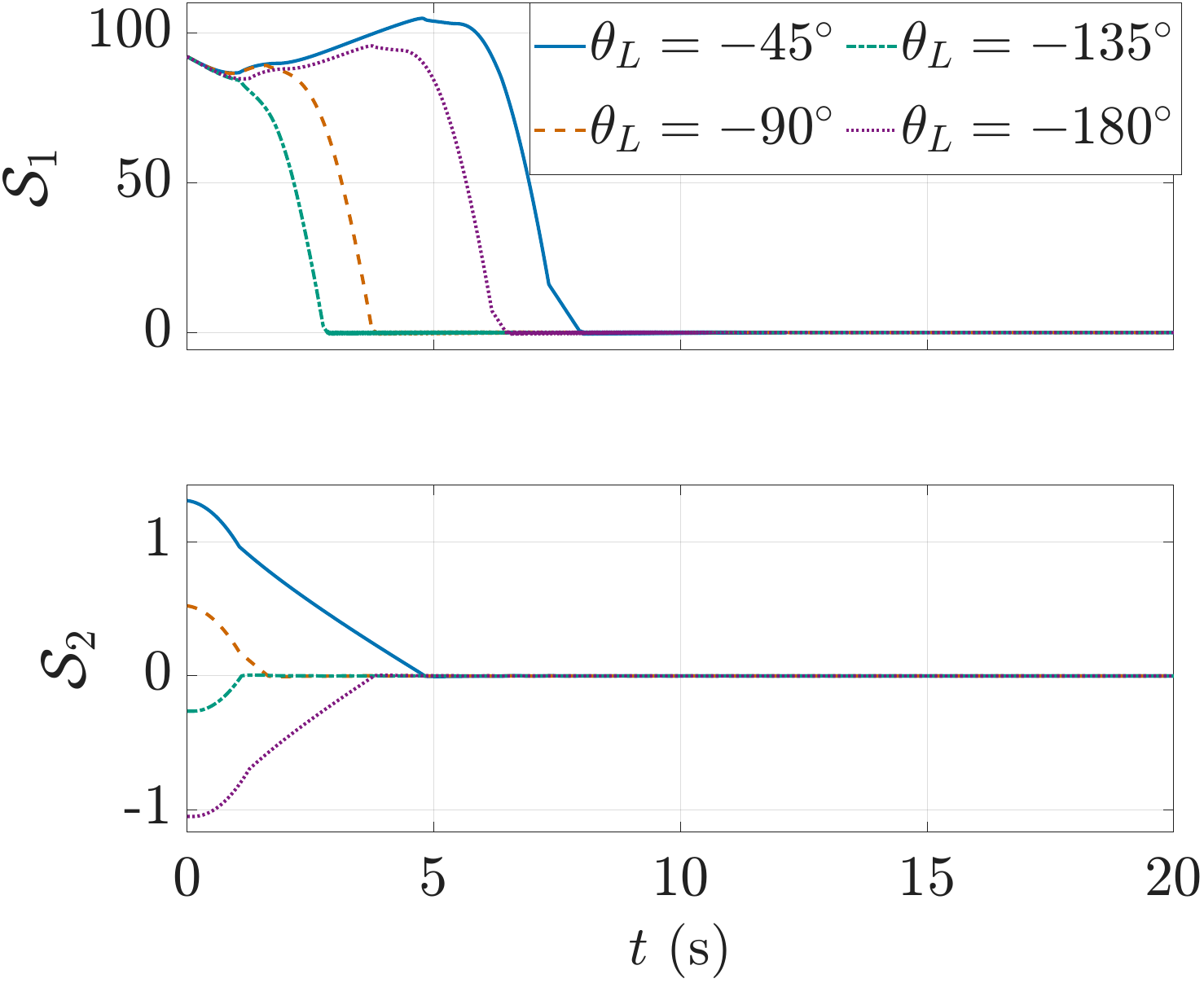}
		\caption{Sliding surfaces.}
		\label{fig:Stat_Sliding_Surfaces_rao}
	\end{subfigure}%
	\begin{subfigure}{0.33\linewidth}
		\centering
		\includegraphics[width=\linewidth]{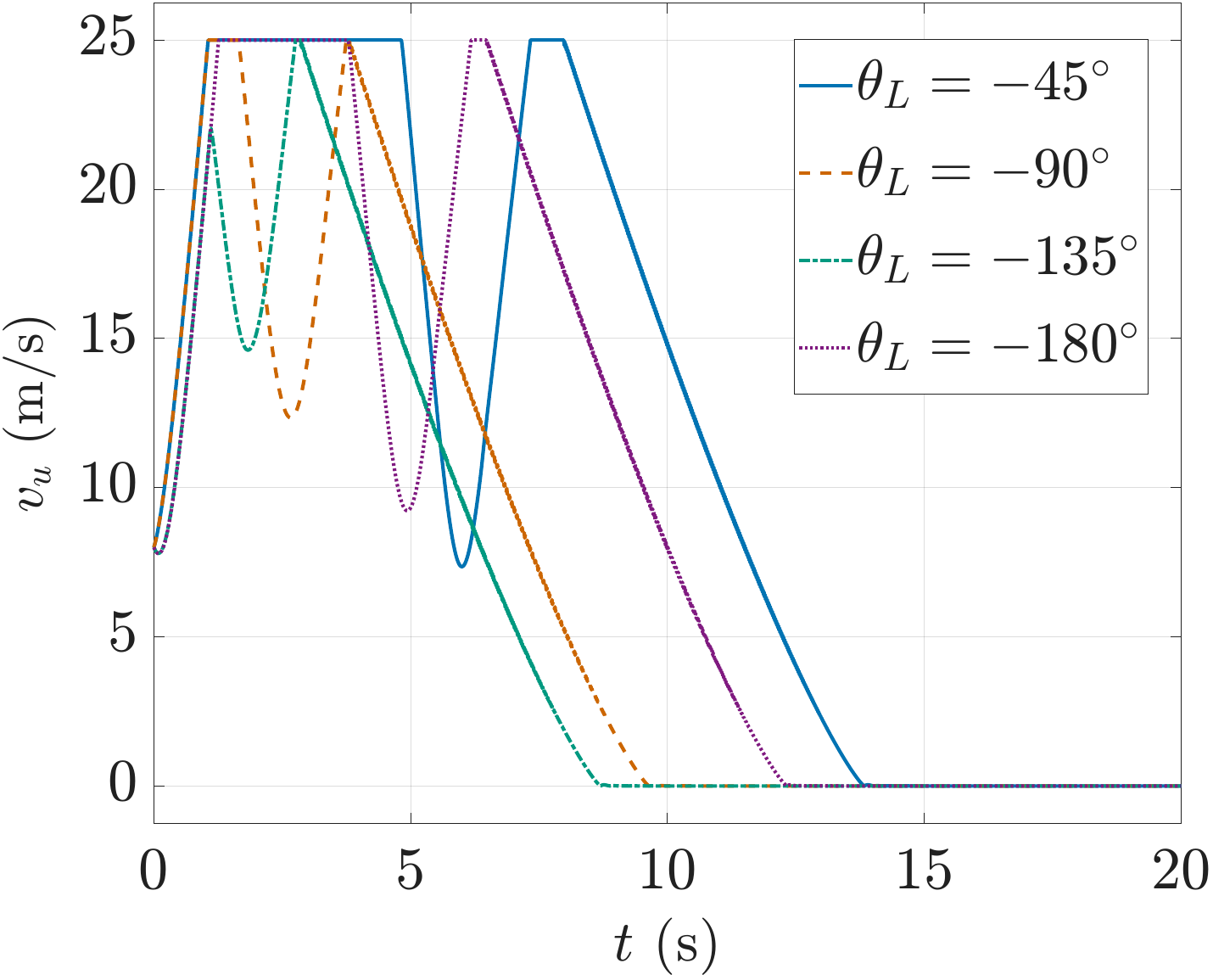}
		\caption{Relative range and heading angle.}
		\label{fig:Stat_Velocity_rao}
	\end{subfigure}
	\begin{subfigure}{0.33\linewidth}
		\centering
		\includegraphics[width=\linewidth]{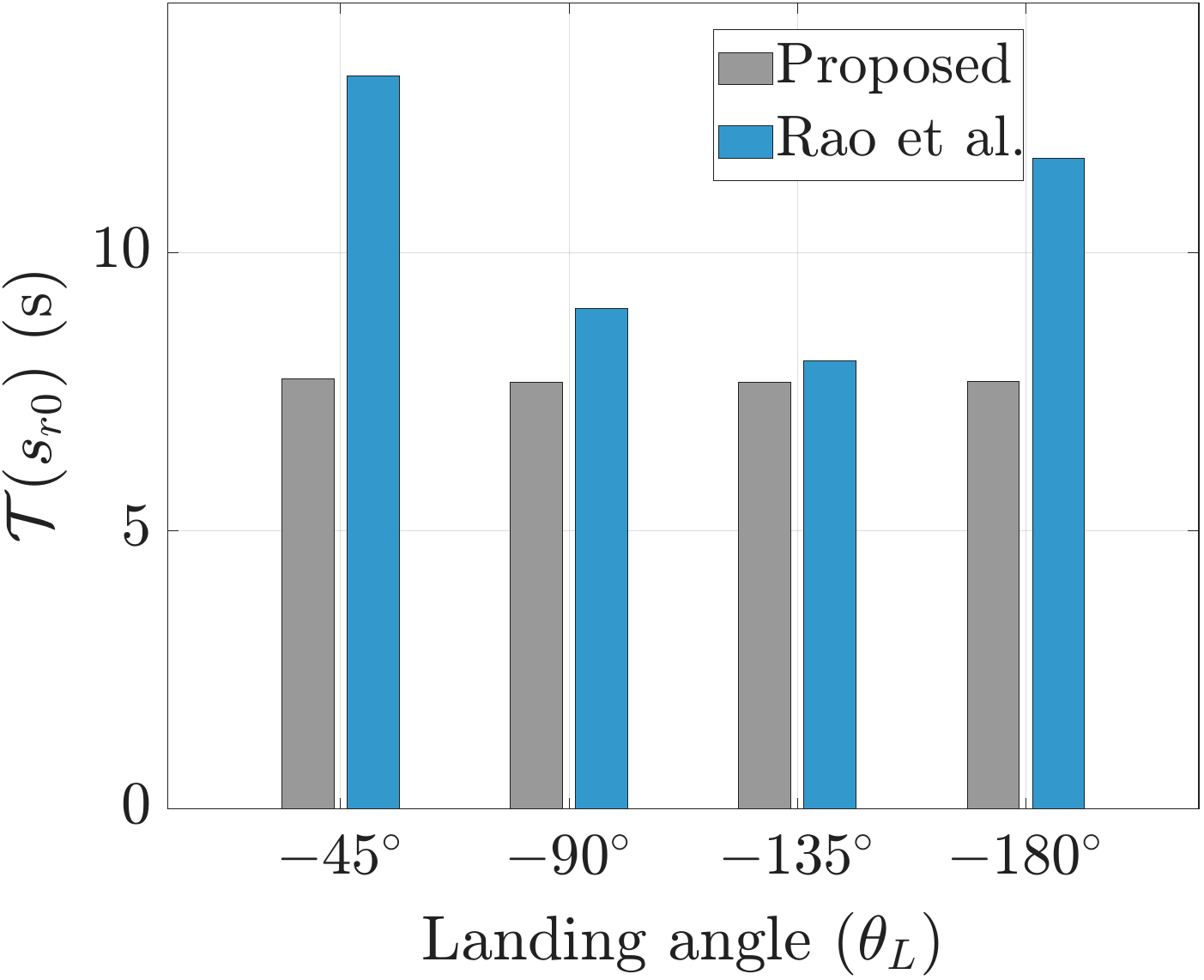}
		\caption{Convergence time to $\mathcal{S}_{1}$.}
		\label{fig:Metric_T_Slide_S1}
	\end{subfigure}%
	\begin{subfigure}{0.33\linewidth}
		\centering
		\includegraphics[width=\linewidth]{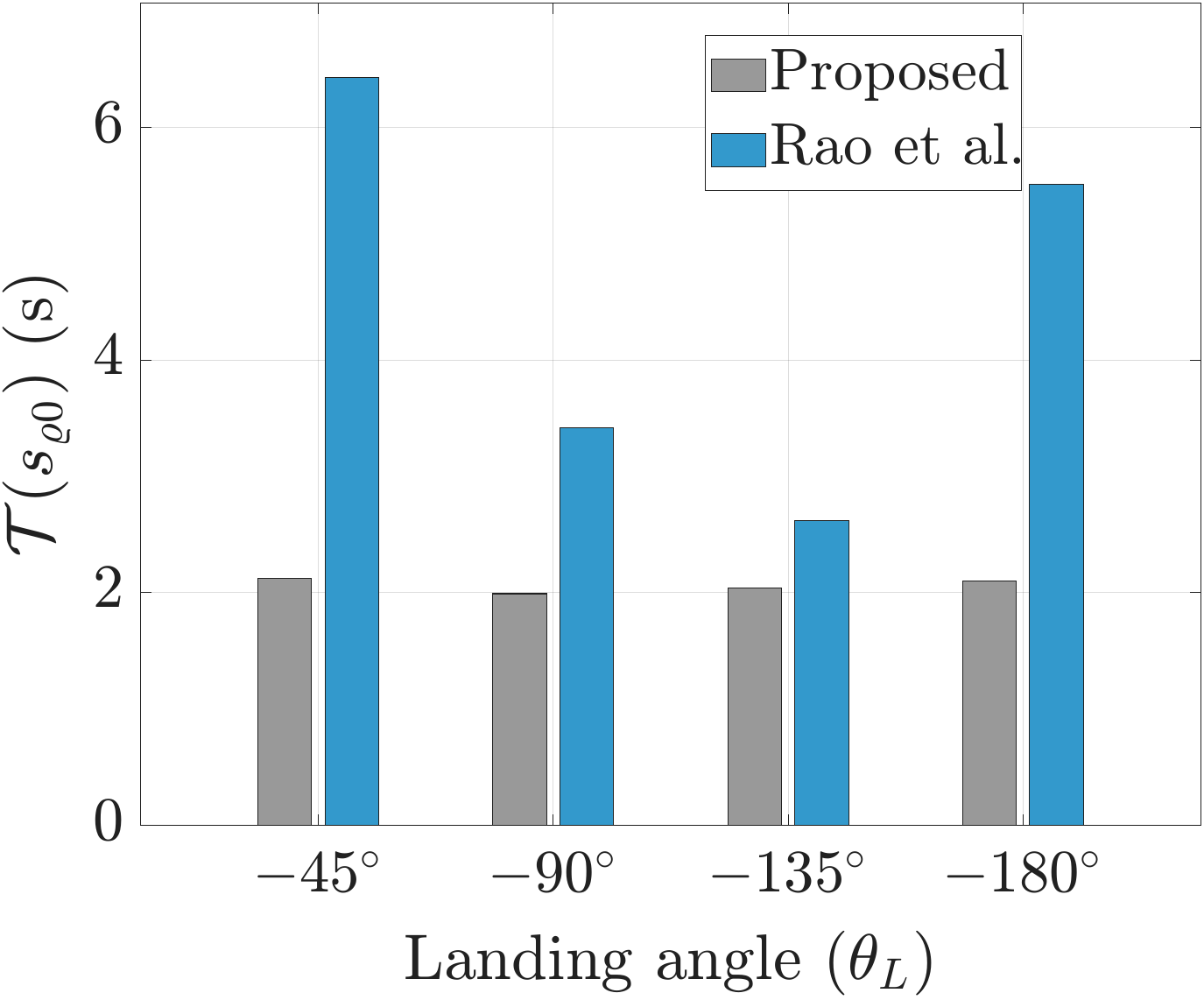}
		\caption{Convergence time to $\mathcal{S}_{2}$.}
		\label{fig:Metric_T_Slide_S2}
	\end{subfigure}%
	\begin{subfigure}{0.33\linewidth}
		\centering
		\includegraphics[width=\linewidth]{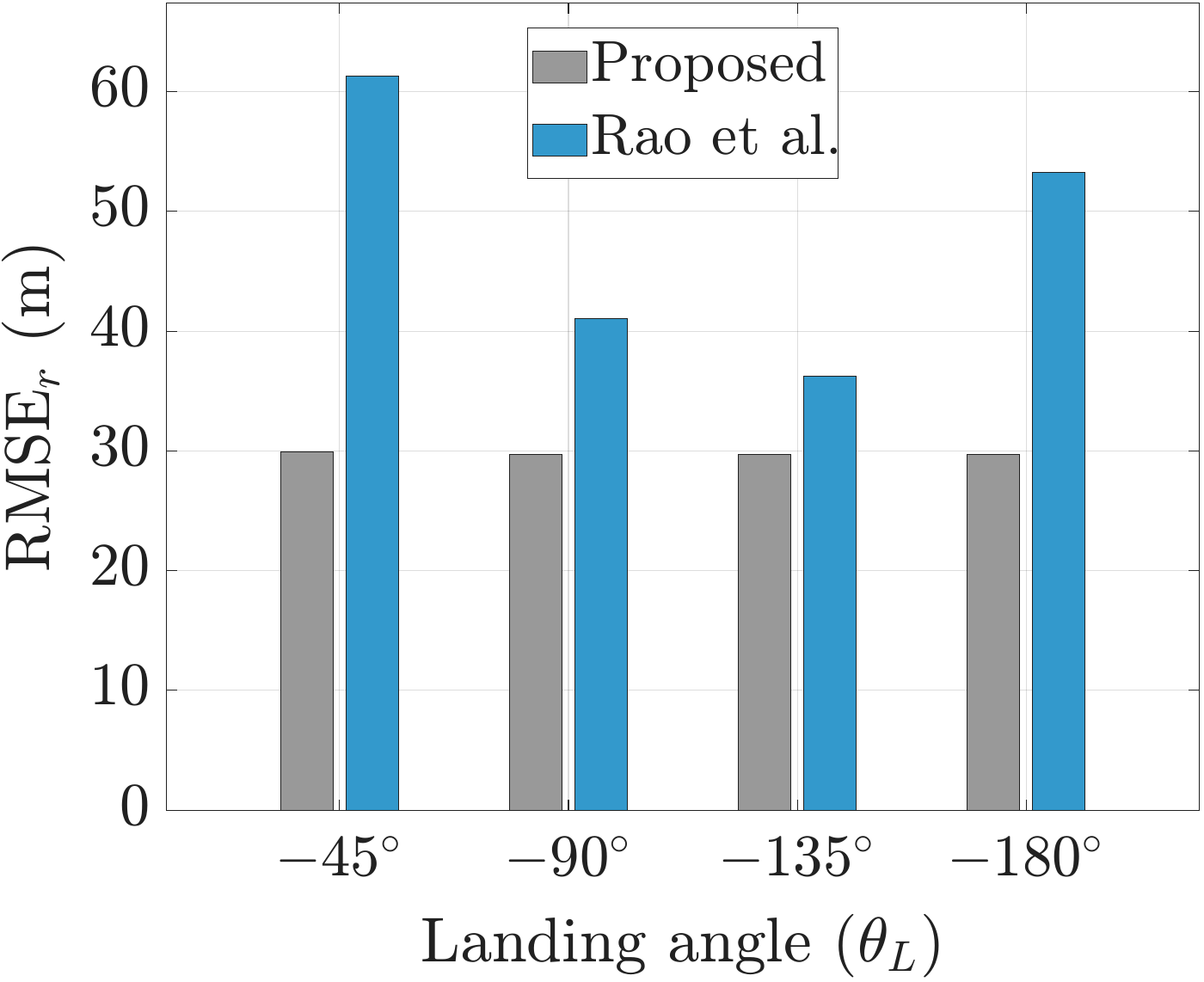}
		\caption{Root mean square error for range.}
		\label{fig:Metric_RMSE_r}
	\end{subfigure}
	\begin{subfigure}{0.33\linewidth}
		\centering
		\includegraphics[width=\linewidth]{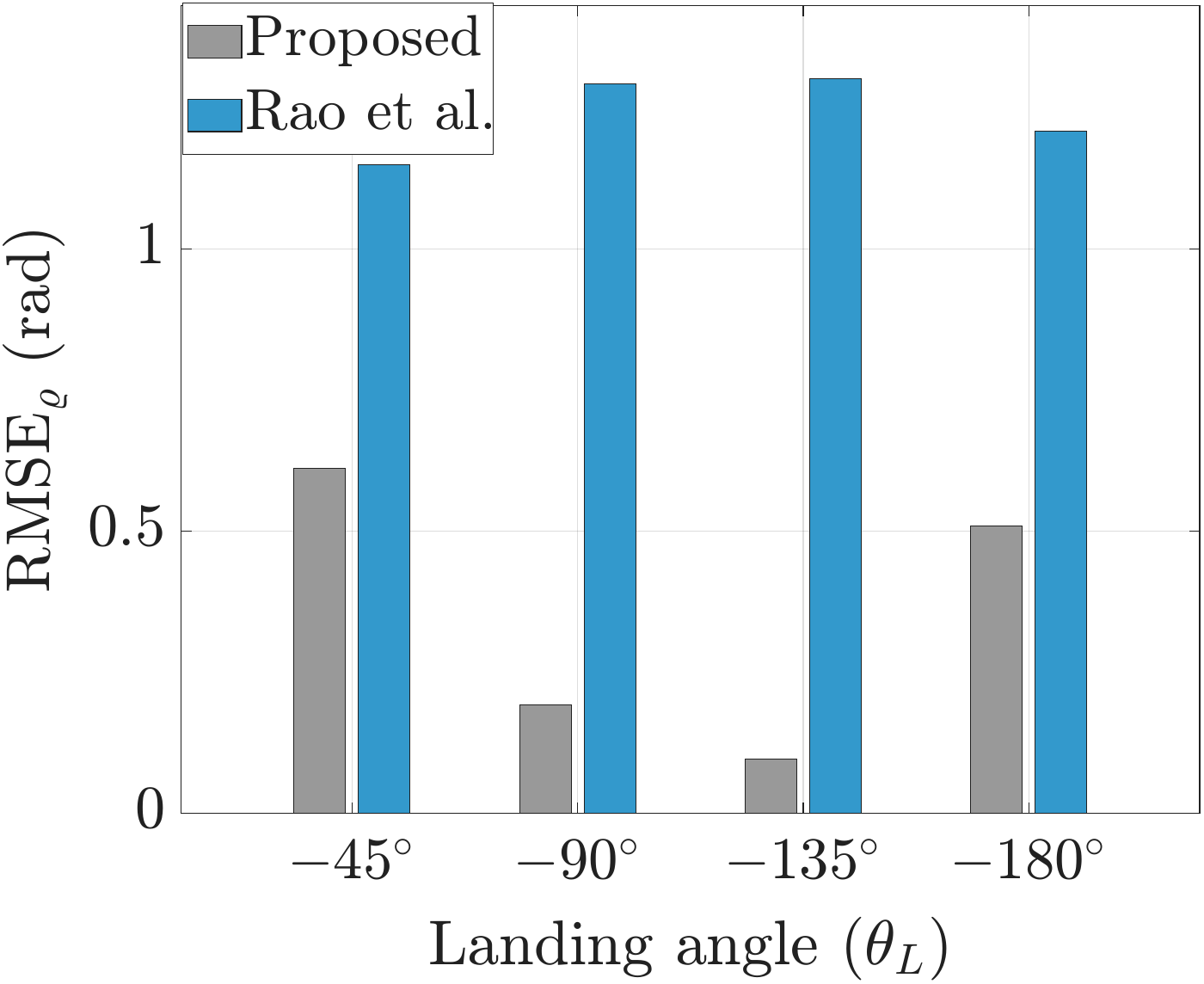}
		\caption{Root mean square error for LOS angle.}
		\label{fig:Metric_RMSE_th}
	\end{subfigure}%
	\begin{subfigure}{0.33\linewidth}
		\centering
		\includegraphics[width=\linewidth]{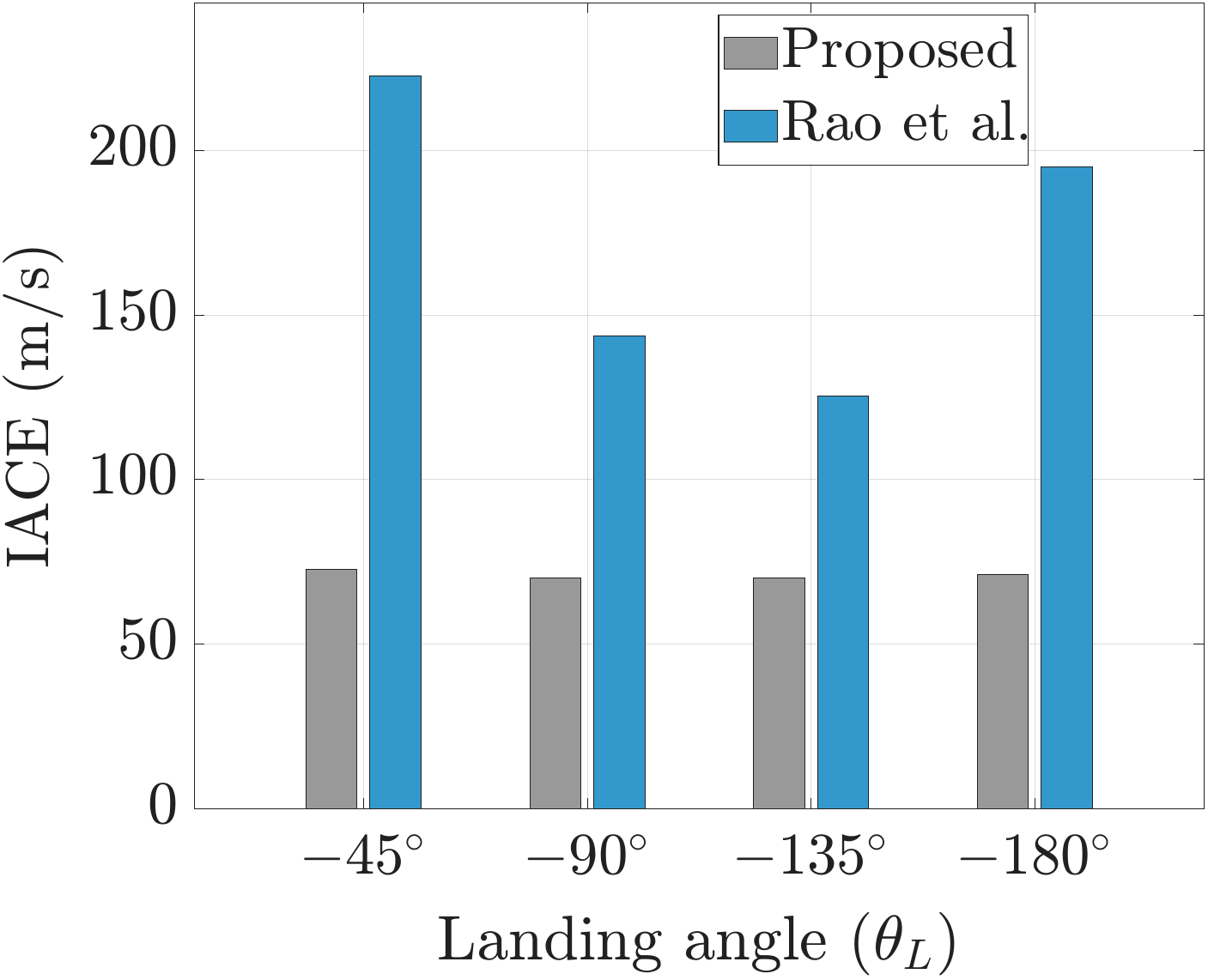}
		\caption{Control effort of equivalent agent.}
		\label{fig:Metric_IACE}
	\end{subfigure}%
	\begin{subfigure}{0.33\linewidth}
		\centering
		\includegraphics[width=\linewidth]{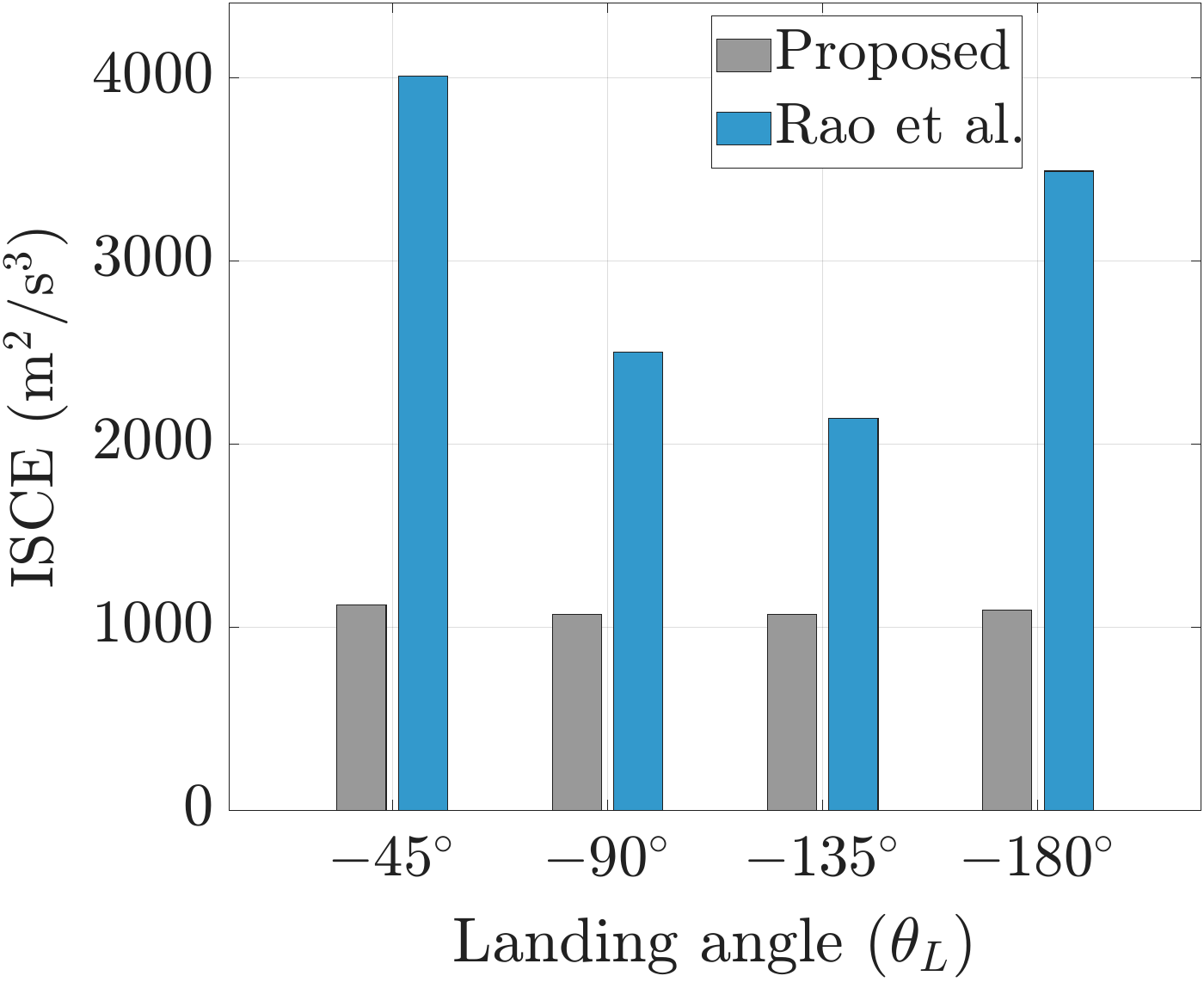}
		\caption{Control effort of equivalent agent.}
		\label{fig:Metric_ISCE}
	\end{subfigure}
	\begin{subfigure}{0.33\linewidth}
		\centering
		\includegraphics[width=\linewidth]{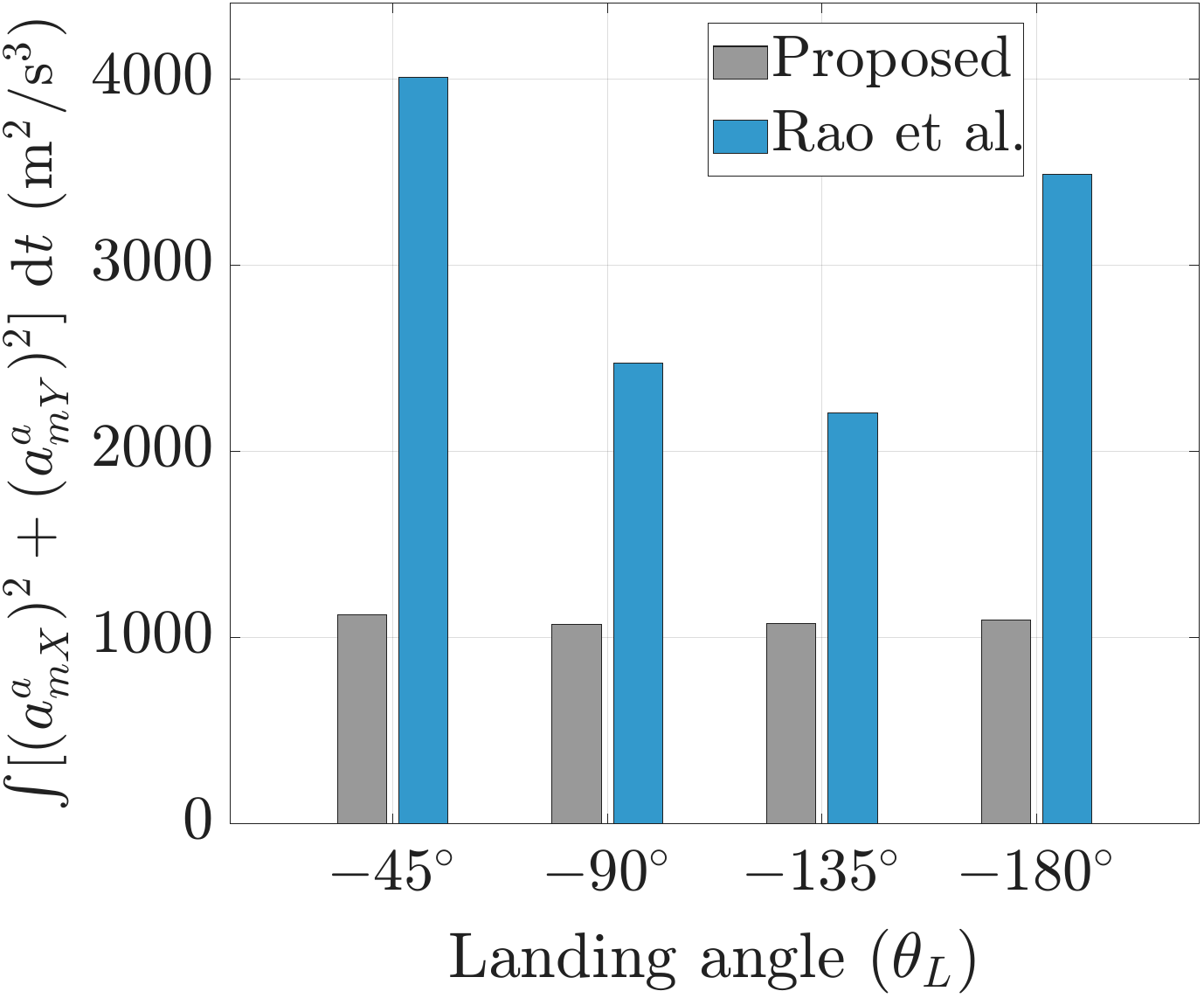}
		\caption{Control energy expended by UAV A.}
		\label{fig:Metric_E_U1}
	\end{subfigure}%
	\begin{subfigure}{0.33\linewidth}
		\centering
		\includegraphics[width=\linewidth]{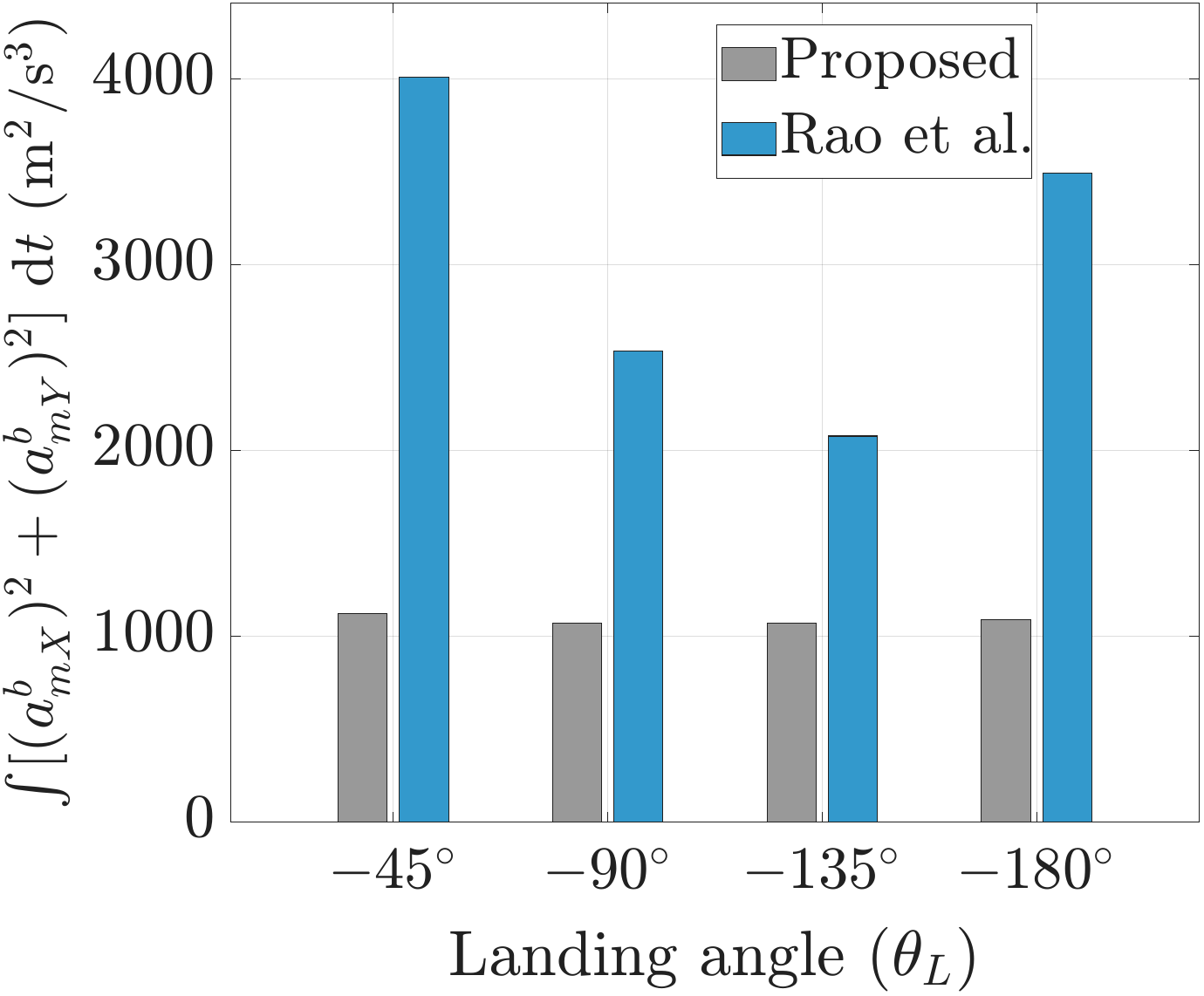}
		\caption{Control energy expended by UAV B.}
		\label{fig:Metric_E_U2}
	\end{subfigure}%
	\begin{subfigure}{0.33\linewidth}
		\centering
		\includegraphics[width=\linewidth]{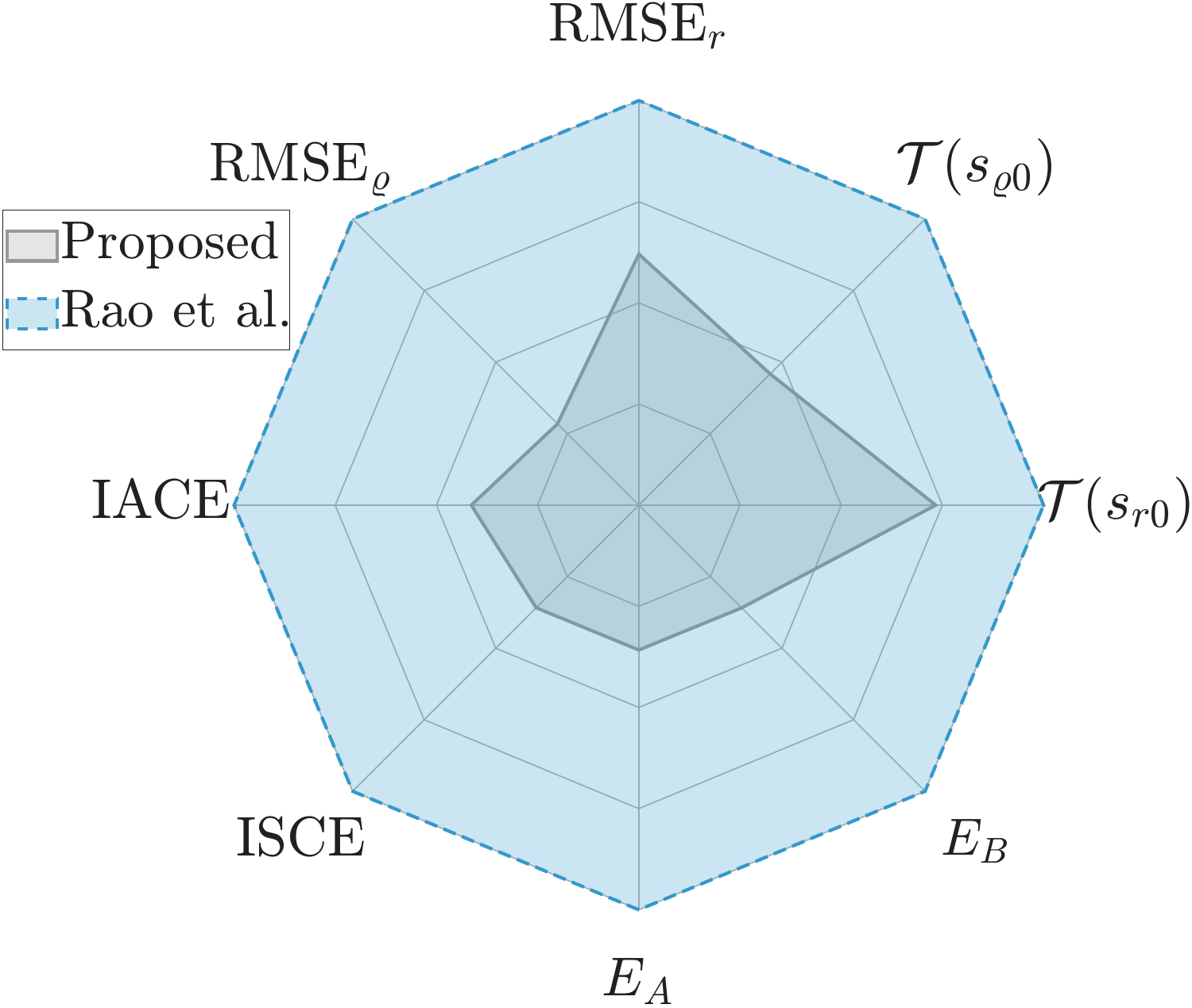}
		\caption{Normalized radar chart.}
		\label{fig:Metric_Spider}
	\end{subfigure}
	\caption{Comparative performance analysis: UAV delivering payload on a stationary platform.}
	\label{fig:stat_comp}
\end{figure}
\section{Conclusions}\label{ch:conclusions}
This paper presented a guidance framework for cooperative transportation of a rigid payload by a two-UAV system to stationary and maneuvering landing platforms. By exploiting the rigid-link geometry, the translational motion of the coupled UAV--payload system was represented through a virtual equivalent agent, allowing the payload-delivery problem to be formulated as a relative engagement problem with respect to the landing platform. This formulation provided a guidance-level description of the cooperative transportation task in terms of relative range and LOS dynamics, while retaining a systematic mapping to the individual UAVs. For a stationary platform, the payload can be delivered with an arbitrary prescribed landing angle by regulating the terminal LOS geometry. On the other hand, for a maneuvering platform, successful delivery with zero relative velocity requires velocity and heading synchronization between the equivalent agent and the platform, which consequently restricts the terminal landing angle to zero. Based on the equivalent-agent engagement formulation, fixed-time sliding-mode guidance laws were developed for the relative-range and LOS dynamics. For stationary platforms, the proposed strategy regulates the range and LOS angle variables toward the desired terminal configuration, while for maneuvering platforms, bounded platform acceleration is incorporated into the guidance design. A link-orientation controller and an acceleration-allocation scheme were subsequently developed to realize the virtual guidance commands at the individual UAVs while maintaining the prescribed rigid-link configuration. Numerical simulations demonstrated successful payload delivery for multiple prescribed landing angles on stationary platforms and under different maneuvering conditions of the landing platform. Comparative simulations further showed improved convergence and tracking performance, together with reduced integrated acceleration effort relative to the benchmark strategy considered in this work. A processor-in-the-loop implementation on a Raspberry Pi demonstrated that the proposed guidance architecture can be executed within the selected sampling interval, supporting its computational feasibility for real-time operation. Future work will focus on extending the proposed guidance framework to three-dimensional settings, directly incorporating actuator constraints into the guidance design, and leveraging heterogeneous UAV capabilities.
\bibliographystyle{aiaa}
\bibliography{references}
\end{document}